\documentclass[]{article}
\usepackage{sectsty}
\usepackage{xcolor}
\usepackage{soul}
\usepackage{cancel}
\usepackage{amsmath}
\usepackage{amsthm}
\usepackage{amssymb}
\usepackage{units}
\PassOptionsToPackage{normalem}{ulem}
\usepackage{ulem}
\usepackage[margin=1.2in]{geometry}
\usepackage[unicode=true,pdfusetitle,
bookmarks=true,bookmarksnumbered=true,bookmarksopen=false,
breaklinks=false,pdfborder={0 0 0},pdfborderstyle={},backref=false,colorlinks=true,psdextra]
{hyperref}
\usepackage{algpseudocode}
\usepackage{algorithm}
\usepackage{bm}
\usepackage[all,cmtip]{xy}
\usepackage{natbib}
\usepackage{graphicx}
\usepackage{wrapfig}

\newtheorem{thm}{Theorem}

\newtheorem{lemma}[thm]{Lemma}
\newtheorem{prop}[thm]{Proposition}

\newtheorem{conj}[thm]{Conjecture}
\newtheorem{cor}[thm]{Corollary}

\newcommand*{\mat}[1]{{\footnotesize \left(\begin{matrix} #1\end{matrix}\right)}}

\algnewcommand\algorithmicinput{\textbf{Input:}}
\algnewcommand\Input{\item[\algorithmicinput]}
\algnewcommand\algorithmicoutput{\textbf{Output:}}
\algnewcommand\Output{\item[\algorithmicoutput]}
\algnewcommand\Continue{\textbf{Continue}}

\newcommand*{\plog}{\partial \log}

\newcommand*{\dbeta}[1]{\frac{d {#1}}{d\beta}}
\newcommand*{\pbeta}[1]{\frac{\partial {#1}}{\partial\beta}}
\newcommand*{\bb}[1]{\mathbb{#1}}
\newcommand*{\norm}[1]{\left\| #1 \right\|}
\DeclareMathOperator*{\supp}{supp}
\DeclareMathOperator*{\eig}{eig}

\ifdefined\colorcalcuations
\newcommand*{\y}{{\color{violet} y}}
\newcommand*{\yp}{{\color{blue} y'}}
\newcommand*{\ypp}{{\color{magenta} y''}}
\newcommand*{\xhat}{{\color{cyan} \hat{x}}}
\newcommand*{\xhatp}{{\color{red} \hat{x}'}}
\newcommand*{\xhatpp}{{\color{teal} \hat{x}''}}
\newcommand*{\x}{{\color{purple} x}}
\newcommand*{\xp}{{\color{olive} x'}}
\newcommand*{\xpp}{{\color{gray} x''}}
\else
\newcommand*{\y}{y}
\newcommand*{\yp}{y'}
\newcommand*{\ypp}{y''}
\newcommand*{\xhat}{\hat{x}}
\newcommand*{\xhatp}{\hat{x}'}
\newcommand*{\xhatpp}{\hat{x}''}
\newcommand*{\x}{x}
\newcommand*{\xp}{x'}
\newcommand*{\xpp}{x''}
\fi

\newcommand*{\compilefigs}{}

\title{The Information Bottleneck's Ordinary Differential Equation: First-Order Root-Tracking for the IB}
\author{Shlomi Agmon$^1$\thanks{This work was partially supported by the ISF under grant 1641/21.}}

\date{%
	$^1$ Formerly with the School of Computer Science and Engineering, \\
	The Hebrew University of Jerusalem,
	Jerusalem,
	Israel\\
	Email: shlomi.agmon@mail.huji.ac.il \\
	[2ex]%
	\today
}

\begin{document}

\maketitle

\begin{abstract}
	The Information Bottleneck (IB) is a method of lossy compression of relevant information. Its rate-distortion (RD) curve describes the fundamental tradeoff between input compression and the preservation of relevant information embedded in the input. However, it conceals the underlying dynamics of optimal input encodings. We argue that these typically follow a piecewise smooth trajectory when input information is being compressed, as recently shown in RD. These smooth dynamics are interrupted when an optimal encoding changes qualitatively, at a \textit{bifurcation}. By leveraging the IB's intimate relations with RD, we provide substantial insights into its solution structure, highlighting caveats in its finite-dimensional treatments. Sub-optimal solutions are seen to collide or exchange optimality at its bifurcations.
	
	Despite the acceptance of the IB and its applications, there are surprisingly few techniques to solve it numerically, even for finite problems whose distribution is known. 
	We derive anew the IB's first-order Ordinary Differential Equation, which describes the dynamics underlying its optimal tradeoff curve. 
	To exploit these dynamics, we not only detect IB bifurcations but also identify their type in order to handle them accordingly. 
	Rather than approaching the IB's optimal curve from sub-optimal directions, the latter allows us to follow a solution's trajectory along the optimal curve under mild assumptions. 
	We thereby translate an understanding of IB bifurcations into a surprisingly accurate numerical algorithm.
\end{abstract}

\renewcommand*{\thefootnote}{\fnsymbol{footnote}}
\footnotetext[0]{ 
	The author is grateful to Or Ordentlich for helpful conversations and for his support, and to Noam and Dafna Agmon for their relentless support throughout this journey. The author thanks the late Naftali Tishby for insightful conversations and Etam Benger for his involvement during the early stages of this work. }
\renewcommand*{\thefootnote}{\arabic{footnote}}
\setcounter{footnote}{0}


\section{Introduction}
\label{sec:introduction}

The Information Bottleneck (IB) describes the fundamental tradeoff between the compression of information on an \textit{input} $X$ to the preservation of relevant information on a hidden \textit{reference} variable $Y$. 
Formally, let $X$ and $Y$ be random variables defined respectively on finite \textit{source} and \textit{label alphabets} $\mathcal{X}$ and $\mathcal{Y}$, and let $p_{Y|X}(\y|\x)p_X(\x)$ be their joint probability distribution\footnote{ Without loss of generality, we may assume that $p(\x) > 0$ for every $x\in \mathcal{X}$, and so $p_{Y|X}$ is well-defined.}, or $p(\y|\x)p(\x)$ for short. 
One seeks \citep{tishby1999} to maximize the information $I(Y; \hat{X})$ over all Markov chains $Y \longleftrightarrow X \longleftrightarrow \hat{X}$, subject to a constraint on the mutual information $I(X; \hat{X}) := \bb{E}_{p(\xhat|\x) p(\x)} \log \frac{p(\xhat|\x)}{p(\xhat)}$,
\begin{equation}			\label{eq:IB-curve-def}
	I_Y(I_X) := \max_{p(\xhat|\x)} \left\{ I(Y; \hat{X}): \; I(X; \hat{X}) \leq I_X \right\} \;.
\end{equation}
The latter maximization is over conditional probability distributions or \textit{encoders} $p(\xhat|\x)$.
The graph of $I_Y(I_X)$ is the \textit{IB curve}. 
We write $T := |\hat{\cal{X}}|$, for a codebook or \textit{representation alphabet} $\hat{\cal{X}}$. 
An encoder $p(\xhat|\x)$ which achieves the maximum in \eqref{eq:IB-curve-def} is \textit{IB optimal} or simply \textit{optimal}. 

Written in a Lagrangian\footnote{ Normalization constraints are omitted for clarity.} formulation $\mathcal{L} := I(X; \hat{X}) - \beta \; I(Y; \hat{X})$ with $\beta > 0$, \citep{tishby1999} showed that a necessary condition for extrema in \eqref{eq:IB-curve-def} is that the \textit{IB Equations} hold. Namely,
\begin{align}
	p(\xhat|\x)	&= \frac{p(\xhat)}{Z(\x,\beta)} \exp{\left\{ -\beta \; D_{KL}\big[p(\y|\x) || p(\y|\xhat)\big] \right\}} \;,		\label{eq:IB-eq-encoder} \\
	p(\y|\xhat) 	&= \sum_{x} p(\y|\x) p(\x|\xhat) \;, \quad \text{and} \label{eq:IB-eq-decoder} \\
	p(\xhat)		&= \sum_{x} p(\xhat|\x) p(\x) \;.			\label{eq:IB-eq-marginal}
\end{align}
In these, $Z(\x, \beta) := \sum_{\hat{x}} p(\xhat) \exp{\left\{ -\beta D_{KL}\left[p(\y|\x) || p(\y|\xhat)\right] \right\}}$ is the \textit{partition function}, $p(\x|\xhat)$ in \eqref{eq:IB-eq-decoder} is defined by the Bayes rule $\nicefrac{p(\xhat|\x) p(\x)}{p(\xhat)}$, and $D_{KL}$ is the Kullback-Leibler divergence, $D_{KL}[p||q] := \sum_i p(i) \log \nicefrac{p(i)}{q(i)}$. 
The IB Equations \eqref{eq:IB-eq-encoder}-\eqref{eq:IB-eq-marginal} are a necessary condition for an extremum of $\mathcal{L}$ also when it is considered as a functional in three independent families of normalized distributions $\{p(\xhat|\x)\}$, $\{p(\y|\xhat)\}$ and $\{p(\xhat)\}$, \cite[Section 3.3]{tishby1999}, rather than in $\{p(\xhat|\x)\}$ alone.
While satisfying them is necessary to achieve the curve \eqref{eq:IB-curve-def}, it is not sufficient. 
Indeed, Equations \eqref{eq:IB-eq-encoder}-\eqref{eq:IB-eq-marginal} have solutions that do not achieve curve \eqref{eq:IB-curve-def}, and so are \textit{sub-optimal}. 
This results in sub-optimal IB curves, which intersect or \textit{bifurcate} as the multiplier $\beta$ varies (see Section 3.4 there). 

Iterating over the IB Equations \eqref{eq:IB-eq-encoder}-\eqref{eq:IB-eq-marginal} is essentially Blahut-Arimoto's algorithm variant for the IB (BA-IB) due to \citep{tishby1999}, brought here as Algorithm \ref{algo:BA-IB}. 
While the minimization problem \eqref{eq:IB-curve-def} can be solved exactly in special cases, \citep[Section IV]{witsenhausen1975conditional}, exact solutions of an arbitrary finite IB problem whose distribution is known are usually obtained nowadays using BA-IB. 
cf., \cite[Section 3]{zaidi2020information} for a survey on other computation approaches. 
We write $BA_\beta$ for a single iteration of the BA-IB Algorithm \ref{algo:BA-IB}.
Since $BA_\beta$ encodes an iteration over the IB Equations \eqref{eq:IB-eq-encoder}-\eqref{eq:IB-eq-marginal}, then an encoder $p(\xhat|\x)$ is its fixed point, $BA_\beta\left[p(\xhat|\x)\right] = p(\xhat|\x)$, if and only if it satisfies the IB Equations.
Or equivalently, if $p(\xhat|\x)$ is a root of the \textit{IB operator}
\begin{equation}			\label{eq:IB-operator-def}
	F := Id - BA_\beta \;,
\end{equation}
in a manner similar to \cite{agmon2021critical}. We shall then call it an \textit{IB root}. 
\citeauthor{agmon2021critical} used a similar formulation of rate-distortion (RD) and its relations in \citep{bachrach2003} to the IB, to show that BA-IB suffers from \textit{critical slowing down} near \textit{critical points}, where the marginal $p(\xhat)$ of a representor $\xhat$ in an optimal encoder vanishes gradually. 
That is, the number of BA-IB iterations required till convergence increases dramatically as one approaches such points.

Formulating fixed points of an iterative algorithm as operator roots can also be leveraged for computational purposes in a constrained-optimization problem, as noted recently by \cite{agmon2022RTRD} for RD. 
Indeed, let $F(\cdot, \beta)$ be a differentiable operator on $\bb{R}^n$ for some $n > 0$, $F: \bb{R}^{n} \times \bb{R} \to \bb{R}^{n}$, where $\beta$ is a (real) constraint parameter. 
Suppose now that $(\bm{x}, \beta)$ is a root of $F$, 
\begin{equation}		\label{eq:root-of-functional-eq-implicit}
	F(\bm{x}, \beta) = \bm{0} \;,
\end{equation}
such that $\bm{x} = \bm{x}(\beta)$ is a differentiable function of $\beta$.
Write $D_{\bm{x}} F := \big( \tfrac{\partial}{\partial x_j}F_i \big)_{i, j}$ for its Jacobian matrix, and $D_\beta F := \big( \tfrac{\partial}{\partial \beta} F_i \big)_i$ for its vector of partial derivatives with respect to $\beta$. The point $(\bm{x}, \beta)$ of evaluation is omitted whenever understood. 
As is often discussed along with the Implicit Function Theorem, e.g., \citep{de2014implicit}, applying the multivariate chain rule to $F\left(\bm{x}(\beta), \beta\right)$ in \eqref{eq:root-of-functional-eq-implicit} yields an implicit ordinary differential equation (ODE)
\begin{equation}			\label{eq:implicit-beta-ODE}
	D_{\bm{x}} F \; \tfrac{d\bm{x}}{d\beta} = -D_\beta F \;,
\end{equation}
for the roots of $F$. 
Plugging in explicit expressions for the first-order derivative tensors $D_{\bm{x}} F$ and $D_\beta F$, one can specialize \eqref{eq:implicit-beta-ODE} to a particular setting, which allows one to compute the \textit{implicit derivatives} $\tfrac{d\bm{x}}{d\beta}$ numerically. 
While \cite{agmon2022RTRD} discovered the \textit{RD ODE} this way, they showed that \eqref{eq:implicit-beta-ODE} can be generalized to arbitrary order under suitable differentiability assumptions. 
Namely, they showed that the derivatives $\tfrac{d^l\bm{x}}{d\beta^l}$ implied by $F = \bm{0}$ \eqref{eq:root-of-functional-eq-implicit} can be computed via a recursive formula, for an arbitrary-order $l > 0$. 
By specializing this with the higher derivatives of Blahut's algorithm \citep{blahut1972}, they obtained a family of numerical algorithms for following the path of an optimal RD root (Part I there). 

In this work, we specialize the implicit ODE \eqref{eq:implicit-beta-ODE} to the IB. 
Namely, we plug into \eqref{eq:implicit-beta-ODE} the first-order derivatives of the IB operator $Id - BA_\beta$ \eqref{eq:IB-operator-def} to obtain the \textit{IB ODE}, and then use it to reconstruct the path of an optimal IB root, in a manner similar to \cite{agmon2022RTRD}. 
This is not to be confused with the gradient flow (of arbitrary encoders) towards an optimal root at a fixed $\beta$ value, described at \cite[Equation (6)]{gedeon2012mathematical} by an ODE, which is a different optimization approach. 
In contrast, the implicit Equation \eqref{eq:implicit-beta-ODE} describes how a root evolves \textit{with} $\beta$. 
So, in principle, one may compute an optimal IB root once and then follow its evolution along the IB curve \eqref{eq:IB-curve-def}. 
While the discovery of the IB ODE is due to \cite{agmon2022thesis}, we derive it here anew in a form that is better suited for computational (and other) purposes, especially when there are fewer possible labels $\mathcal{Y}$ than input symbols $\mathcal{X}$, as often is the case. 
To that end, we consider several natural choices of a coordinate system for the IB in Section \ref{sec:coords-exchange-for-the-IB} and compare their properties. 
This allows us to make an apt choice for the ODE's variable $\bm{x}$ in \eqref{eq:implicit-beta-ODE}. 
In Section \ref{sec:IB-ODE}, we present the IB ODE in these coordinates (Theorem \ref{thm:IB-ODE}).
This enables one to numerically compute the first-order implicit derivatives at an IB root, if it can be written as a differentiable function in $\beta$. 
So long that an optimal root remains differentiable, a simple way to reconstruct its trajectory is by taking small steps at a direction determined by the IB ODE. This is \textit{Euler's method} for the IB. 
The error accumulated by Euler's method from the true solution path is roughly proportional to the step size, when small enough. 
For comparison, reverse deterministic annealing \citep{rose1990deterministic} with BA-IB is nowadays common for computing IB roots. 
The dependence of its error on the step size is roughly the same as in Euler's method. 
This is discussed in Section \ref{sec:euler-method}, where we combine Euler's method with BA-IB to obtain a modified numerical method whose error decreases at a faster rate than either of the above.

However, the differentiability of optimal IB roots breaks where the solution changes qualitatively. 
Such a point is often called a phase transition in the IB literature, or a \textit{bifurcation} --- namely, a point where there is a change in the problem's number of solutions. 
e.g., \cite[Section 2.3]{kuznetsov2004elements} for basic definitions. 
As noted already by \citeauthor{tishby1999}, their existence in the IB stems from restricting the cardinality of the representation alphabet $\hat{\mathcal{X}}$. 
Indeed, the gap between achieving the IB curve \eqref{eq:IB-curve-def} to merely satisfying the fixed-point equations \eqref{eq:IB-eq-encoder}-\eqref{eq:IB-eq-marginal} lies in understanding the solution structure of the IB operator \eqref{eq:IB-operator-def}, or equivalently its bifurcations. 
While IB bifurcations were analyzed in several works, including \citep{gedeon2012mathematical, zaslavsky2019thesis, ngampruetikorn2021perturbation} and others, little is known about the practical value of understanding them. \citep{wu2020learnability, wu2020phase} showed that they correspond to the onset of learning new classes, while \citep{agmon2021critical} showed that they inflict a hefty computational cost to BA-IB. 
Following \cite{agmon2022RTRD}, this work demonstrates that understanding bifurcations can be translated to a new numerical algorithm to solve the IB. 
To that end, merely detecting a bifurcation along a root's path does not suffice. 
But rather, it is also necessary to identify its type, as this allows one to handle the bifurcation accordingly. 
One can then continue following the path dictated by the IB ODE.

Almost all of the literature on IB bifurcations is based on a perturbative approach, in a manner similar to \cite[Section IV.C]{rose1994mapping}. 
That is, suppose that the first variation\footnote{ For finite IB problems, condition \eqref{eq:first-variational-deriv-of-Lagrangian} boils down to requiring that the gradient of $\mathcal{L}$ vanishes, while condition \eqref{eq:second-variational-deriv-of-Lagrangian} is equivalent to requiring that its Hessian matrix has a non-trivial kernel, as both are conditions on directional derivatives. e.g., \cite{giaquinta2004calculus}. } 
\begin{equation}			\label{eq:first-variational-deriv-of-Lagrangian}
	\frac{\partial}{\partial \epsilon} \mathcal{L}\left[ p(\xhat|\x) + \epsilon \Delta p(\xhat|\x); \beta \right]\Big\rvert_{\epsilon=0} 
\end{equation}
of the IB Lagrangian $\mathcal{L}$ vanishes, for every perturbation $\Delta p(\xhat|\x)$. 
This condition is necessary for extremality and implies the IB Equations \eqref{eq:IB-eq-encoder}-\eqref{eq:IB-eq-marginal}, \citep{tishby1999}. 
Then, $\big(p(\xhat|\x), \beta \big)$ is said to be a \textit{phase transition} only if there exists a particular direction $\Delta q(\xhat|\x)$ at which $p(\xhat|\x)$ can be perturbed without affecting the Lagrangian's value to second order,
\begin{equation}			\label{eq:second-variational-deriv-of-Lagrangian}
	\frac{\partial^2}{\partial \epsilon^2} \mathcal{L}\left[ p(\xhat|\x) + \epsilon \Delta q(\xhat|\x); \beta \right]\Big\rvert_{\epsilon=0} = 0 \;.
\end{equation}
\cite{gedeon2012mathematical, wu2020learnability, wu2020phase} and \cite{ngampruetikorn2021perturbation} take such an approach. 
\cite{zaslavsky2019thesis} similarly analyzes one type of IB bifurcation. 

While a perturbative approach is common in analyzing phase transitions, it has several shortcomings when applied to the IB, as noted by \cite{agmon2022thesis}. 
First, the IB's Lagrangian $\mathcal{L}$ is constant on a linear manifold of encoders $p(\xhat|\x)$, \cite[Section 3.1]{gedeon2012mathematical}, and so condition \eqref{eq:second-variational-deriv-of-Lagrangian} leads to false-detections. 
While this was considered there and in its sequel \citep{parker2022symmetry_breaking} by giving subtle conditions on the nullity of the second variation in \eqref{eq:second-variational-deriv-of-Lagrangian}, in practice it is difficult to tell whether a particular direction $\Delta q(\xhat|\x)$ is in the kernel due to a bifurcation or due to other reasons, as they note. 
Second, note that a finite IB problem can be written as an \textit{infinite} RD problem, \citep{harremoes2007information}. 
As discussed in Section \ref{sec:IB-bifurcations}, representing an IB root by a finite-dimensional vector leads to inherent subtleties in its computation. 
Among other things, these may well result in a bifurcation \textit{not} being detectable under certain circumstances (Section \ref{sub:discontinuous-IB-bifs}). 
To our understanding, many of the difficulties that hindered the understanding of IB bifurcations throughout the years are, in fact, artifacts of finite dimensionality. 
Third, conditions \eqref{eq:first-variational-deriv-of-Lagrangian}-\eqref{eq:second-variational-deriv-of-Lagrangian} do not suffice to reveal the type of the bifurcation, information which is necessary for handling it when following a root's path. 
While \cite[Section 2.9]{parker2022symmetry_breaking} give conditions for identifying the type, these partially agree with our findings and do not suggest a straightforward way for handling a bifurcation.

Rather than imposing conditions on the scalar functional $\mathcal{L}$, our approach to IB bifurcations follows that of \citep{agmon2022RTRD} for RD. 
That is, we rely on the fact that the IB's local extrema are fixed points of an iterative algorithm, and so they also satisfy a vector equation $F = \bm{0}$ \eqref{eq:root-of-functional-eq-implicit}. 
We shall now consider a toy problem to motivate our approach. 
``\textit{Bifurcation Theory can be briefly described by the investigation of problem \eqref{eq:root-of-functional-eq-implicit} in a neighborhood of a root where $D_{\bm{x}} F$ is singular}'', \citep{kielhofer2011bifurcation}. 
Indeed, recall that if $D_{\bm{x}} F$ is non-singular at a root $(\bm{x}_0, \beta_0)$, then by the Implicit Function Theorem (IFT), there exists a function $\bm{x}(\beta)$ through the root, $\bm{x}(\beta_0) = \bm{x}_0$, which satisfies $F\big(\bm{x}(\beta), \beta\big) = \bm{0}$ \eqref{eq:root-of-functional-eq-implicit} at the vicinity of $\beta_0$. 
The function $\bm{x}(\beta)$ is then not only unique at some neighborhood of $(\bm{x}_0, \beta_0)$, but further, $\bm{x}(\beta)$ inherits the differentiability properties of $F$, \cite[I.1.7]{kielhofer2011bifurcation}. 
In particular, if the operator $F$ is real-analytic in its variables --- as with the IB operator \eqref{eq:IB-operator-def} --- then so is its root $\bm{x}(\beta)$. 
While a bifurcation can occur only if $D_{\bm{x}} F$ is singular, singularity is not sufficient for a bifurcation to occur. 
For example, the roots of the operator 
\begin{equation}						\label{eq:toy-example-operator-def}
	F(x,y; \beta) := (x - \beta, 0) 
\end{equation}
on $\bb{R}^2$ consist of the vertical line $x = \beta$, $\{(\beta, y): \; y \in \bb{R}\}$, for every $\beta \in \bb{R}$. 
For a fixed $y$, each such root is real-analytic in $\beta$. 
However, one cannot deduce this directly from the IFT, as the Jacobian $\mat{1 & 0 \\ 0 & 0}$ of $F$ \eqref{eq:toy-example-operator-def} is always singular. 
Note, however, that in this particular example, the $x$ coordinate alone suffices to describe the problem's dynamics, and so its $y$ coordinate is redundant. 
One can ignore the $y$ coordinate by considering the ``reduction'' $\tilde{F}(x; \beta) := x - \beta$ of $F$ to $\bb{R}^1$. 
Further, discarding $y$ also removes or \textit{mods-out} the direction $\mat{0\\1}$ from $\ker D_{\bm{x}} F$, which does not pertain to a bifurcation in this case. 
This results in the non-singular Jacobian matrix $\mat{1}$ of $\tilde{F}$, and so it is now possible to invoke the IFT on the reduced problem. 
The root guaranteed by the IFT can always be considered in $\bb{R}^2$ by putting back a redundant $y$ coordinate at some fixed value. 
\cite{agmon2022RTRD} used a similarly defined \textit{reduction} of finite RD problems to show that their dynamics are piecewise real-analytic under mild assumptions. 

The intuition behind our approach is similar to \cite[Section III]{harremoes2007information}, who observed that ``\textit{in the IB one can also get rid of irrelevant variables \underline{within} the model}''.
Nevertheless, the details differ. 
Mathematically, we consider\footnote{ This formulation can be made precise by using the \textit{tangent space} of a differentiable manifold, e.g., \cite[Section 3]{lee2013manifolds}. However, that shall not be necessary. } the \textit{quotient} $\nicefrac{V}{W}$ of a vector space $V$ by its subspace $W$. 
Elements of $V$ are identified in the quotient if they differ by an element of $W$: $v_1 \sim v_2 \Leftrightarrow v_1 - v_2 \in W$, for $v_1, v_2\in V$. 
This way, one ``mods-out'' $W$, collapsing it to a single point in the quotient vector space $\nicefrac{V}{W}$. 
The resulting problem is smaller and so easier to handle, whether for theoretical or practical purposes. 
This is how the one-dimensional vector space $\ker D_{\bm{x}} F$ in our toy example \eqref{eq:toy-example-operator-def} was reduced to the trivial $\ker D_{\bm{x}} \tilde{F} = \{\bm{0}\}$. 
However, one needs to understand the solution structure, for example, to ensure that the directions in $W$ are not due to a bifurcation. 
We note in passing that $\nicefrac{V}{W}$ has a simple geometric interpretation as the translations of $W$ in $V$, in a manner reminiscent of its better-known counterparts of quotient groups and rings. e.g., \cite[Section 10.2]{dummit2004abstract_alg}. 
To keep things simple, however, we shall not use quotients explicitly. 
Instead, the reader may simply consider the sequel as a removal of redundant coordinates. 
For, we shall only remove coordinates that the reader does not care about anyway, as in the above toy example. 

To achieve this approach, one needs to consider the IB in a coordinate system that permits a simple reduction as in \eqref{eq:toy-example-operator-def}, and to understand its solution structure. 
We achieve these in Section \ref{sec:IB-bifurcations} by exploiting two properties of the IB which are often overlooked. 
First, proceeding with the coordinates-exchange of Section \ref{sec:coords-exchange-for-the-IB}, the intimate relations \citep{harremoes2007information, bachrach2003} of the IB with RD suggest a ``minimally sufficient'' coordinates system for the IB, just as the $x$ axis is for problem \eqref{eq:toy-example-operator-def}. 
\textit{Reducing} an IB root to these coordinates is a natural extension of reduction in RD, \citep{agmon2022RTRD}. 
Reduction of IB roots facilitates a clean treatment of IB bifurcations. 
These are roughly divided into \textit{continuous} and \textit{discontinuous} bifurcations, in Subsections \ref{sub:continuous-IB-bifs} and \ref{sub:discontinuous-IB-bifs} respectively. 
While understanding continuous bifurcations is straightforward, the IB's relations with RD allow us to understand the discontinuous bifurcation examples of which we are aware as a \textit{support switching bifurcation} in RD, by leveraging \cite[Section 6]{agmon2022RTRD}. 
A second property is the analyticity of the IB operator \eqref{eq:IB-operator-def}, which stems from the analyticity of the IB Equations \eqref{eq:IB-eq-encoder}-\eqref{eq:IB-eq-marginal}. 
By building on the first property, analyticity leads us to argue that the Jacobian of the IB operator \eqref{eq:IB-operator-def} is generally non-singular (Conjecture \ref{conj:BA-IB-Jacob-in-decoder-coords-is-nonsingular-at-reduced-root}) when considered in reduced coordinates as above. 
As an immediate consequence, the dynamics underlying the IB curve \eqref{eq:IB-curve-def} are piecewise real-analytic in $\beta$, in a manner similar to RD. 
Indeed, the fact that there exist dynamics underlying the IB curve \eqref{eq:IB-curve-def} in the first place can arguably be attributed to analyticity; cf., the discussion following Conjecture \ref{conj:BA-IB-Jacob-in-decoder-coords-is-nonsingular-at-reduced-root}. 
Combining both properties sheds light on several subtle yet important practical caveats in solving the IB (Subsection \ref{sub:discontinuous-IB-bifs}) due to using finite-dimensional representations of its roots.  
These subtleties are compatible with our numerical experience. 
The results here suggest that, unlike RD, the IB is inherently infinite-dimensional, even for finite problems.

\begin{figure}[h!]
	\centering
	\vspace*{10pt}
	\ifdefined\compilefigs
	\includegraphics[width=.8\textwidth]{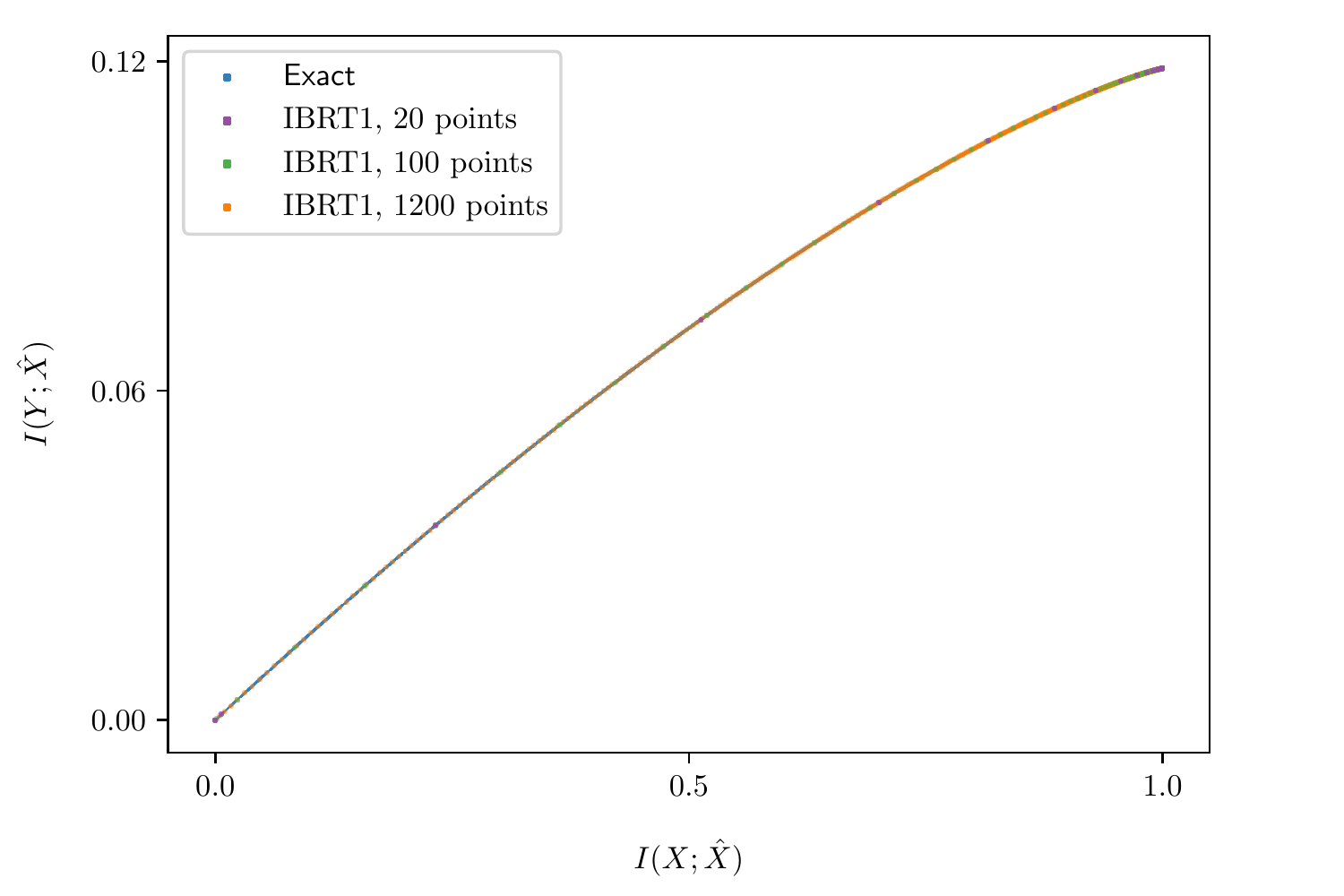}
	\else
	\includegraphics[width=.8\textwidth]{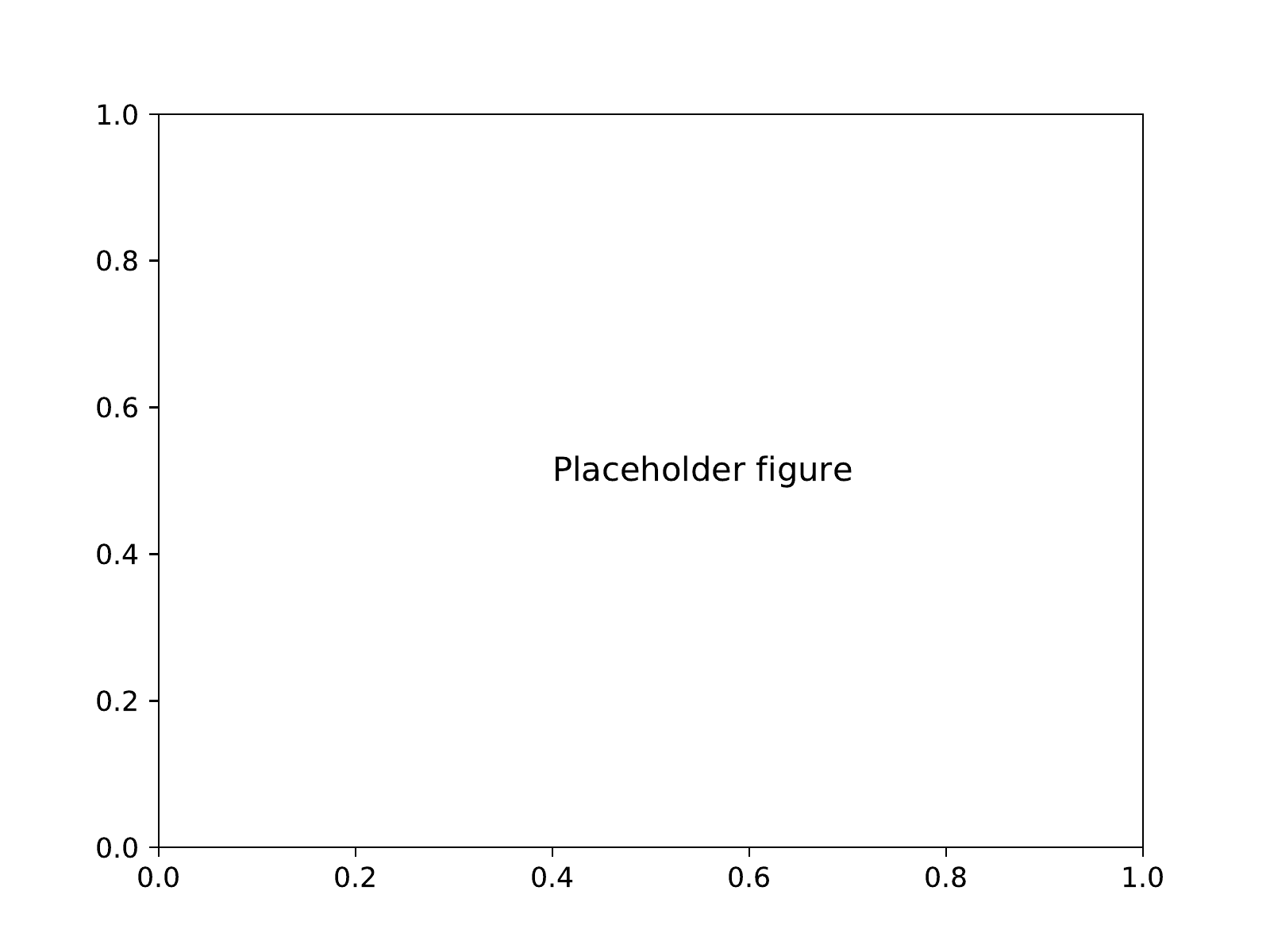}
	\fi
	\caption{
		\textbf{The approximate IB curves yielded by Algorithm \ref{algo:IBRT1}, using the IB ODE \eqref{eq:IB-beta-ODE-in-decoder-coords}}.
		Despite the algorithm's approximation errors (Section \ref{sub:IBRT1-numerical-results}), the approximate curves it yields are visually indistinguishable from the true IB curve \eqref{eq:IB-curve-def}, even on relatively few grid points. 
		The reasons for this are discussed below (Section \ref{sub:IBRT1-discussion}). 
		\newline 
		Our First-order Root-Tracking for the IB (IBRT1) Algorithm \ref{algo:IBRT1} was used to approximate the optimal IB roots of a binary symmetric channel with crossover probability $0.3$ and a uniform source, BSC$(0.3)$, for several grid densities. 
		The points in the information plane yielded from these approximations are plotted on top of the problem's exact solution (see Appendix \ref{sec:analytical-IB-sol-for-BSC-appendix}).
	}
	\label{fig:IBRT1-IB-curve-for-several-densitires}
\end{figure}

Finally, Section \ref{sec:IBRT1-algo} combines the modified Euler method of Section \ref{sec:euler-method} with the understanding of IB bifurcations in Section \ref{sec:IB-bifurcations}, to obtain Algorithm \ref{algo:IBRT1} (IBRT1) for following the path of an optimal IB root, in Subsection \ref{sub:IBRT1-algo-spec}. That is, First-order Root-Tracking for the IB. 
For simplicity, we focus mainly on continuous IB bifurcations, as these are the ones most often encountered in practice; 
cf., the comments in Subsection \ref{sub:IBRT1-discussion}. 
The resulting approximations in the information plane are surprisingly close to the true IB curve \eqref{eq:IB-curve-def}, even on relatively sparse grids (i.e., with large step sizes), as seen in Figure \ref{fig:IBRT1-IB-curve-for-several-densitires}. 
See Subsection \ref{sub:IBRT1-numerical-results} for the numerical results underlying the latter. 
The reasons for this are discussed in Subsection \ref{sub:IBRT1-discussion}, along with the algorithm's basic properties. 
Unlike BA-IB, which suffers from an increased computational cost near bifurcations, our Algorithm \ref{algo:IBRT1} suffers from a reduced accuracy there, in a manner similar to root-tracking for RD, \citep{agmon2022RTRD}. 

With that, we note that there are standard techniques in Bifurcation Theory for handling a non-trivial kernel of $D_{\bm{x}} F$ at a root. 
For example, the \textit{Lyapunov-Schmidt reduction} replaces the high-dimensional problem $F = \bm{0}$ \eqref{eq:root-of-functional-eq-implicit} on $\bb{R}^n$ by a smaller but equivalent problem $\Phi = \bm{0}$, where $\Phi(\cdot, \beta)$ maps vectors in the (right) kernel of $D_{\bm{x}} F$ to vectors in its left kernel. 
To achieve this, it separates the kernel- and non-kernel directions of the problem, essentially handling each at its turn. 
e.g., \cite[Theorem I.2.3]{kielhofer2011bifurcation} or \cite[Section 9.7]{teschl2020topics}. 
This technique is generic, as it does not rely on any particular property of the problem at hand. 
As such, it is considerably more involved than removing redundant coordinates\footnote{ Applied to our toy problem \eqref{eq:toy-example-operator-def} for instance, Lyapunov-Schmidt reduces $F = \bm{0}$ \eqref{eq:root-of-functional-eq-implicit} to choosing a continuously differentiable function $\Phi$ on the $y$-axis there, which is obtained by first solving for $x = \beta$ (see the proof of \cite[Theorem I.2.3]{kielhofer2011bifurcation} for details). However, since $y$ is redundant in this example, then solving for $\Phi$ can provide no useful information on the dynamics of its roots. }, which requires an understanding of the solution structure. 
In contrast, reduction in the IB is straightforward. 
For the purpose of following a root's path, carrying on with redundant kernel directions is burdensome, computationally expensive, and sensitive to approximation errors. 
\cite{parker2022symmetry_breaking} use a variant of the Lyapunov-Schmidt reduction to consider IB bifurcations due to symmetry breaking. 
While our findings are in agreement with theirs' for continuous IB bifurcations, they differ for discontinuous bifurcations (see Subsections \ref{sub:continuous-IB-bifs} and \ref{sub:discontinuous-IB-bifs}). 

\medskip 
\paragraph*{Notations.}
Vectors are written in boldface $\bm{x}$, scalars in a regular font $x$. 
A distribution $p$ pertaining to a particular Lagrange multiplier value $\beta$ (in Equations \eqref{eq:IB-eq-encoder}-\eqref{eq:IB-eq-marginal}) is denoted with a subscript, $p_\beta$. 
The \textit{probability simplex} on a set $S$ is denoted $\Delta[S]$ (see Section \ref{sub:IB-as-an-RD-problem-and-non-singularity-conj}). 
The \textit{support} of a probability distribution $p$ on $S$ is $\supp p := \{s\in S: \; p(s) \neq 0\}$. 
The \textit{source}, \textit{label} and \textit{representation} alphabets of an IB problem are denoted $\mathcal{X}, \mathcal{Y}$ and $\hat{\mathcal{X}}$, respectively; we write $T:= |\hat{\mathcal{X}}|$. 
$\delta$ denotes Dirac's delta function, $\delta_{i, j} = 1$ if $i = j$ and zero otherwise.

\begin{algorithm}
	\caption{Blahut-Arimoto for the Information Bottleneck (BA-IB), \cite{tishby1999}.}
	\label{algo:BA-IB}
	\begin{algorithmic}[1]
		\Function{BA-IB}{$p_0(\xhat|\x); p_{Y|X} \; p_X, \beta$}
		\Input
		\Statex An initial encoder $p_0(\xhat|\x)$, a problem definition $p(\y|\x)p(\x)$, and $\beta > 0$.
		\Output
		\Statex A fixed point $p(\xhat|\x)$ of the IB Equations.
		\State Initialize $i \gets 0$.
		\Repeat		\nonumber
		\State $p_{i}(\xhat) \leftarrow \sum_{\x} p_i(\xhat|\x)p(\x)$	\label{eq:IB-BA-cluster_marginal}
		\State $p_i(\x|\xhat) \leftarrow \nicefrac{p_i(\xhat|\x)p(\x)}{p_i(\xhat)}$	\label{eq:IB-BA-bayes-for-computing-inverse-enc}
		\State $p_{i}(\y|\xhat) \leftarrow \sum_{\x} p(\y|\x)p_i(\x|\xhat)$	\label{eq:IB-BA-decoder-eq}
		\State $Z_{i}(\x,\beta) \leftarrow \sum_{\xhat} p_{i}(\xhat) \exp{\big\{ -\beta \; D_{KL}\big[p(\y|\x) || p_{i}(\y|\xhat)\big] \big\}}$		\label{eq:IB-BA-partition-func}
		\State $p_{i+1}(\xhat|\x) \leftarrow \frac{p_{i}(\xhat)}{Z_{i}(\x,\beta)} \exp \big\{ -\beta \; D_{KL}\big[p(\y|\x) || p_{i}(\y|\xhat)\big] \big\}$		\label{eq:IB-BA-new-direct-enc}
		\State $i \gets i + 1$
		\Until{convergence.}
		\EndFunction
	\end{algorithmic}
\end{algorithm}

\newpage 
\section{Coordinates exchange for the IB}
\label{sec:coords-exchange-for-the-IB}

Just as a point in the plane can be described by different coordinate systems, so can IB roots.
As demonstrated recently by \cite{agmon2022RTRD} for the related rate-distortion theory, picking the right coordinates matters when analyzing its bifurcations. 
The same holds also for the IB. 
Our primary motivations for exchanging coordinates are to reduce computational costs and to mod-out irrelevant kernel directions, as explained in Section \ref{sec:introduction}. 
In this Section, we discuss three natural choices of a coordinate system for parametrizing IB roots and the reasoning behind our choice for the sequel before setting to derive the IB ODE in the following Section \ref{sec:IB-ODE}. 
This work is complemented by the later Subsection \ref{sub:IB-as-an-RD-problem-and-non-singularity-conj}, which facilitates a transparent analysis of IB bifurcations. 

\medskip
IB roots have been classically parameterized in the literature by (direct) encoders $p(\xhat|\x)$, following \cite{tishby1999}. 
Considering the BA-IB Algorithm \ref{algo:BA-IB} reveals two other natural choices, illustrated by Equation \eqref{eq:coordinate-sets-parameterizing-an-IB-root} below. 
First, an encoder $p(\xhat|\x)$ determines a \textit{cluster marginal} $p(\xhat)$ and an \textit{inverse encoder} $p(\x|\xhat)$, via \algref{algo:BA-IB}{eq:IB-BA-cluster_marginal} and \algref{algo:BA-IB}{eq:IB-BA-bayes-for-computing-inverse-enc}, respectively. 
These can be interpreted geometrically as $p(\xhat)$-weighted points $q_{\xhat}(x)$ in the simplex $\Delta[\cal{X}]$ of $X$, so long that these are well-defined, $\forall \xhat \; p(\xhat) \neq 0$.
No more than $|\mathcal{X}| + 1$ points in the simplex are required to represent an IB root, \citep{witsenhausen1975conditional}.
The latter is readily seen to analyze the IB in these coordinates\footnote{ Although known among IB practitioners, this reference has generally escaped broader attention.}, though it pre-dates \citep{tishby1999}. 
Second, an inverse encoder determines a \textit{decoder} $p(\y|\xhat)$, via \algref{algo:BA-IB}{eq:IB-BA-decoder-eq}. Along with the cluster marginal, $\big( p(\y|\xhat), p(\xhat) \big)$ can be similarly interpreted as $p(\xhat)$-weighted points $r_{\xhat}(y)$ in the simplex $\Delta[\mathcal{Y}]$ of $Y$.
This choice of coordinates is implied already by \cite[Theorem 5]{tishby1999}.
Cycling around Equation \eqref{eq:coordinate-sets-parameterizing-an-IB-root}, a decoder $\big( p(\y|\xhat), p(\xhat)\big)$ determines via \algref{algo:BA-IB}{eq:IB-BA-partition-func} and \algref{algo:BA-IB}{eq:IB-BA-new-direct-enc} a new encoder, which may differ from the one with which we have started. 
For notational simplicity, we shall usually write $\big( p(\y|\xhat), p(\xhat) \big)$ rather than $\big( r_{\xhat}(\y), p(\xhat) \big)$ for decoder coordinates (similarly, for inverse-encoder coordinates).

\begin{equation}			\label{eq:coordinate-sets-parameterizing-an-IB-root}
\xymatrix{
		&	p(\xhat|\x) \ar@(l, u)[dl]_{\algref{algo:BA-IB}{eq:IB-BA-cluster_marginal}, \algref{algo:BA-IB}{eq:IB-BA-bayes-for-computing-inverse-enc}}		&	\\
		\big( p(\x|\xhat), p(\xhat) \big) \ar@(d, d)[rr]_{\algref{algo:BA-IB}{eq:IB-BA-decoder-eq}}	&&
		\big( p(\y|\xhat), p(\xhat) \big) \ar@(u, r)[ul]_{\algref{algo:BA-IB}{eq:IB-BA-partition-func}, \algref{algo:BA-IB}{eq:IB-BA-new-direct-enc}}
	}
\end{equation}

The above allows us to define \textit{three} BA operators as the composition of three consecutive maps in Equation \eqref{eq:coordinate-sets-parameterizing-an-IB-root}, encoding an iteration of Algorithm \ref{algo:BA-IB}.
When starting at an encoder $p(\xhat|\x)$, its output is a newly-defined encoder. 
Similarly, when starting at one of the other two vertices, it sends an inverse-encoder $\big( p(\x|\xhat), p(\xhat) \big)$ or a decoder pair $\big( p(\y|\xhat), p(\xhat) \big)$ to newly-defines one. 
By abuse of notation, we denote all three compositions by $BA_\beta$, with the choice of coordinates system mentioned accordingly. Indeed, these are representations of a single BA-IB iteration in three different coordinate systems, and so may be considered as distinct representations of the same operator. 
For completeness, $BA_\beta$ in decoder coordinates is spelled out explicitly at Equation \eqref{eq:BA-operator-def-in-decoder-coords-appendix} in Appendix \ref{sec:BA-IB-in-decoder-coords-appendix}. 
A newly-defined encoder (or inverse-encoder or decoder) at a cycle's completion need not generally equal the one at which we have started. 
These are equal precisely at IB roots, when the IB Equations \eqref{eq:IB-eq-encoder}-\eqref{eq:IB-eq-marginal} hold. 
Therefore, the choice of a coordinates system does \textit{not} matter then, and so moving around Equation \eqref{eq:coordinate-sets-parameterizing-an-IB-root} from one vertex to another yields different parameterizations of the same root, at least when $\forall \xhat \; p(\xhat) \neq 0$. 
In particular, this shows that the inverse-encoders $q_{\xhat}$ in $\Delta[\cal{X}]$ of an IB root are in bijective correspondence with its decoders $r_{\xhat}$ in $\Delta[\cal{Y}]$, an observation which shall come in handy at Section \ref{sec:IB-bifurcations}. 

\medskip
Next, we consider how well each of these coordinate systems can serve for following the path of an IB root. 
The minimal number of symbols $\xhat$ needed to write down an IB root typically varies with the constraints; cf., \cite[Section 3.4]{tishby1999} or \cite[Section II.A]{witsenhausen1975conditional}. 
Therefore, inverse-encoder and decoder coordinates are better suited than encoder coordinates for considering the dynamics of a root with $\beta$, as they allow us to consider its evolution via a varying number of points in a fixed space, $\Delta[\mathcal{X}]$ or $\Delta[\mathcal{Y}]$, respectively. 
Indeed, a direct encoder $p(\xhat|\x)$ can be interpreted geometrically as a point in the $|\mathcal{X}|$-fold product $\Delta[\hat{\mathcal{X}}]^{\cal{X}}$ of simplices $\Delta[\hat{\mathcal{X}}]$, \cite[Section 2]{gedeon2012mathematical}.
So, if a particular symbol $\xhatp$ is not in use anymore, $p(\xhatp) = 0$, then one is forced to choose between replacing $\Delta[\hat{\mathcal{X}}]$ by a smaller space $\Delta[\hat{\mathcal{X}} \setminus \{\xhatp\}]$, to carrying on with a redundant symbol $\xhatp$. 
The latter leads to non-trivial kernels in the IB due to duplicate clusters\footnote{ In addition to the IB's ``perpetual kernel'', \citep{gedeon2012mathematical, parker2022symmetry_breaking}.} (e.g., Section 3.1 there), making it difficult to tell whether a particular kernel direction pertains to a bifurcation. 
In contrast, when considered in decoder coordinates, for example, an IB root is nothing but $p(\xhat)$-weighted paths $r_1, \dots, r_T$ in $\Delta[\mathcal{Y}]$, with $\beta \mapsto r_{\xhat}(\beta)$ a path for each $r_{\xhat}$.
And so, once a symbol $\xhatp$ is not needed anymore, then one can discard the path $r_{\xhatp}$ without replacing the underlying space $\Delta[\mathcal{Y}]$.
This permits the clean treatment of IB bifurcations in Section \ref{sec:IB-bifurcations}.

The computational cost of solving the first-order ODE \eqref{eq:implicit-beta-ODE} for $\tfrac{d\bm{x}}{d\beta}$ numerically depends on $\dim \bm{x}$. 
Much of this cost is due to computing a linear pre-image under $D_{\bm{x}} F$, which is of order $O(\dim \bm{x})^3$, \cite[Section 28.4]{cormen2022introduction}. cf., Section \ref{sec:IBRT1-algo}. 
Representing an IB root on $T$ clusters in encoder coordinates requires\footnote{ With the coordinates considered as independent variables, ignoring normalization constraints.} $|\mathcal{X}|\cdot T$ dimensions, in inverse-encoder coordinates $(|\mathcal{X}| + 1)\cdot T$ dimensions, and in decoder coordinates $(|\mathcal{Y}| + 1)\cdot T$ dimensions.
Thus, when there are fewer possible labels $\mathcal{Y}$ than input symbols $\mathcal{X}$, then the computational cost is lowest in decoder coordinates.

\begin{figure}[h!]
	\centering
	\vspace*{10pt}
	\ifdefined\compilefigs
	\includegraphics[width=.65\textwidth]{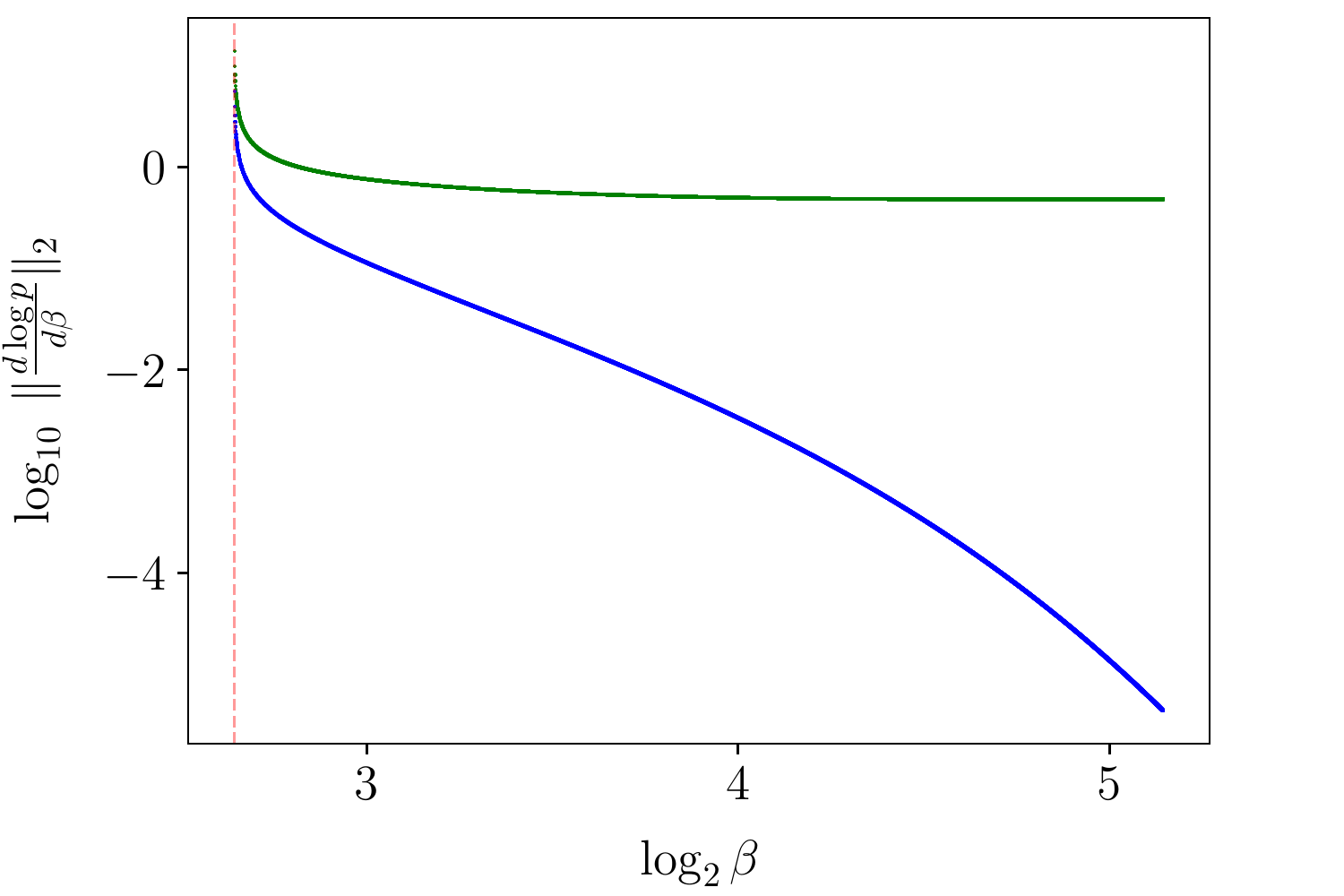}
	\else
	\includegraphics[width=.65\textwidth]{figs/empty_figure}
	\fi
	\caption{
		\textbf{Derivatives' norm by coordinate system}, for the exact solution of BSC(0.3) with a uniform source, as in Figure \ref{fig:IBRT1-IB-curve-for-several-densitires}; see Appendix \ref{sec:analytical-IB-sol-for-BSC-appendix}.
		The derivative's $L_2$-norm is plotted in green for encoder coordinates and blue for decoder coordinates.
		The solution barely changes at high $\beta$ values, and so the derivative in \textit{decoder} coordinates is smaller (see main text). 
		Nevertheless, the derivative in \textit{encoder} coordinates does not vanish then, due to Equation \eqref{eq:coordinates-exchange-dec-to-enc-explicit}.
		At low $\beta$ values, however, the derivative in either coordinate system may generally be large. 
		Both vanish to the left of the bifurcation in this problem (dashed red vertical), as the solution there is trivial (single-clustered). 
		The derivatives diverge near the bifurcation (to its right) regardless of the coordinate system, as might be expected by the implicit ODE \eqref{eq:implicit-beta-ODE} --- see also Section \ref{sub:IBRT1-algo-spec}. 
	}
	\label{fig:norm-of-analytical-derivs-for-BSC}
\end{figure}

A-priori, one might expect that derivatives with respect to $\beta$ vanish when the solution barely changes, regardless of the choice of coordinates system.
For example, at a very large ``$\beta = \infty$'' value, an obvious IB root is the diagonal encoder\footnote{ That is, set $\hat{\mathcal{X}} := \mathcal{X}$ and $p(\xhat|\x) := 1$ if $\xhat = \x$ and 0 otherwise. }, as can be seen by a direct examination of the IB Equations \eqref{eq:IB-eq-encoder}-\eqref{eq:IB-eq-marginal}. 
It consists of one IB cluster of weight (or \textit{mass}) $p(\x)$ at $p_{Y|X=x} \in \Delta[\mathcal{Y}]$ for each $\x\in \mathcal{X}$, and so one might expect that it would barely change so long that $\beta$ is very large. 
However, the logarithmic\footnote{ See Section \ref{sec:IB-ODE} below on the use of logarithmic derivatives.} derivative $\tfrac{d\log p_\beta(\xhat|\x)}{d\beta}$ in encoder coordinates need \textit{not} vanish even when the derivatives $\tfrac{d\log p_\beta(\y|\xhat)}{d\beta}$ and $\tfrac{d\log p_\beta(\xhat)}{d\beta}$ in decoder coordinates do, as seen to the right of Figure \ref{fig:norm-of-analytical-derivs-for-BSC}.
Indeed, given the derivative in decoder coordinates, one can exchange it to encoder coordinates by
\begin{multline}								\label{eq:coordinates-exchange-dec-to-enc-explicit}
	\dbeta{\log p_\beta(\xhat|\x)} = 
	J_{\text{dec}}^{\text{enc}} \; \dbeta{\log p_\beta(\yp|\xhatp)} +
	J_{\text{mrg}}^{\text{enc}} \; \dbeta{\log p_\beta(\xhatp)} 
	\\ - D_{KL}\big[p(\y|\x) || p_\beta(\y|\xhat)\big] 
	+ \sum_{\xhatpp} p_\beta(\xhatpp|\x) D_{KL}\big[p(\y|\x) || p_\beta(\y|\xhatpp)\big] \;,
\end{multline}
where $J_{\text{dec}}^{\text{enc}}$ and $J_{\text{mrg}}^{\text{enc}}$ are the two coordinate exchange Jacobian matrices of orders $(T\cdot |\mathcal{X}|)\times (T\cdot |\mathcal{Y}|)$ and $(T\cdot |\mathcal{X}|)\times T$ respectively, given by Equations \eqref{eq:coordinates-exchange-dec-to-enc-enc-in-terms-of-dec:appendix} and \eqref{eq:coordinates-exchange-dec-to-enc-enc-in-terms-of-marginal:appendix} in Appendix \ref{subsub:decoder-to-encoder-coords-jacobian-appendix}. 
And so, $\tfrac{d\log p_\beta(\xhat|\x)}{d\beta}$ would often be non-zero even if both $\tfrac{d\log p_\beta(\y|\xhat)}{d\beta}$ and $\tfrac{d\log p_\beta(\xhat)}{d\beta}$ vanish. 
This unintuitive behavior of the derivative in encoder coordinates is due to the explicit dependence of the IB's encoder Equation \eqref{eq:IB-eq-encoder} on $\beta$. This dependence is the source of the last two terms in Equation \eqref{eq:coordinates-exchange-dec-to-enc-explicit} (see Equation \eqref{eq:coordinates-exchange-dec-to-enc-enc-in-terms-of-beta:appendix}). 
The comparison between encoder and inverse-encoder coordinates can be seen to be similar. See Appendix \ref{sub:coordinates-exchange-jacobians-appendix} for further details.

\medskip
In light of the above, we proceed with decoder coordinates in the sequel.

\medskip
\section{Implicit derivatives at an IB root and the IB's ODE}
\label{sec:IB-ODE}

We now specialize the implicit ODE \eqref{eq:implicit-beta-ODE} (of Section \ref{sec:introduction}) to the IB, using the decoder coordinates of the previous Section \ref{sec:coords-exchange-for-the-IB}. 
This allows us to compute first-order implicit derivatives at an IB root (Theorem \ref{thm:IB-ODE}) at a remarkable accuracy, under one primary assumption. Namely, that the root is a differentiable function of $\beta$. 
While differentiability breaks at IB bifurcations (Section \ref{sec:IB-bifurcations}), this allows to reconstruct a solution path from its local approximations in the following Section \ref{sec:euler-method}, so long that it holds.

\medskip
To simplify calculations, we take the logarithm $\big(\log p(\y|\xhat), \log p(\xhat) \big)$ of the decoder coordinates of Section \ref{sec:coords-exchange-for-the-IB} as our variables. 
Exchanging the $BA_\beta$ operator to log-decoder coordinates is immediate, by writing $\log BA_\beta[\exp{\left(\log p(\y|\xhat)\right)}, \exp{\left(\log p(\xhat)\right)} ]$. 
For short, we denote it $BA_\beta[\log p(\y|\xhat), \log p(\xhat)]$ when in these coordinates, by abuse of notation. 
Similarly, exchanging the IB ODE (below) back to non-logarithmic coordinates is immediate, via $\tfrac{d}{d\beta} \log p = \tfrac{1}{p} \tfrac{d}{d\beta} p$. 
In Section \ref{sec:IBRT1-algo} we shall assume that $p(\xhat)$ never vanishes.
To ensure that taking logarithms is well-defined, we also require\footnote{ While a decoder $p(\y|\xhat)$ may have a well-defined derivative $\dbeta{}p(\y|\xhat)$ even without this requirement, the calculation details below would differ. } that no coordinate of $p(\y|\xhat)$ vanishes. 
A sufficient condition for that is that $p(\y|\x) > 0$ for every $\x$ and $\y$ (Lemma \ref{lemma:sufficient-condition-for-IB-decoder-to-never-vanish} in Appendix \ref{sec:BA-IB-in-decoder-coords-appendix}). 

Next, define a variable $\bm{x} \in \bb{R}^{T\cdot(|\mathcal{Y}| + 1)}$ as the concatenation of the vector $\big( \log p_\beta(\y|\xhat) \big)_{\y\in \mathcal{Y}, \xhat\in \hat{\mathcal{X}}}$ with $\big(\log p_\beta(\xhat) \big)_{\xhat\in \hat{\mathcal{X}}}$. 
Differentiating $\nicefrac{\partial}{\partial \log p}$ with respect to log-probabilities is given by $p \cdot \tfrac{\partial}{\partial p}$, by the chain rule\footnote{ Defining $u := \log p$, the $u$-derivative of $f(p)$ is given by $\tfrac{df}{du} = \tfrac{df}{dp} \tfrac{dp}{du}$, or equivalently $\tfrac{df}{d\log p} = p\cdot \tfrac{df}{dp}$. See also Appendix \ref{sub:BA-IB-jacobian-appendix:calculation-setup-and-goals} for a gentler treatment. }. 
This gives meaning to the Jacobian matrix $D_{\bm{x}} (\cdot)$ with respect to our logarithmic variable $\bm{x}$. 
The Jacobian $D_{\log p(\y|\xhat), \log p(\xhat)} BA_\beta$ of a single Blahut-Arimoto iteration in these log-decoder coordinates is a square matrix of order $T\cdot (|\mathcal{Y}| + 1)$. 
Its $(T\cdot |\mathcal{Y}|)\times (T\cdot |\mathcal{Y}|)$ upper-left block (below) corresponds to perturbations in BA's output log-decoder $\log p(\y|\xhat)$ due to varying an input log-decoder $\log p(\yp|\xhatp)$. 
Since we prime input but not output coordinates, this is to say that the \textit{columns} of this block are indexed\footnote{ Alternatively, one can enumerate the label and representation alphabets explicitly, $\mathcal{Y} := \{\y_1, \dots, \y_{|\mathcal{Y}|}\}$ and $\hat{\mathcal{X}} := \{\xhat_1, \dots, \xhat_T\}$. This allows to replace $(\y, \xhat)$ and $(\yp, \xhatp)$ throughout by $(\y_i, \xhat_j)$ and $(\y_k, \xhat_l)$, respectively, with $i, k=1,\dots, |\mathcal{Y}|$ and $j, l = 1, \dots, T$.
} by pairs $(\yp, \xhatp)$ and its \textit{rows} by $(\y, \xhat)$. 
Its $(T\cdot |\mathcal{Y}|)\times T$ upper-right block corresponds to perturbations in BA's output log-decoder $\log p(\y|\xhat)$ due to varying an input log-marginal $\log p(\xhatp)$. 
That is, its columns are indexed by $\xhatp$ and rows by $(\y, \xhat)$. 
Similarly, for the bottom-left and bottom-right blocks, of respective sizes $T \times (T\cdot |\mathcal{Y}|)$ and $T\times T$.
See \eqref{eq:BA-Jacob-wrt-decoder-coords-as-block-matrix-implicit} ff., in Appendix \ref{sub:BA-IB-jacobian-appendix:decoder-deriv-matrix}, and the end-result at Equation \eqref{eq:BA-Jacob-wrt-decoder-coords-as-block-matrix-explicit} there. 
Explicitly, when evaluated at an IB root $\big( \log p(\y|\xhat), \log p(\xhat) \big)$, BA's Jacobian matrix is given by
\begin{multline}		\label{eq:BA-Jacob-wrt-decoder-coords-at-main-text}
	D_{\log p(\y|\xhat), \log p(\xhat)} BA_\beta[\log p(\y|\xhat), \log p(\xhat)] = \\[10pt]
\left(
	\begin{array}{c|c}
		\beta \cdot \sum_{\xhatpp, \ypp} 
		\left( \delta_{\xhatpp, \xhatp} - \delta_{\xhat, \xhatp} \right)
		\cdot \left[1 - \tfrac{\delta_{\ypp, \y} }{p_{\beta}(\y|\xhat)} \right]
		C(\xhat, \xhatpp; \beta)_{\yp, \ypp} 		&
		\left( 1 - \beta \right) \cdot \sum_{\ypp} \Big[ 
		1 - \tfrac{\delta_{\ypp, \y}}{p_{\beta}(\y|\xhat)} 
		\Big] B(\xhat, \xhatp; \beta)_{\ypp}	\\	\\
		\hline	\\
		\beta \cdot \Big[
		\delta_{\xhat, \xhatp} \; p_{\beta}(\yp|\xhat) -
		B(\xhat, \xhatp; \beta)_{\yp}
		\Big] &
		\left(1 - \beta \right) \cdot 
		\Big[ \delta_{\xhat, \xhatp} - A(\xhat, \xhatp; \beta) \Big]
	\end{array}
\right)
\end{multline}
where $\delta_{i, j} = 1$ if $i = j$ and is 0 otherwise. 
As mentioned above, primed coordinates $\yp$ and $\xhatp$ index the columns, and un-primed coordinates $\y$ and $\xhat$ the rows. Indices $\ypp$ and $\xhatpp$ with more than a single prime are summation variables. 
$A, B$ and $C$ are a scalar, a vector, and a matrix, each involving two IB clusters.  They are defined by,
\begin{equation}			\label{eq:B_defs-for-BA-Jacob-wrt-decoder-coords}
	\begin{split}
		A(\xhat, \xhatp; \beta) 			:= &\sum_{\xpp} p_\beta(\xhatp|\xpp) p_\beta(\xpp|\xhat) \;,	\\
		B(\xhat, \xhatp; \beta)_{\y} 		:= &\sum_{\xpp} p(\y|\xpp) p_\beta(\xhatp|\xpp) p_\beta(\xpp|\xhat) \;, \quad \text{and} \\
		C(\xhat, \xhatp; \beta)_{\y, \yp} 	:= &\sum_{\xpp} p(\y|\xpp) p(\yp|\xpp) p_\beta(\xhatp|\xpp) p_\beta(\xpp|\xhat) \;.
	\end{split}
\end{equation}
In these, $\y$ indexes $B$ and the rows of $C$, $\yp$ the columns of $C$, and $\xpp$ is a summation variable. 
These $(\xhat, \xhatp)$-labeled tensors have only $|\mathcal{Y}|$ entries along each axis, thanks to the choice of decoder coordinates. 
$A$ and $B$ can be expressed in terms of $C$ via some obvious relations; see Equation~\eqref{eq:B_C_defs-for-BA-Jacob-wrt-decoder-coords-appendix} and below in Appendix \ref{sub:BA-IB-jacobian-appendix:decoder-deriv-matrix}.
Appendix \ref{sub:BA-IB-jacobian-appendix:calculation-setup-and-goals} elaborates on the mathematical subtleties involved in calculating the Jacobian \eqref{eq:BA-Jacob-wrt-decoder-coords-at-main-text}.
See also Equation \eqref{eq:BA-Jacob-wrt-decoder-coords-as-block-matrix-explicit-implementation-friendly} in Appendix \ref{sub:BA-IB-jacobian-appendix:decoder-deriv-matrix} for an implementation-friendly form of \eqref{eq:BA-Jacob-wrt-decoder-coords-at-main-text}.

Together with $D_\beta BA_\beta$ (Equations \eqref{eq:partial-beta-deriv-calcs-beta-deriv-of-output-margianl} and \eqref{eq:partial-beta-deriv-calcs-beta-deriv-of-output-decoder} in Appendix \ref{sub:BA-IB-beta-deriv-appendix-wrt-decoder-coords}), we have both of the first-order derivative tensors of $BA_\beta$ in log-decoder coordinates. 
This allows us to specialize the implicit ODE \eqref{eq:implicit-beta-ODE} (of Section \ref{sec:introduction}) to the IB, in terms of our variable $\bm{x}$. By abuse of notation, we write $\big( \log p_\beta(\y|\xhat), \log p_\beta(\xhat) \big)_{\y, \xhat}$ for its $|\mathcal{Y}|\cdot T + T$ coordinates, and similarly for its derivatives vector $\bm{v}$ \eqref{eq:IB-beta-ODE-notation-for-implicit-deriv} below. 

\begin{thm}[The IB's ODE]			\label{thm:IB-ODE}
	Let $\big( p(\y|\xhat), p(\xhat) \big)$ be an IB root, 
	and suppose that it can be written as a differentiable function $\beta\mapsto \big( p_\beta(\y|\xhat), p_\beta(\xhat) \big)$ in $\beta$. 
	If none of its coordinates vanish, then the vector 
	\begin{equation}			\label{eq:IB-beta-ODE-notation-for-implicit-deriv}
		\bm{v} := \left( \dbeta{\log p_\beta(\y|\xhat)}, \dbeta{\log p_\beta(\xhat)} \right)_{\y, \xhat}
	\end{equation}
	of its implicit logarithmic derivatives is well defined and satisfies an ordinary differential equation in $\beta$,
	\begin{equation}			\label{eq:IB-beta-ODE-in-decoder-coords}
		\Big( I - D_{\log p(\y|\xhat), \log p(\xhat)} BA_\beta \Big) \bm{v} = 
		\mat{
			- \sum_{\x, \xhatpp} \left[ 1 - \frac{p(\y|\x)}{p_\beta(\y|\xhat) } \right] \cdot \Big[ \delta_{\xhat, \xhatpp} - p_\beta(\xhatpp|\x) \Big] \; p_\beta(\x|\xhat) \; D_{KL}\big[ p(\y|\x) || p_\beta(\y|\xhatpp) \big] \\ \\
			\sum_{\x, \xhatpp} \Big[ \delta_{\xhat, \xhatpp} - p_\beta(\xhatpp|\x) \Big] \; p_\beta(\x|\xhat) \; D_{KL}\big[ p(\y|\x) || p_\beta(\y|\xhatpp) \big] 
		}
	\end{equation}
	where $I$ is the identity matrix of order $T\cdot (|\mathcal{Y}| + 1)$, and the Jacobian matrix $D_{\log p(\y|\xhat), \log p(\xhat)} BA_\beta$ at the given IB root is given by Equation \eqref{eq:BA-Jacob-wrt-decoder-coords-at-main-text}. 
	The right-hand side of \eqref{eq:IB-beta-ODE-in-decoder-coords} is indexed as in \eqref{eq:IB-beta-ODE-notation-for-implicit-deriv}, by $(\y, \xhat)$ at its top and $\xhat$ at its bottom coordinates. 
\end{thm}

While the IB ODE was discovered by \cite{agmon2022thesis}, it is derived here anew in log-decoder coordinates due to the considerations in Section \ref{sec:coords-exchange-for-the-IB}. 
It is analogous to the RD ODE, due to \cite{agmon2022RTRD}; Corollary \ref{cor:RD-root-through-of-tangent-prob-at-reduced-IB-root-is-analytic} and around (in Section \ref{sub:IB-as-an-RD-problem-and-non-singularity-conj}) provides a relation between these two ODEs. 
We emphasize that the first assumption of Theorem \ref{thm:IB-ODE}, that the IB root is a differentiable function of $\beta$, is essential. It is comprised of two parts: (i) that the root can be written as a function of $\beta$, and (ii) that this function is differentiable. 
These are precisely the assumptions needed to compute the first-order implicit multivariate derivative $\bm{v}$ \eqref{eq:IB-beta-ODE-notation-for-implicit-deriv} at the given root, \cite[Section 2.1]{agmon2022RTRD}. 
Continuous IB bifurcations violate (ii) (Subsection \ref{sub:continuous-IB-bifs}), while discontinuous ones violate (i) (Subsection \ref{sub:discontinuous-IB-bifs}). 
In contrast, the requirement that no coordinate vanishes is a technical one, due to our choice of logarithmic coordinates. 

It is not necessary for the Jacobian of the IB operator \eqref{eq:IB-operator-def} (to the left of \eqref{eq:IB-beta-ODE-in-decoder-coords}) to be non-singular in order to solve the IB ODE numerically. 
Nevertheless, non-singularity of the Jacobian will follow from the sequel (see Conjecture \ref{conj:BA-IB-Jacob-in-decoder-coords-is-nonsingular-at-reduced-root} in Section \ref{sec:IB-bifurcations}). 
With that, the derivatives $\bm{v} = \tfrac{d}{d\beta} \big( \log p_\beta(\y|\xhat), \log p_\beta(\xhat) \big)$ \eqref{eq:IB-beta-ODE-notation-for-implicit-deriv} computed numerically from the IB ODE \eqref{eq:IB-beta-ODE-in-decoder-coords} at an exact root are remarkably accurate, as demonstrated in Figure \ref{fig:err-in-numerical-derivs-for-BSC-and-BAs-accuracy}. 
As in RD, \citep{agmon2022RTRD}, calculating implicit derivatives numerically loses its accuracy when approaching a bifurcation because the Jacobian is increasingly ill-conditioned there.
For comparison, the BA-IB Algorithm \ref{algo:BA-IB} also loses its accuracy near a bifurcation. This is a consequence of BA's critical slowing down, \citep{agmon2021critical}, just as with its corresponding RD variant. 

\begin{figure}[h!]
	\centering
	\vspace*{10pt}
	\ifdefined\compilefigs
	\includegraphics[width=1\textwidth]{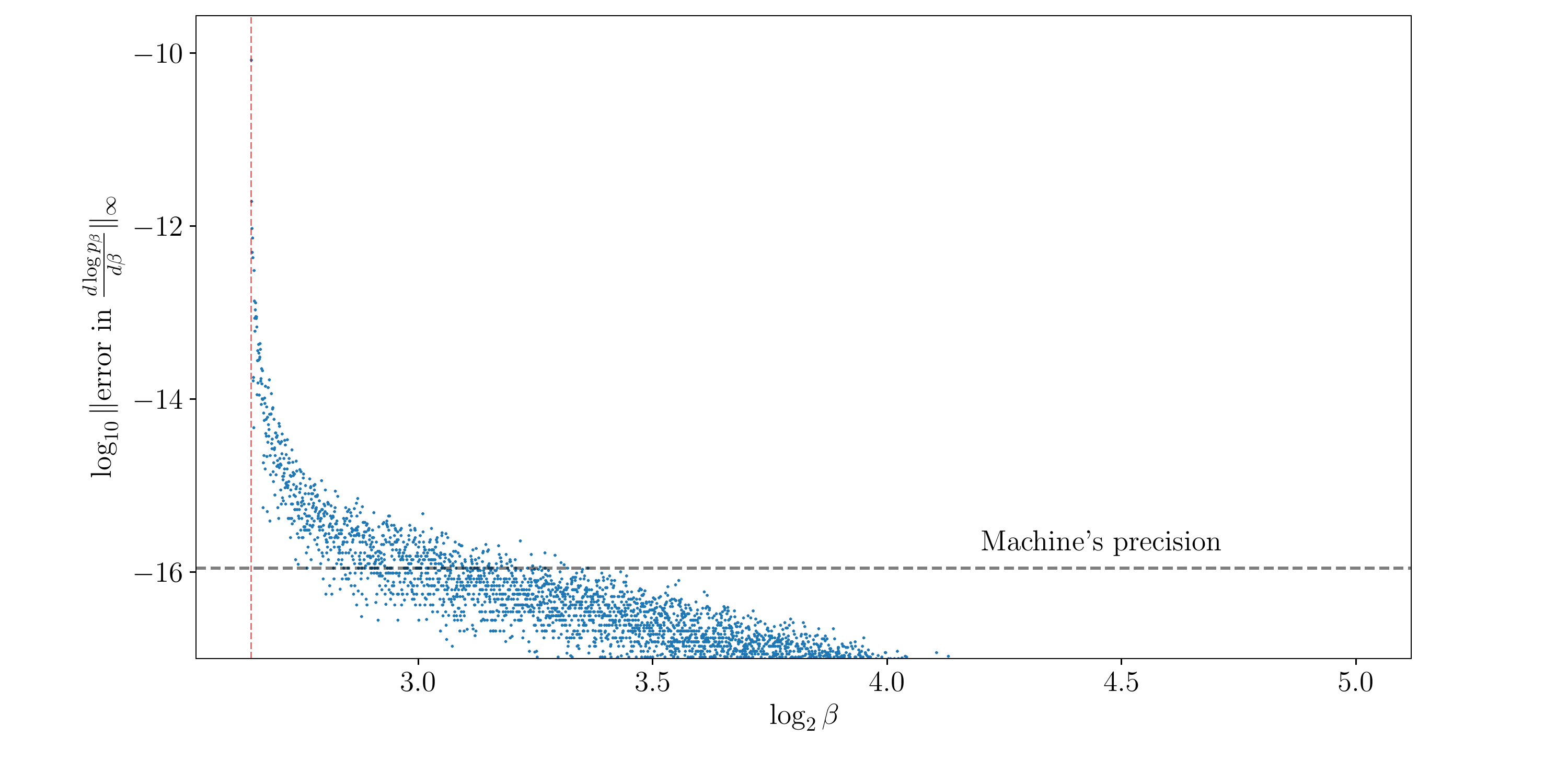}
	\includegraphics[width=1\textwidth]{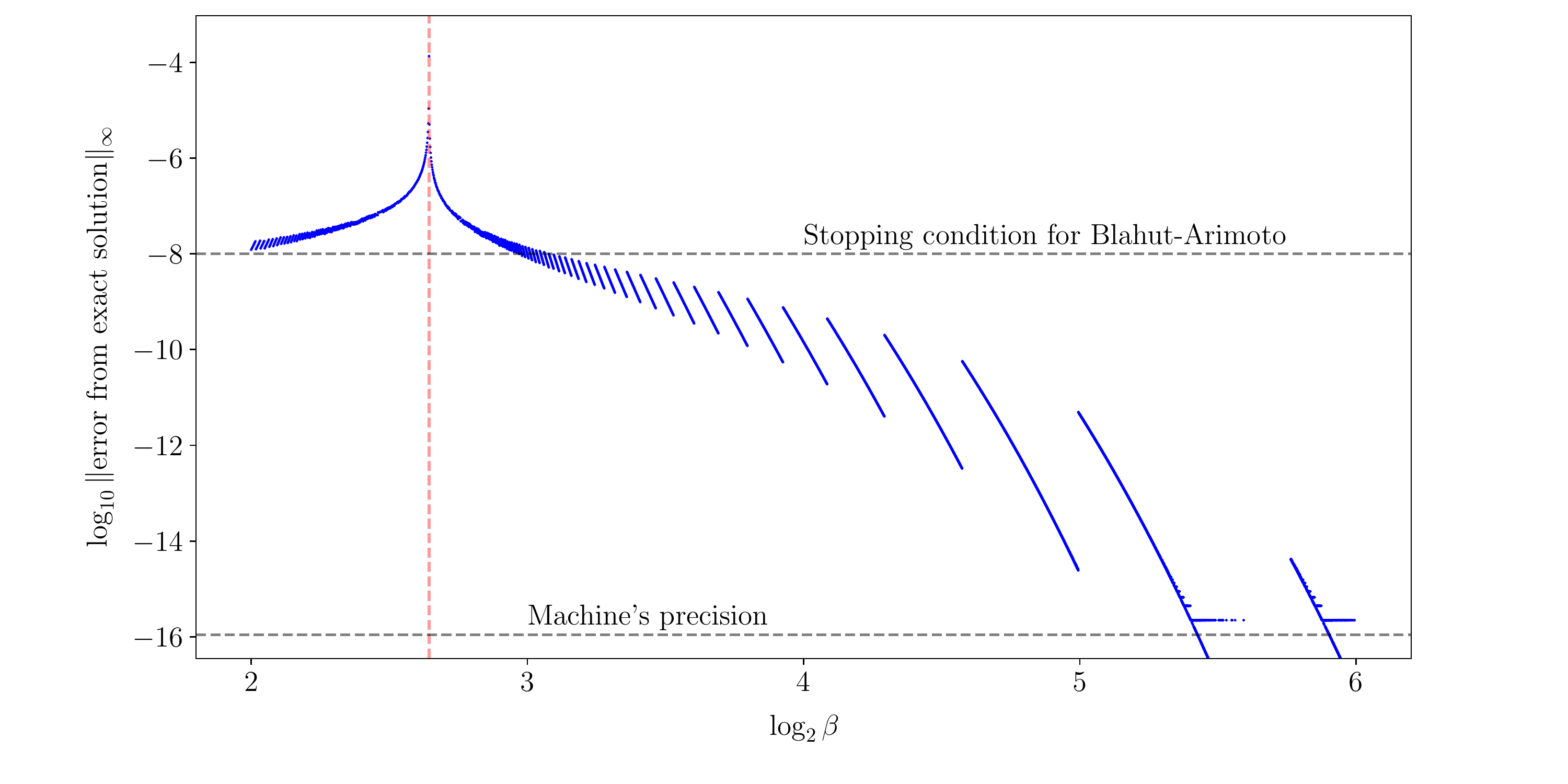}
	\else
	\includegraphics[width=1\textwidth]{figs/empty_figure}
	\includegraphics[width=1\textwidth]{figs/empty_figure}
	\fi
	\caption{
		\textbf{The implicit derivatives computed from the IB ODE \eqref{eq:IB-beta-ODE-in-decoder-coords} are very accurate, as is the BA-IB Algorithm \ref{algo:BA-IB}. However, both lose their accuracy near a bifurcation.}
		To verify their accuracy, we compared both to the exact solutions of BSC$(0.3)$ with a uniform source (see Appendix \ref{sec:analytical-IB-sol-for-BSC-appendix}).
		\textbf{Top:} Derivatives were computed at the problem's exact solution using the IB ODE \eqref{eq:IB-beta-ODE-in-decoder-coords} and compared to the problem's exact derivatives. These are accurate beyond the machine's precision, except when approaching the bifurcation (red vertical), since the Jacobian of the IB operator \eqref{eq:IB-operator-def} is ill-conditioned there. 
		\textbf{Bottom:} The $L_\infty$-errors of the solutions produced by the BA-IB Algorithm \ref{algo:BA-IB}, with a $10^{-8}$ stopping condition, and uniform initial conditions. Error is measured from the true direct encoder to avoid biases due to clusters of low mass.
		Both plots are as in \cite[Figure 2.3]{agmon2022RTRD}.
	}
	\label{fig:err-in-numerical-derivs-for-BSC-and-BAs-accuracy}
\end{figure}

Each coordinate of $\big( p(\y|\xhat), p(\xhat) \big)$ is treated by the IB ODE \eqref{eq:IB-beta-ODE-in-decoder-coords} as an independent variable. However, normalization of $p(\y|\xhat)$ imposes one constraint\footnote{ Ignoring the normalization of the marginal $p(\xhat)$.} per cluster $\xhat$. 
Thus, one might expect the behavior of BA's Jacobian \eqref{eq:BA-Jacob-wrt-decoder-coords-at-main-text} to be determined by fewer than $T\cdot \left(|\mathcal{Y}| + 1 \right)$ coordinates, at least qualitatively. 
This intuition is justified by the following Lemma \ref{lem:smaller-matrix-for-ker-of-J-in-decoder-coords}, which allows to consider the kernel of the IB operator \eqref{eq:IB-operator-def} by a smaller and simpler matrix $S$;
see Appendix \ref{sec:kernel-of-BA-IB-J-decoder-coords:appendix} for its proof. 

\begin{lemma}			\label{lem:smaller-matrix-for-ker-of-J-in-decoder-coords}
	Given an IB root as above, define a square matrix of order $T \cdot |\mathcal{Y}|$ by
	\begin{equation}			\label{eq:smaller-matrix-for-ker-of-IB-operator-in-dec-coords-main-text}
		S_{(\y, \xhat), (\yp, \xhatp)} :=
		\sum_{\x} p_\beta(\x|\xhat) \Big[ \beta \cdot \tfrac{p(\y|\x) }{p_\beta(\y|\xhat)} + \left(1 - 2\beta\right) \Big] 
			p(\yp|\x) \Big[ \delta_{\xhat, \xhatp} - p_\beta(\xhatp|\x) \Big] \;.
	\end{equation}
	Then, the nullity of the Jacobian $I - D_{\log p(\y|\xhat), \log p(\xhat)} BA_\beta$ of the IB operator \eqref{eq:IB-operator-def} equals that of $I - S$, 
	where $I$ is the identity matrix (of the respective order), and $S$ is defined by \eqref{eq:smaller-matrix-for-ker-of-IB-operator-in-dec-coords-main-text},
	\begin{equation}			\label{eq:nullity-rank-equality-with-smaller-matrix}
		\dim \ker \left( I - S \right) =
		\dim \ker \left( I - D_{\log p(\y|\xhat), \log p(\xhat)} BA_\beta \right) \;.
	\end{equation}
	Specifically, write $\bm{v} := \left(v_{\y, \xhat}\right)_{\y, \xhat}$ for a left eigenvector which corresponds to $1 \in \eig S$.
	Then, there is a bijective correspondence between the left kernels at both sides of \eqref{eq:nullity-rank-equality-with-smaller-matrix}, mapping 
	\begin{equation}			\label{eq:bijective-correspondence-between-IB-op-ker-to-eigs-of-smaller-matrix}
		\bm{v} \mapsto (\bm{v}, \bm{u}) \;,
	\end{equation}
	where $\bm{u} := \left(u_{\xhat}\right)_{\xhat}$ is defined by $u_{\xhat} := \tfrac{1 - \beta}{\beta} \cdot \sum_{\y} v_{\y, \xhat}$.
\end{lemma}


In addition to offering a form more transparent than BA's Jacobian in \eqref{eq:BA-Jacob-wrt-decoder-coords-at-main-text}, Lemma \ref{lem:smaller-matrix-for-ker-of-J-in-decoder-coords} also reduces the computational cost of testing $I - D_{\log p(\y|\xhat), \log p(\xhat)} BA_\beta$ \eqref{eq:IB-beta-ODE-in-decoder-coords} for singularity, by using the smaller $I - S$ \eqref{eq:smaller-matrix-for-ker-of-IB-operator-in-dec-coords-main-text} in its place.
This makes it easier to detect upcoming bifurcations (see Conjecture \ref{conj:BA-IB-Jacob-in-decoder-coords-is-nonsingular-at-reduced-root} in Section \ref{sec:IB-bifurcations}). 
Further, one can verify directly that the IB ODE \eqref{eq:IB-beta-ODE-in-decoder-coords} indeed follows the right path.
Indeed, if the ODE is non-singular, then, by the Implicit Function Theorem, there is (locally) a unique IB root, which is a differentiable function of $\beta$. 
And so, there is a unique solution path for a numerical approximation to follow. 
Finally, we note that a relation similar to \eqref{eq:nullity-rank-equality-with-smaller-matrix} holds also for eigenvalues of $D_{\log p(\y|\xhat), \log p(\xhat)} BA_\beta$ \eqref{eq:BA-Jacob-wrt-decoder-coords-at-main-text} other than 1. This can be seen either empirically or by tracing the proof of Lemma \ref{lem:smaller-matrix-for-ker-of-J-in-decoder-coords}. 

\medskip
In Section \ref{sec:IB-bifurcations}, we shall proceed with this line of thought of removing irrelevant coordinates. 
In the following Section \ref{sec:euler-method}, we turn to reconstruct a solution path from implicit derivatives at a point, with bifurcations ignored for now.

\medskip
\section{A modified Euler method for the IB}
\label{sec:euler-method}

We follow the path of a given IB root away of bifurcation by using its implicit derivatives computed from the IB ODE \eqref{eq:IB-beta-ODE-in-decoder-coords}, of Section \ref{sec:IB-ODE}.
We follow the classic Euler method for simplicity, modifying it slightly to get the most out of the calculated derivatives. 
Improvements using more sophisticated numerical methods are left to future work. 
The detection and handling of IB bifurcations are deferred to the next Section \ref{sec:IB-bifurcations}, and thus are ignored in this Section. 

\medskip
Let $\tfrac{d\bm{x}}{d\beta} = f(\bm{x}, \beta)$ and $\bm{x}(\beta_0) = \bm{x}_0$ define an initial value problem.
In numerical approximations of ordinary differential equations (ODEs), the \textit{Euler method} for this problem is defined by setting
\begin{equation}		\label{eq:euler-method-def}
	\bm{x}_{n+1} := \bm{x}_n + \Delta \beta \cdot f(\bm{x}_n, \beta_n) \;,
\end{equation}
where $\beta_{n+1} := \beta_n + \Delta \beta$, and $| \Delta \beta |$ is the \textit{step size}.
The \textit{global truncation error} $\max_n \|\bm{x}_n - \bm{x}(\beta_n)\|_\infty$ is the largest error of the approximations $\bm{x}_n$ from the true solutions $\bm{x}(\beta_n)$.
A numerical method for solving ODEs is said to be \textit{of order $d$} if its global truncation error is of order $O(|\Delta \beta|^d)$, for step sizes $|\Delta \beta|$ small enough. 
Euler's method error analysis is a standard result, e.g., \cite[Theorem 212A]{butcher2016numerical} or \cite[Theorem 2.4]{atkinson2011numerical}, brought as Theorem \ref{thm:euler-method-error-analysis} below. 
It shows that Euler's method \eqref{eq:euler-method-def} is of order $d = 1$, under mild assumptions, as demonstrated in Figure \ref{fig:Euler-method-for-IB-dec}.  
The immediate generalization of \eqref{eq:euler-method-def} using derivatives till order $d$ is \textit{Taylor's method}, which is a method of order $d$.

\begin{thm}[Euler's method error analysis]			\label{thm:euler-method-error-analysis}
	Let an initial-value problem be defined on $\left[\beta_0, \beta_f\right]$ by $\tfrac{d\bm{x}}{d\beta} = f(\bm{x}, \beta)$ as above\footnote{ 
		The initial condition $\bm{x}_0$ in \eqref{eq:euler-method-gte-bound} is allowed to deviate from the true solution $\bm{x}(\beta_0)$. 
	}, and suppose that $f$ satisfies the Lipschitz condition with some constant $L > 0$. Namely, $\|f(\bm{x}, \beta) - f(\bm{x}', \beta)\|_\infty \leq L \cdot \|\bm{x} - \bm{x}'\|_\infty$ for every $\bm{x}, \bm{x}'$ and $\beta \in \left[\beta_0, \beta_f\right]$. 
	
	Then, Euler method's \eqref{eq:euler-method-def} global truncation error satisfies
	\begin{equation}		\label{eq:euler-method-gte-bound}
		\max_{\beta_0 \leq \beta_n \leq \beta_f} \| \bm{x}_n - \bm{x}(\beta_n) \|_\infty \leq
		e^{(\beta_f - \beta_0) L} \| \bm{x}_0 - \bm{x}(\beta_0) \|_\infty + 
		\frac{e^{(\beta_f - \beta_0) L} - 1}{L} \cdot \tfrac{1}{2} |\Delta \beta| \max_{\beta_0 \leq \beta \leq \beta_f} \left\| \tfrac{d^{2} \bm{x}(\beta)}{d\beta^{2}} \right\|_\infty \;.
	\end{equation}
\end{thm}

\setcounter{footnote}{8}

\begin{figure}[h!]
	\centering
	\vspace*{-10pt}
	\ifdefined\compilefigs
	\includegraphics[trim={0 0 0 1cm}, clip, width=.9\textwidth]{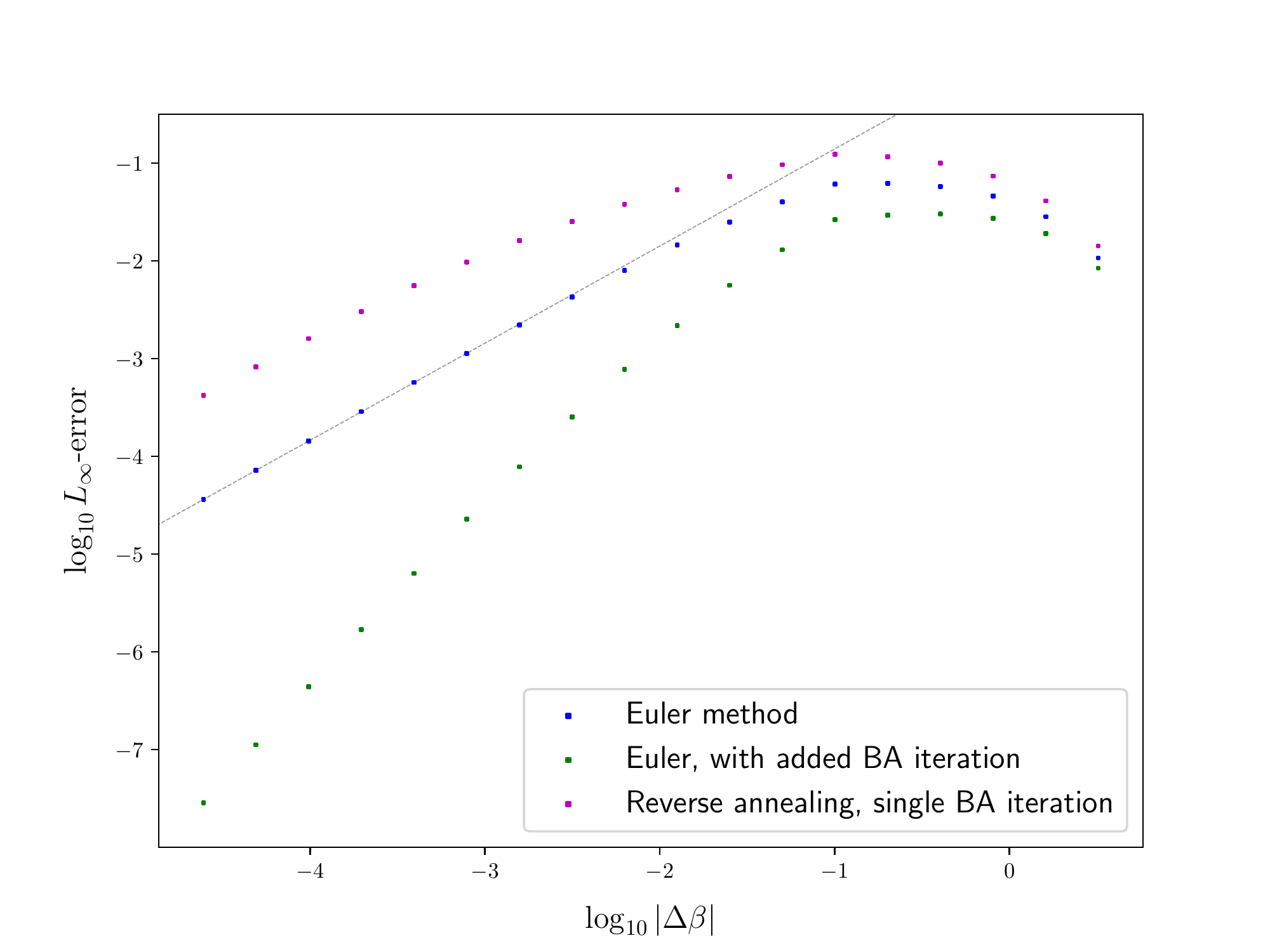}
	\else
	\includegraphics[width=.9\textwidth]{figs/empty_figure}
	\fi
	\caption{
		\textbf{Error by step-size for a vanilla Euler method using the IB ODE \eqref{eq:IB-beta-ODE-in-decoder-coords}, and with an added BA-IB iteration at each step}, for BSC$(0.3)$ with a uniform source (Appendix \ref{sec:analytical-IB-sol-for-BSC-appendix}).
		The linear regression (dashed black) of the third leftmost markers for the vanilla Euler method is of slope $0.99$ ($R^2 \simeq 1$), matching the theory's prediction almost perfectly.
		A similar regression (not shown) for Euler's method with a single added BA iteration is of nearly double slope $1.93$.
		For comparison, reverse deterministic annealing with a single BA iteration at each grid point yields a slope of $0.91$ in this example. 
		Taking a larger (pre-determined) number of iterations at each grid point pushes the error downwards, as expected. 
		Yet, the resulting slopes approach 1 as the number of iterations is increased (not shown). 
		See main text and Appendix \ref{sec:approx-error-analysis-for-BA-and-Euler-for-IB-appendix} for details.
		\newline
		The error was calculated as the supremum of the pointwise errors as in Figure \ref{fig:err-in-numerical-derivs-for-BSC-and-BAs-accuracy}, over the interval $[\beta_c + \tfrac{1}{10}, \beta_0]$ which contains no bifurcation. 
		Each method was initialized with the exact solution at $\beta_0 = 2^5$, with $\Delta \beta = -\frac{103}{32}$ halved between consecutive markers.
	}
	\label{fig:Euler-method-for-IB-dec}
	%
	%
	%
	%
	%
	%
	%
	%
\end{figure}

Specializing Euler's method to our needs, replace $\bm{x}$ in \eqref{eq:euler-method-def} above by the log-decoder coordinates of an IB root, as in Section \ref{sec:IB-ODE}.
So long that an IB root $\bm{p}_\beta := \big( p_\beta(\y|\xhat), p_\beta(\xhat) \big)$ is a differentiable function of $\beta$ in the vicinity of $\beta_n$, it can be approximated by
\begin{align}				\label{eq:first-order-approx}
\begin{split}
	\log p_{\beta_{n+1}}(\y|\xhat) &\approx
	\log p_{\beta_n}(\y|\xhat) + \Delta \beta \cdot \frac{d \log p_\beta(\y|\xhat) }{d\beta} \bigg\rvert_{ \bm{p}_{\beta_n}}  \quad \text{and} \\
	\log p_{\beta_{n+1}}(\xhat) &\approx
	\log p_{\beta_n}(\xhat) + \Delta \beta \cdot \frac{d \log p_\beta(\xhat) }{d\beta} \bigg\rvert_{ \bm{p}_{\beta_n}} \;,
\end{split}
\end{align}
where $\tfrac{d \log p_\beta(\y|\xhat) }{d\beta}$ and $\tfrac{d \log p_\beta(\xhat) }{d\beta}$ are calculated from the IB ODE \eqref{eq:IB-beta-ODE-in-decoder-coords}.
Thus, applying \eqref{eq:first-order-approx} repeatedly, we obtain an Euler method for the IB.
We shall take only negative steps $\Delta \beta < 0$ when approximating the IB, due to reasons explained in Subsection \ref{sub:discontinuous-IB-bifs} (after Proposition \ref{prop:bif-is-detectable-only-on-enough-clusters}).
In contrast to the BA-IB Algorithm \ref{algo:BA-IB}, Euler's method \eqref{eq:first-order-approx} can be used to extrapolate intermediate points, yielding a piecewise linear approximation of the root.

The problem of tracking an operator's root belongs in general to a family of hard-to-solve numerical problems --- known as \textit{stiff} --- if the problem has a bifurcation, \cite[Section 7.2]{agmon2022RTRD}. 
e.g., \cite{butcher2016numerical} or \cite{atkinson2011numerical} on stiff differential equations. 
Stopping early in the vicinity of a bifurcation restricts the computational difficulty and permits convergence guarantees. 
Early-stopping in the IB shall be handled later, in Subsection \ref{sub:continuous-IB-bifs}.
\cite[Theorem 5]{agmon2022RTRD} proves that Euler's method convergence guarantees (Theorem \ref{thm:euler-method-error-analysis}) hold for the closely related Euler method for RD with early stopping. 
While Euler's method may inadvertently switch between solution branches, the latter guarantees ensure that it indeed follows the true solution path between bifurcations, if the step size $|\Delta \beta|$ is small enough and initializing close enough to the true solution. 
While we do not dive into these details for brevity, we note that similar convergence guarantees can also be proven here. 
Alternatively, Euler's method can be ensured to follow the true solution path by noting that an optimal IB root is (strongly) stable when negative steps $\Delta\beta < 0$ are taken; these details are deferred to Subsection \ref{sub:IBRT1-discussion}, as they depend on Section \ref{sec:IB-bifurcations}. 

Following the discussion in Section \ref{sec:coords-exchange-for-the-IB}, there is a subtle disadvantage in choosing decoder coordinates as our variables compared to the other two coordinate systems there.
Indeed, recall that the IB is defined as a maximization over Markov chains $Y \longleftrightarrow X \longleftrightarrow \hat{X}$.
An (arbitrary) encoder $p(\xhat|\x)$ defines a joint probability distribution $p(\xhat|\x) p(\y|\x) p(\x)$ which is Markov.
An inverse-encoder pair also similarly defines a Markov chain. 
In contrast, an arbitrary decoder pair $\big(p(\y|\xhat), p(\xhat)\big)$ need \textit{not} necessarily define a Markov chain.
Rather, by invoking the error analysis of Euler's method, one can see that Markovity is approximated at an increasingly improved quality as the step-size $|\Delta \beta|$ in \eqref{eq:first-order-approx} becomes smaller. 
To enforce Markovity, we shall perform a single BA iteration (in decoder coordinates) after each Euler method step.
This ensures that the newly generated decoder pair satisfies the Markov condition, as it is now generated from an encoder.

As a side effect, adding a single BA-IB iteration after each Euler method step improves the approximation's quality significantly. 
By linearizing $BA_\beta$ around a fixed point, one can show that deterministic annealing with a fixed number of BA iterations per grid point is a first-order method. Thus, deterministic annealing may arguably be considered a first-order method, as is with Euler's method. 
A similar argument shows that adding a single BA iteration after each Euler method step yields a second-order method. 
However, while a larger number of added BA iterations obviously improves the approximation's quality, it does not improve the method's order. 
See Appendix \ref{sec:approx-error-analysis-for-BA-and-Euler-for-IB-appendix} for an approximate error analysis. 
The predicted orders are in good agreement with the ones found empirically, shown in Figure \ref{fig:Euler-method-for-IB-dec}. 
We note that while \citep{agmon2022RTRD} did not attempt an added BA iteration, they do discuss a variety of other improvements to root-tracking (see Section 3.4 there).

\medskip
\section{On IB bifurcations}
\label{sec:IB-bifurcations}

For a bifurcation to exist, it is necessary that the Jacobian of the IB operator \eqref{eq:IB-operator-def} would be singular, as illustrated by Figure \ref{fig:eigenvals-of-BA-IB-jacobian-for-BSC}. 
However, a priori singularity is not sufficient to detect a bifurcation (cf., \cite[Section 3.1]{gedeon2012mathematical}), nor does this allow to distinguish between bifurcations of different types.
In order to be able to exploit the IB's ODE \eqref{eq:IB-beta-ODE-in-decoder-coords} (Section \ref{sec:IB-ODE}), we shall now take a closer look into its bifurcations. 
These can be broadly classified into two types: where an optimal root is continuous in $\beta$ and where it is not. 
As noted after Theorem \ref{thm:IB-ODE}, each type violates an assumption necessary to compute implicit derivatives. 
Sections \ref{sub:continuous-IB-bifs} and \ref{sub:discontinuous-IB-bifs} provide the means to identify bifurcations, distinguish between their types, and handle them accordingly (for continuous bifurcations).
To facilitate the discussion, Section \ref{sub:IB-as-an-RD-problem-and-non-singularity-conj} considers the IB as a rate-distortion problem, following \cite{harremoes2007information} and others. 
This allows us to leverage recent insights on RD bifurcations, \citep{agmon2022RTRD}, while suggesting a ``minimally sufficient'' choice of coordinates for the IB. 
The latter permits a clean treatment of continuous IB bifurcations in Section \ref{sub:continuous-IB-bifs}.
Viewing the IB as an infinite-dimensional RD problem facilitates the understanding of its discontinuous bifurcations, which in turn highlight subtleties in its finite-dimensional coordinate systems (of Section \ref{sec:coords-exchange-for-the-IB}). 
These provide insight into the IB and are also of practical implications (Section \ref{sub:discontinuous-IB-bifs}), and so are necessary for our algorithms at Section \ref{sec:IBRT1-algo}. 

\begin{figure}[h!]
	\centering
	\vspace*{10pt}
	\ifdefined\compilefigs
	\includegraphics[width=.8\textwidth]{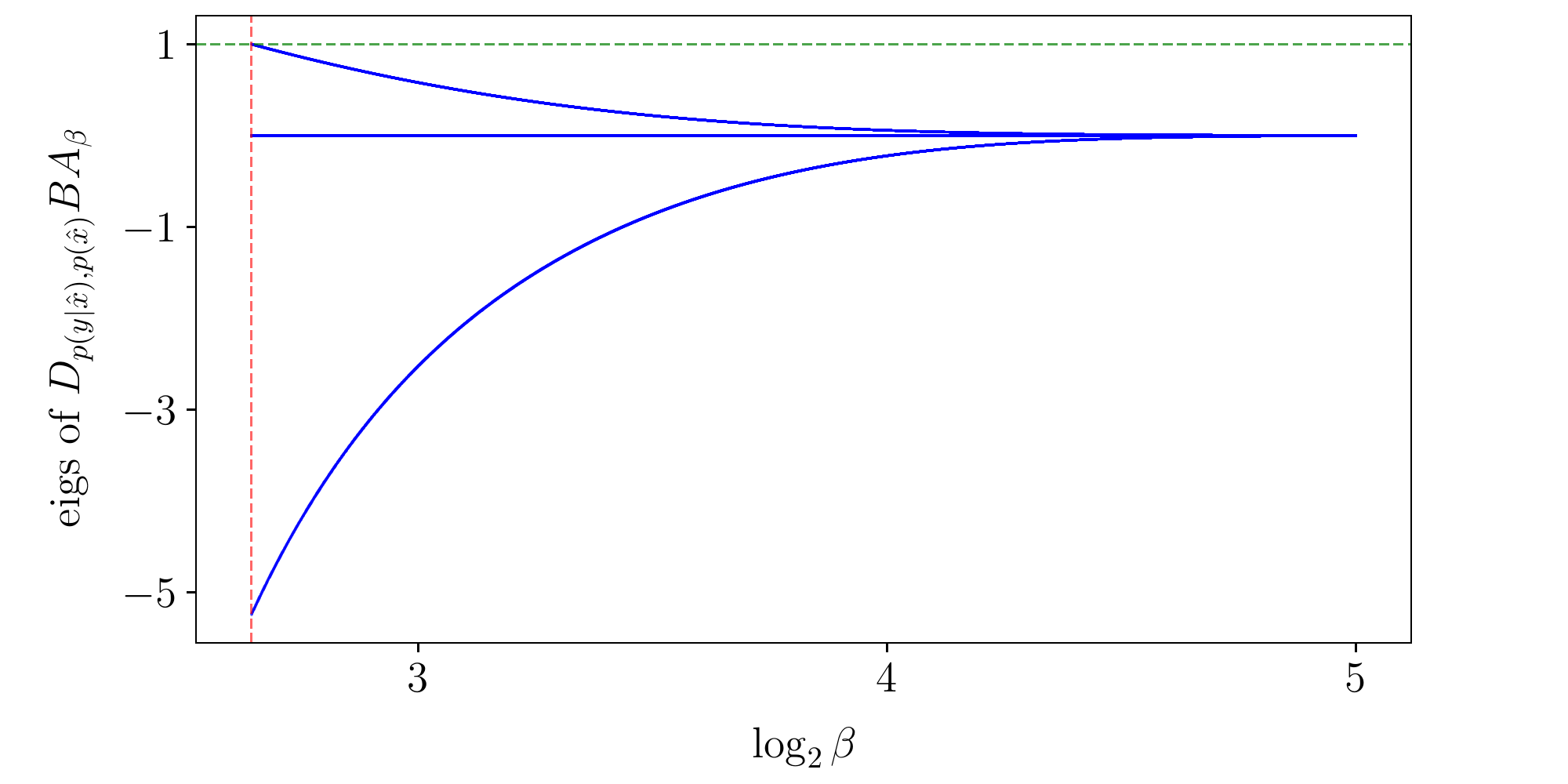}
	\else
	\includegraphics[width=.8\textwidth]{figs/empty_figure}
	\fi
	\caption{
		\textbf{While the Jacobian $D_{\log p(\y|\xhat), \log p(\xhat)} (Id - BA_\beta)$ must be singular at a bifurcation, this does not suffice to identify its type}. 
		The Jacobian eigenvalues of $BA_\beta$ \eqref{eq:BA-Jacob-wrt-decoder-coords-at-main-text} with respect to log-decoder coordinates are plotted for BSC$(0.3)$ with a uniform source, as in Figure \ref{fig:IBRT1-IB-curve-for-several-densitires}; see Appendix \ref{sec:analytical-IB-sol-for-BSC-appendix} for its exact solution.
		An eigenvalue reaches one (dashed green) precisely at the bifurcation (dashed red vertical), as expected by Conjecture \ref{conj:BA-IB-Jacob-in-decoder-coords-is-nonsingular-at-reduced-root} in Section \ref{sub:IB-as-an-RD-problem-and-non-singularity-conj}. 
		In particular, the Jacobian is increasingly ill-conditioned when approaching the bifurcation, as noted in Figure \ref{fig:err-in-numerical-derivs-for-BSC-and-BAs-accuracy} (top). 
		While this allows to detect the bifurcation, identifying its type is necessary for handling it.
	}
	\label{fig:eigenvals-of-BA-IB-jacobian-for-BSC}
\end{figure}

\medskip
\subsection{The IB as a rate-distortion problem}
\label{sub:IB-as-an-RD-problem-and-non-singularity-conj}

We now explore the intimate relation between the IB and RD, following \cite{harremoes2007information} and \cite{bachrach2003}. This leads to a ``minimally sufficient'' coordinate system for the IB, thereby completing the work of Section \ref{sec:coords-exchange-for-the-IB}. 
In this coordinate system, results \citep{agmon2022RTRD} on the dynamics of RD roots are readily considered in IB context. This leads to Conjecture \ref{conj:BA-IB-Jacob-in-decoder-coords-is-nonsingular-at-reduced-root}, that the IB operator \eqref{eq:IB-operator-def} in these coordinates is typically non-singular. 
The discussion here facilitates the treatment of IB bifurcations in the following Sections \ref{sub:continuous-IB-bifs} and \ref{sub:discontinuous-IB-bifs}.

\medskip 
First, recall a few definitions.
A \textit{rate distortion problem} on a \textit{source alphabet} $\mathcal{X}$ and a \textit{reproduction alphabet} $\hat{\mathcal{X}}$ is defined by a \textit{distortion measure}\footnote{ Also called a \textit{single-letter fidelity criterion}. It is a non-negative real-valued function on $\mathcal{X}\times \hat{\mathcal{X}}$, with no further requirements. Or equivalently, an $|\mathcal{X}|$-by-$|\hat{\mathcal{X}}|$ matrix. e.g., \cite[Section 2.2]{berger71}. } $d:\mathcal{X}\times \hat{\mathcal{X}} \to \bb{R}_{\geq 0}$ and a source distribution $p_X(x)$. 
One seeks the minimal rate $I(X; \hat{X})$ subject to a constraint $D$ on the expected distortion $\bb{E}\left[d(\x, \xhat)\right]$, \citep{shannon1948mathematical, shannon1959fidelity},
\begin{equation}			\label{eq:RD-curve-def}
	R(D) := \min_{p(\xhat|\x)} \left\{ I(X; \hat{X}): \; \bb{E}_{p(\xhat|\x) p_X(\x)}\left[d(\x, \xhat)\right] \leq D \right\} \;,
\end{equation}
known as the \textit{rate-distortion curve}.
The minimization is over \textit{test channels} $p(\xhat|\x)$. 
A test channel that attains the RD curve \eqref{eq:RD-curve-def} is called an \textit{achieving distribution}.
We say that an RD problem is \textit{finite} if both of the alphabets $\mathcal{X}$ and $\hat{\mathcal{X}}$ are finite. 
Using Lagrange multipliers for \eqref{eq:RD-curve-def} with\footnote{ Normalization constraints omitted for clarity.} $I(X; \hat{X}) + \beta \; \bb{E}[d(x, \hat{x})]$, one obtains a pair of fixed-point equations
\begin{equation}			\label{eq:RD-fixed-point-eqs}
	p(\xhat|\x) = \frac{p(\xhat) e^{-\beta \; d(\x, \xhat)}}{\sum_{\xhat} p(\xhat) e^{-\beta \; d(\x, \xhat)}}
	\quad \text{and} \quad
	p(\xhat) = \sum_{\x} p(\xhat|\x) p(\x)
\end{equation}
in the marginal $p(\xhat)$ and test channel $p(\xhat|\x)$, similar to the IB Equations \eqref{eq:IB-eq-encoder} and \eqref{eq:IB-eq-marginal}. 
Iterating over these is Blahut's algorithm for RD, \cite{blahut1972}, denoted $BA_\beta^{RD}$ here.
As with the IB \eqref{eq:IB-curve-def}, $\beta$ parametrizes the slope of the optimal curve \eqref{eq:RD-curve-def} also for RD.
See \cite{Cover2006} or \cite{berger71} for an exposition of rate-distortion theory. 

We clarify a definition needed to rewrite the IB as an RD problem. 
We define the \textit{simplex} $\Delta[S]$ on a (possibly infinite) set $S$ as the collection of finite formal convex combinations $\sum_s a_s \cdot s$ of elements of $S$. 
That is, as the $S$-indexed vectors\footnote{ Equivalently, as functions mapping each element $s$ of $S$ to a real number $a_s$. } $(a_s)_{s\in S}$ that satisfy $\sum_s a_s = 1$ and $a_s \geq 0$, with $a_s$ non-zero for only finitely many elements $s$ (the \textit{support} of $(a_s)_s$). 
Addition and multiplication are defined pointwise, as in $\sum_s a_s\cdot s + \sum_s b_s\cdot s = \sum_s (a_s+b_s)\cdot s$. 
$\Delta[S]$ is closed under finite convex combinations because the sum of finitely supported vectors is finitely supported. 
When taking $S = \{\bm{e}_1, \dots, \bm{e}_n\}$ the standard basis vectors $(\bm{e}_i)_j = \delta_{i, j}$ of $\bb{R}^n$, then one can identify the formal operations with those in $\bb{R}^n$, reducing the simplex $\Delta[S]$ to its usual definition. 
We write $r$ for an element of $\Delta[\mathcal{Y}]$. 
In particular, an element of $\Delta\big[\Delta[\mathcal{Y}]\big]$ is merely a finite convex combination $\sum_{\xhat} p(\xhat) r_{\xhat}$ of distinct\footnote{ Note that $\Delta[S]$ is a set.} probability distributions $r_{\xhat}(y) \in \Delta[\mathcal{Y}]$ on $\mathcal{Y}$. 
When setting $\hat{\mathcal{X}} \subset \Delta[\mathcal{Y}]$ to be a finite subset of distributions, $|\hat{\mathcal{X}}| < \infty$, then $\Delta[\hat{\mathcal{X}}]$ is a special case\footnote{ Unlike $\Delta[\hat{\mathcal{X}}]$ here, the decoder coordinates of Section \ref{sec:coords-exchange-for-the-IB} are \textit{not} required to have their clusters $r$ distinct.} of the decoder coordinates of Section \ref{sec:coords-exchange-for-the-IB}. 

Now, let a finite IB problem be defined by a joint probability distribution $p_{Y|X} \; p_X$, as in Section \ref{sec:introduction}.
To write it down as an RD problem, \citep{harremoes2007information, bachrach2003}, define the \textit{IB distortion measure} by
\begin{equation}			\label{eq:d_IB}
	d_{IB}(x, r) := D_{KL}\left[ p_{Y|X=x} || r \right] \;,
\end{equation}
for $x\in \mathcal{X}$, $r \in \Delta[\mathcal{Y}]$, and $p_{Y|X=x}\in \Delta[\mathcal{Y}]$ the conditional probability distribution at $X = x$. 
The distortion measure $d_{IB}$ \eqref{eq:d_IB} and $p_X$ define an RD problem on the continuous reproduction alphabet $\hat{\mathcal{X}} := \Delta[\mathcal{Y}]$. 
Minimizing the IB Lagrangian $\mathcal{L}$ (at Section \ref{sec:introduction}) is equivalent to minimizing the Lagrangian of this RD problem, \citep[Theorem 5]{harremoes2007information}. 
That is, the IB is a rate-distortion problem when considered in these coordinates\footnote{ 
	The astute reader might notice that the IB Equations \eqref{eq:IB-eq-encoder} and \eqref{eq:IB-eq-marginal} are then equivalent to RD's fixed-point Equations \eqref{eq:RD-fixed-point-eqs}, with \eqref{eq:IB-eq-decoder} implied by the IB's Markovity. 
	The IB's $Y$-information $I(Y; \hat{X})$ equals the expected distortion $\bb{E} \left[d_{IB}(\x, \xhat)\right]$ at \eqref{eq:RD-curve-def} up to a constant, \citep[Section 5]{harremoes2007information}, and so is linear in the test channel $p(r|\x)$.
}. 
IB clusters $r\in \Delta[\mathcal{Y}]$ assume the role of RD reproduction symbols, while an IB root (considered now as an RD root) is equivalently described either by the probabilities of each cluster --- namely, by a point in $\Delta\big[\Delta[\mathcal{Y}]\big]$ --- or, by a test channel $p(r|\x)$.
Unlike the finite-dimensional coordinate systems of Section \ref{sec:coords-exchange-for-the-IB}, this definition of the IB entails no subtleties due to finite-dimensionality, such as duplicate clusters (see more below). 
However, while it allows to spell out the IB explicitly as an RD problem, handling an infinite reproduction alphabet is difficult for practical purposes. 
Since no more than $|\mathcal{X}| + 1$ reproduction symbols are needed to write down an IB root (see therein), this motivates one to consider the IB's \textit{local} behavior, with clusters fixed.

So instead, one may require the reproduction symbols of $d_{IB}$ \eqref{eq:d_IB} to be in a list\footnote{ We use here a tuple $(r_{\xhat_1}, \dots, r_{\xhat_T})$ rather than a set, since the points $r_{\xhat}$ need \textit{not} be distinct a priori.}  $\left(r_{\xhat}\right)_{\xhat}$ indexed by some finite set $\hat{\mathcal{X}}$, with each $r_{\xhat}$ in $\Delta[\mathcal{Y}]$. 
This defines a finite RD problem, for which $d_{IB}$ \eqref{eq:d_IB} is merely an $|\mathcal{X}|$-by-$T$ matrix.
Yet, placing identical clusters in the list $\left(r_{\xhat}\right)_{\xhat}$ inadvertently introduces degeneracy to the matrix $d_{IB}$ \eqref{eq:d_IB}, as discussed below. 
\cite[Section 6]{bachrach2003} take $\left(r_{\xhat}\right)_{\xhat}$ to be the decoders defined by a given encoder $p(\xhat|\x)$, as in Equation \eqref{eq:coordinate-sets-parameterizing-an-IB-root} (Section \ref{sec:coords-exchange-for-the-IB}). We shall then refer to $d_{IB}$ \eqref{eq:d_IB} as \textit{the distortion matrix defined by $p(\xhat|\x)$}. 
When $p_{\beta_0}(\xhat|\x)$ is an optimal IB root then the problem $(d_{IB}, p_X)$ defined by it is called \textit{the tangent RD problem}. 
Indeed, its RD curve \eqref{eq:RD-curve-def} coincides\footnote{ Note that an optimal choice of IB clusters is already encoded into $d_{IB}$ \eqref{eq:d_IB} here. With $d_{IB}$ and $p_X$ given, solving the IB boils down to finding the optimal cluster weights $p(\xhat)$, which is an RD problem. } with the IB curve \eqref{eq:IB-curve-def} at this point. However, the curves differ outside this point since IB clusters usually vary with $\beta$, while the distortion of the tangent problem is defined at $p_{\beta_0}(\xhat|\x)$ and so is fixed. 
By definition \eqref{eq:IB-curve-def}, it follows that the IB curve is a lower envelope of the curves of its tangent RD problems, \citep[Corollary 2]{bachrach2003}.
We note that a similar construction can also be carried out in inverse-encoder coordinates; cf., \citep{witsenhausen1975conditional}.

Regardless of the formulation used to rewrite the IB as an RD problem, the associated RD problem has an expected distortion $\bb{E}[d_{IB}]$ of $I(X; Y) - I(\hat{X}; Y)$ at an IB root, \citep[Section 5]{harremoes2007information}, \citep[Lemma 8]{bachrach2003}. 
That is, the IB is a lossy compression method that strives to preserve the relevant information $I(\hat{X}; Y)$.  
Due to the Markov condition, information on $Y$ is available only through $X$. 
Thus, one may intuitively consider the IB as a lossy compression method of the information on $Y$ that is embedded in $X$. 
These intimate relations between the IB and RD suggest that studying bifurcations in either context could be leveraged to understand the other. 
Bifurcations in finite RD problems are discussed at length in \cite[Section 6]{agmon2022RTRD}.
To facilitate the study of IB bifurcations in the sequel (Sections \ref{sub:continuous-IB-bifs} and \ref{sub:discontinuous-IB-bifs}) using results from RD, we need a ``minimally-sufficient'' coordinate system for the IB.

Consider an IB root in decoder coordinates as finitely-many $p(\xhat)$-weighted points $r_{\xhat}(\y)$ in $\Delta[\mathcal{Y}]$, as in Section \ref{sec:coords-exchange-for-the-IB}.
Exchanging \textit{to} decoder coordinates (Equation \eqref{eq:coordinate-sets-parameterizing-an-IB-root} there) is well-defined so long that there are no zero mass clusters, $\forall \xhat \; p(\xhat) \neq 0$. 
Yet, even then, the points $r_{\xhat}$ in $\Delta[\mathcal{Y}]$ yielded by Equations \algref{algo:BA-IB}{eq:IB-BA-cluster_marginal} through \algref{algo:BA-IB}{eq:IB-BA-decoder-eq} need \textit{not} be distinct. 
Namely, they may yield identical clusters $r_{\xhat} = r_{\xhatp}$ at distinct indices $\xhat \neq \xhatp$.
This leads to a discussion of structural symmetries of the IB (its degeneracies), which is not of use for our purposes; cf., \cite{gedeon2012mathematical}.
To avoid such subtleties, we shall say that an IB root is \textit{reduced} if it has no zero-mass clusters, $\forall \xhat \; p(\xhat) \neq 0$, and all its clusters are distinct, $\xhat = \xhatp \Leftrightarrow r_{\xhat} = r_{\xhatp}$. 
A root that is not reduced is \textit{degenerate} or \textit{represented degenerately}. 
An IB root can be reduced by removing clusters of zero mass and merging identical clusters of distinct indices --- see Algorithm \ref{algo:root-reduction} in Section \ref{sub:continuous-IB-bifs} below. 
It is straightforward to see from the IB Equations \eqref{eq:IB-eq-encoder}-\eqref{eq:IB-eq-marginal} that reduction preserves the property of being an IB root.
Similarly, reducing a root does not change its location in the information plane. 
So, a root achieves the IB curve \eqref{eq:IB-curve-def} if and only if its reduction does.
Therefore, reduction decreases the dimension in which the problem is considered while preserving all its essential properties.
This allows to represent an IB root on the smallest number of clusters possible --- its \textit{effective cardinality} --- by factoring out the IB's structural symmetries. cf., \citep[2.3 in Chapter 7]{zaslavsky2019thesis}, upon which this definition is based. 

While the purpose of reduction is to mod-out redundant kernel coordinates (see Section \ref{sec:introduction}), it highlights the differences between the various IB definitions found in the literature\footnote{ e.g., both of the IB formulations \citep{harremoes2007information, witsenhausen1975conditional} do not impose an a priori restriction on the number of clusters. The former does not enable one to encode duplicate clusters, while the latter does. The formulation \citep{tishby1999} ignores these subtleties altogether. \citep{gedeon2012mathematical, parker2022symmetry_breaking} consider the IB on a pre-determined number of possibly duplicate clusters. }, bringing to light a subtle caveat of finite dimensionality. 
To see this, note that reduction could have been defined above in terms of the other coordinate systems of the IB. 
Its definition in inverse-encoder coordinates is nearly identical to that above, while defining it in encoder coordinates is a straightforward exercise. 
Since the coordinate systems of Section \ref{sec:coords-exchange-for-the-IB} are equivalent at an IB root (without zero-mass clusters), the precise definition does not matter then. 
Each of these parameterizations encodes the coordinates $r(y)$ of a root's clusters $r$ using a finite-dimensional vector $\bm{x}$ (note Equation \eqref{eq:coordinate-sets-parameterizing-an-IB-root}). 
This enables one to represent duplicate clusters $\xhat \neq \xhatp$ with $r_{\xhat} = r_{\xhatp}$, and obliges one to choose the order in which clusters are being encoded into the coordinates of $\bm{x}$. 
A finite-dimensional representation $\bm{x}$ of an IB root is invariant to interchanging clusters $\xhat \neq \xhatp$ precisely when they are identical, $r_{\xhat} = r_{\xhatp}$. 
The IB's functionals (e.g., its $X$ and $Y$-information) are invariant to any cluster permutation. 
cf., \citep{gedeon2012mathematical, parker2022symmetry_breaking}. 
Both of these structural symmetries result from using a finite-dimensional parameterization, with the former eliminated by reduction. 
In contrast, the elements of $\Delta\left[\mathcal{Y}\right]$ are distinct by definition (since $\Delta\left[\mathcal{Y}\right]$ is a set), and so parametrizing the IB by points in $\Delta\left[\Delta[\mathcal{Y}]\right]$ does not permit identical clusters. 
An element $\sum_r p(r) r$ of $\Delta\left[\Delta[\mathcal{Y}]\right]$ assigns a probability mass $p(r)$ to \textit{every} point $r$ in $\Delta[\mathcal{Y}]$, with only finitely-many points $r$ supported. 
Thus, it implicitly encodes all the entries $r(y)$ of every probability distribution $r \in \Delta[\mathcal{Y}]$ in a ``one size fits all'' approach, giving no room for the choices above. 
This leads us to argue that the IB's structural symmetries are \textit{not} an inherent property but rather an artifact of using its finite-dimensional representations. 

In rate-distortion, the \textit{reduction} of a finite RD problem is defined similarly, \cite[Section 3.1]{agmon2022RTRD}, by removing a symbol $\xhat$ from the reproduction alphabet $\hat{\mathcal{X}}$ and its column $d(\cdot, \xhat)$ from the distortion matrix once it is not in use anymore (of zero mass).
A distortion matrix $d$ is \textit{non-degenerate} if its columns are distinct, $d(\cdot, \xhat) \neq d(\cdot, \xhatp)$ for all $\xhat \neq \xhatp$.
Non-degeneracy arises naturally when considering the RD problem tangent to a given IB root $p(\xhat|\x)$.
Indeed, the distortion matrix $d_{IB}$ \eqref{eq:d_IB} defined by $p(\xhat|\x)$ has duplicate columns if the root has identical clusters, while the other direction holds under mild assumptions\footnote{ If the $|\mathcal{X}|$ vectors $p_{Y|X=x}$ span $\bb{R}^{|\mathcal{Y}|}$, then $D_{KL}\left[ p_{Y|X=x} || r_{\xhat} \right] = D_{KL}\left[ p_{Y|X=x} || r_{\xhatp} \right]$ for all $x$ implies that $r_{\xhat} = r_{\xhatp}$.}.
Under these assumptions, the distortion matrix induced by an IB root $p(\xhat|\x)$ is reduced and non-degenerate precisely when $p(\xhat|\x)$ is a reduced IB root.

Reduction in RD provides the means to show that the dynamics underlying the RD curve \eqref{eq:RD-curve-def} are piecewise analytic in $\beta$, \citep{agmon2022RTRD}, under mild assumptions.
Just as in definition \eqref{eq:IB-operator-def} of the IB operator, \cite[Eq. (5)]{agmon2021critical} similarly define the \textit{RD operator} $Id - BA_\beta^{RD}$ in terms of Blahut's algorithm for RD, \citep{blahut1972}. 
By using their Theorem 1, \cite[Section 3.1]{agmon2022RTRD} observed that reducing a finite RD problem to the support\footnote{ The \textit{support} of a distribution $p(\xhat)$ is defined by $\supp p(\xhat) := \{\xhat : \; p(\xhat) > 0\}$. } of a given RD root mods-out redundant kernel coordinates if the distortion measure is finite and non-degenerate. 
That is, the Jacobian $D(Id - BA_\beta^{RD})$ of the RD operator on the reduced problem is then non-singular\footnote{ When considered in the right coordinates system; see therein for details.}, just as with our toy problem \eqref{eq:toy-example-operator-def} in Section \ref{sec:introduction}.
By the Implicit Function Theorem, there is therefore a unique RD root of the reduced problem through the given one; this root is real-analytic in $\beta$ (details there).
Considering this for the RD problem tangent to a reduced IB root immediately yields the following, 

\begin{cor}			\label{cor:RD-root-through-of-tangent-prob-at-reduced-IB-root-is-analytic}
	Let $p_{\beta_0}(\xhat|\x)$ be a reduced IB root of a finite IB problem defined by $p_{Y|X} \; p_X$, such that the matrix $p_{Y|X}$ is of rank $|\mathcal{Y}|$.
	Then, near $\beta_0$, there is a unique function continuous in $\beta$, which is a root of the tangent RD problem through $p_{\beta_0}(\xhat|\x)$; it is real-analytic in $\beta$.
\end{cor}

Corollary \ref{cor:RD-root-through-of-tangent-prob-at-reduced-IB-root-is-analytic} shows that the local approximation of an IB problem (the roots of its tangent RD problem) is guaranteed to be as well-behaved as one could hope for, provided that the IB is viewed in the right coordinates system.
Note, however, that the RD root through $p_{\beta_0}(\xhat|\x)$ of the tangent problem does \textit{not} in general coincide with the IB root outside of $\beta_0$ since the IB distortion $d_{IB}$ \eqref{eq:d_IB} varies along with the clusters that define it.
However, when the IB clusters are fixed, then one might expect that the Jacobian \eqref{eq:BA-Jacob-wrt-decoder-coords-at-main-text} of $BA_\beta$ in log-decoder coordinates would be the same as the Jacobian of its RD variant.
Indeed, the Jacobian matrix of $BA^{RD}_\beta$ is the $T\times T$ bottom-right sub-block of the Jacobian \eqref{eq:BA-Jacob-wrt-decoder-coords-at-main-text} of $BA_\beta$, up to a multiplicative factor.
cf., Equations (5)-(6) in \citep{agmon2021critical}, Equations \eqref{eq:B_defs-for-BA-Jacob-wrt-decoder-coords} and \eqref{eq:BA-Jacob-wrt-decoder-coords-at-main-text} in Section \ref{sec:IB-ODE}, and \eqref{eq:BA-Jacob-wrt-decoder-coords-as-block-matrix-implicit} in Appendix \ref{sub:BA-IB-jacobian-appendix:decoder-deriv-matrix}. 

As in RD, we argue that reduction in the IB also provides the means to show that the dynamics underlying the optimal curve \eqref{eq:IB-curve-def} are piecewise analytic in $\beta$. 
Corollary \ref{cor:RD-root-through-of-tangent-prob-at-reduced-IB-root-is-analytic} concludes that, under mild assumptions, through every reduced IB root passes a unique real-analytic RD root. 
However, its crux is that the Jacobian of the RD operator $Id - BA_\beta^{RD}$ is non-singular at a reduced root. 
Due to the IB's close relations with RD, and since reduction in the IB is a natural extension of reduction in RD, we argue that the same is also to be expected of the IB operator $Id - BA_\beta$ \eqref{eq:IB-operator-def} in decoder coordinates. 
To see this, note that IB roots are finitely supported, \citep[Lemma 2.2(i)]{witsenhausen1975conditional}, and so one may take finitely-supported probability distributions $\Delta\big[\Delta[\mathcal{Y}]\big]$ for the IB's optimization variable. 
Thus, the IB's $BA_\beta$ operator in decoder coordinates (of Section \ref{sec:coords-exchange-for-the-IB}) may be considered as an operator on $\Delta\big[\Delta[\mathcal{Y}]\big]$. 
Next, consider the RD problem defined by $p_X$ and $d_{IB}$ \eqref{eq:d_IB} on the continuous reproduction alphabet $\Delta[\mathcal{Y}]$, as in \citep{harremoes2007information}. 
This defines on $\Delta\big[\Delta[\mathcal{Y}]\big]$ also the BA operator $BA_\beta^{RD}$ for RD. 
Now that both BA operators are considered on an equal footing, we note the following. 
First, while $BA_\beta^{RD}$ iterates over\footnote{ To see this, plug the IB distortion measure $d_{IB}$ \eqref{eq:d_IB} into the Equations \eqref{eq:RD-fixed-point-eqs} defining $BA_\beta^{RD}$.} the IB Equations \eqref{eq:IB-eq-encoder} and \eqref{eq:IB-eq-marginal}, its IB variant $BA_\beta$ iterates also over the decoder Equation \eqref{eq:IB-eq-decoder}. 
The latter Equation is a necessary condition\footnote{ For comparison, only $p(\y|\xhat) = \sum_x p(\y|\x, \xhat) p(\x|\xhat)$ holds for an arbitrary triplet $(Y, X, \hat{X})$ of random variables.} for $Y \to X \to \hat{X}$ to be Markov, and so can be understood as enforcement of Markovity. 
That is, IB roots are RD roots with an extra constraint. 
Second, by \cite[Theorem 1]{agmon2021critical}, reducing $Id - BA_\beta^{RD}$ from the continuous reproduction alphabet $\Delta[\mathcal{Y}]$ to a root of finite support renders it non-singular, under mild assumptions. 
Due to the similarity between these operators, and since reduction in the IB is a natural extension of reduction in RD, this suggests that reducing $Id - BA_\beta$ \eqref{eq:IB-operator-def} from $\Delta\big[\Delta[\mathcal{Y}]\big]$ to a root's effective cardinality should also render it non-singular. 
In line with the discussion of Section \ref{sec:introduction} on reduction, we therefore state the following, 

\begin{conj}			\label{conj:BA-IB-Jacob-in-decoder-coords-is-nonsingular-at-reduced-root}
	The Jacobian matrix $I - D_{\log p(\y|\xhat), \log p(\xhat)} BA_\beta$ at \eqref{eq:IB-beta-ODE-in-decoder-coords} of the IB operator \eqref{eq:IB-operator-def} in log-decoder coordinates is non-singular at reduced IB roots so long that it is well-defined, except perhaps at points of bifurcation.
\end{conj}

The intuition behind this conjecture stems from analyticity, as follows.
The IB operator $Id - BA_\beta$ \eqref{eq:IB-operator-def} is real-analytic, since each of the Equations \algref{algo:BA-IB}{eq:IB-BA-cluster_marginal}-\algref{algo:BA-IB}{eq:IB-BA-new-direct-enc} defining it (in the BA-IB Algorithm \ref{algo:BA-IB}) is real-analytic in its variables. 
For a root $\bm{x}_0$ of a real-analytic operator $F$, one might expect that, in general, (i) no roots other than $\bm{x}_0$ exist in its vicinity and that (ii) $D_{\bm{x}} F\rvert_{\bm{x}_0}$ has no kernel.
That is, unless the operator is degenerate at $\bm{x}_0$ in some manner or $\bm{x}_0$ is a bifurcation. 
To see this, recall \cite[IX.3]{dieudonne1969foundations} that a real-valued function $F_i$ in $\bm{x} \in \bb{R}^n$ is \textit{real analytic} in some open neighborhood of $\bm{x}_0$ if it is a power series in $\bm{x} = (x_1, \dots, x_n)$, within some\footnote{ While a strictly positive radius of convergence is needed here, we omit these details for clarity.} radius of convergence. 
For every practical purpose, one may replace $F_i$ by a polynomial in $(x_1, \dots, x_n)$ when $\bm{x}$ is close enough to the base-point $\bm{x}_0$, by truncating the power series. 
Viewed this way, a root of an operator $F(\bm{x}) = \big(F_1(\bm{x}), \dots, F_n(\bm{x})\big)$ is nothing but a solution of $n$ polynomial equations in $n$ variables. However, a square polynomial system typically has only isolated roots, which is (i). 
This is best understood in terms of B{\'e}zout's Theorem; see \cite[6 in IV.4]{gowers2008princeton} for example.
For (ii), a vector $\bm{v}$ is in $\ker D_{\bm{x}} F$ precisely when it is orthogonal to each of the gradients $\nabla F_i$.
However, $\nabla F_i$ is the vector of the first-order monomial coefficients of $x_1, \dots, x_n$ in $F_i$.
In a general position, these $n$ coefficient vectors $\nabla F_1, \dots, \nabla F_n$ are linearly-independent, and so $\bm{v}$ must vanish as claimed.
If $F$ is degenerate such that $F_i = F_j$ for particular $i \neq j$, for example, then both points fail, of course. 
cf., also \cite[I.2]{coolidge1959treatise} for (i) and (ii). 
This intuition accords with the comments of \cite[Section 2.4]{berger71} on RD: ``\textit{usually, each point on the rate distortion curve [...] is achieved by a unique conditional probability assignment. However, if the distortion matrix exhibits certain form of symmetry and degeneracy, there can be many choices of [a minimizer]}''. 
Indeed, the fact that the dynamics underlying the RD curve \eqref{eq:RD-curve-def} are piecewise real-analytic (under mild assumptions), \citep{agmon2022RTRD}, can be similarly understood to stem from the analyticity of the RD operator $Id - BA_\beta^{RD}$. 

Subject to Conjecture \ref{conj:BA-IB-Jacob-in-decoder-coords-is-nonsingular-at-reduced-root}, a Jacobian eigenvalue of the IB operator \eqref{eq:IB-operator-def} must vanish gradually\footnote{ Note that BA's Jacobian \eqref{eq:BA-Jacob-wrt-decoder-coords-at-main-text} is continuous in the root at which it is evaluated.} as one approaches a bifurcation, causing the critical slowing down of the BA-IB Algorithm \ref{algo:BA-IB}, \citep{agmon2021critical}.
When an IB root traverses a bifurcation in which its effective cardinality decreases, then it is not reduced anymore.
One can then handle the bifurcation by reducing the root anew.
To ensure proper handling by the bifurcation's type, we consider the latter closely in Sections \ref{sub:continuous-IB-bifs} and \ref{sub:discontinuous-IB-bifs} below.
In a nutshell, following the IB's ODE \eqref{eq:IB-beta-ODE-in-decoder-coords} along with proper handling of its bifurcations is the idea behind our Algorithm \ref{algo:IBRT1} (Section \ref{sec:IBRT1-algo}), for approximating the IB numerically.

Conjecture \ref{conj:BA-IB-Jacob-in-decoder-coords-is-nonsingular-at-reduced-root} is compatible with our numerical experience. However, we leave its proof to future work. 
To that end, one could examine closely the smaller matrix $S$ \eqref{eq:smaller-matrix-for-ker-of-IB-operator-in-dec-coords-main-text} (of Lemma \ref{lem:smaller-matrix-for-ker-of-J-in-decoder-coords} in Section \ref{sec:IB-ODE}) for example. 
However, even if Conjecture \ref{conj:BA-IB-Jacob-in-decoder-coords-is-nonsingular-at-reduced-root} were violated, then one could detect that easily by inspecting the Jacobian's eigenvalues. 
Conjecture \ref{conj:BA-IB-Jacob-in-decoder-coords-is-nonsingular-at-reduced-root} also implies that IB roots are locally unique outside of bifurcations when presented in a reduced form. Non-uniqueness of optimal roots is detectable by inspecting the Jacobian's eigenvalues --- see Corollary \ref{cor:convexity-of-IB-optimal-roots-at-beta} in Section \ref{sub:discontinuous-IB-bifs} and the discussion following it. 
cf., \cite[Section 6.3]{agmon2022RTRD} for the respective discussion in RD. 
With that, most of the results in Sections \ref{sub:continuous-IB-bifs} and \ref{sub:discontinuous-IB-bifs} below do not depend on the validity of Conjecture \ref{conj:BA-IB-Jacob-in-decoder-coords-is-nonsingular-at-reduced-root}.

\medskip
\subsection{Continuous IB bifurcations: cluster-vanishing and cluster-merging}
\label{sub:continuous-IB-bifs}

Following \cite{agmon2022thesis}, we consider the evolution of IB roots which are a continuous function of $\beta$.
By representing an IB root in its reduced form (Section \ref{sub:IB-as-an-RD-problem-and-non-singularity-conj}), it is evident that there are two types of \textit{continuous} IB bifurcations. 
We provide a practical heuristic (Algorithm \ref{algo:root-reduction}) for identifying and handling such bifurcations.
The discussion here is complemented by Subsection \ref{sub:discontinuous-IB-bifs} below, which considers the case where continuity does not hold.

\medskip
The evolution of an IB root in $\beta$ obeys the ODE \eqref{eq:IB-beta-ODE-in-decoder-coords} so long that it can be written as a differentiable function in $\beta$, as in Theorem \ref{thm:IB-ODE}.
Considering the root in decoder coordinates, this amounts to an evolution of a $T$-tuple of points $r_{\xhat}$ in $\Delta[\mathcal{Y}]$ and their weights $p(\xhat)$. These typically traverse the simplex smoothly as the constraint $\beta$ is varied, as demonstrated in Figure \ref{fig:exact-sol-of-BSC}.
We now consider two cases where this evolution does not obey the ODE \eqref{eq:IB-beta-ODE-in-decoder-coords}, due to violating differentiability.

\begin{figure}[h!]
	\centering
	\vspace*{10pt}
	\ifdefined\compilefigs
	\includegraphics[width=.75\textwidth]{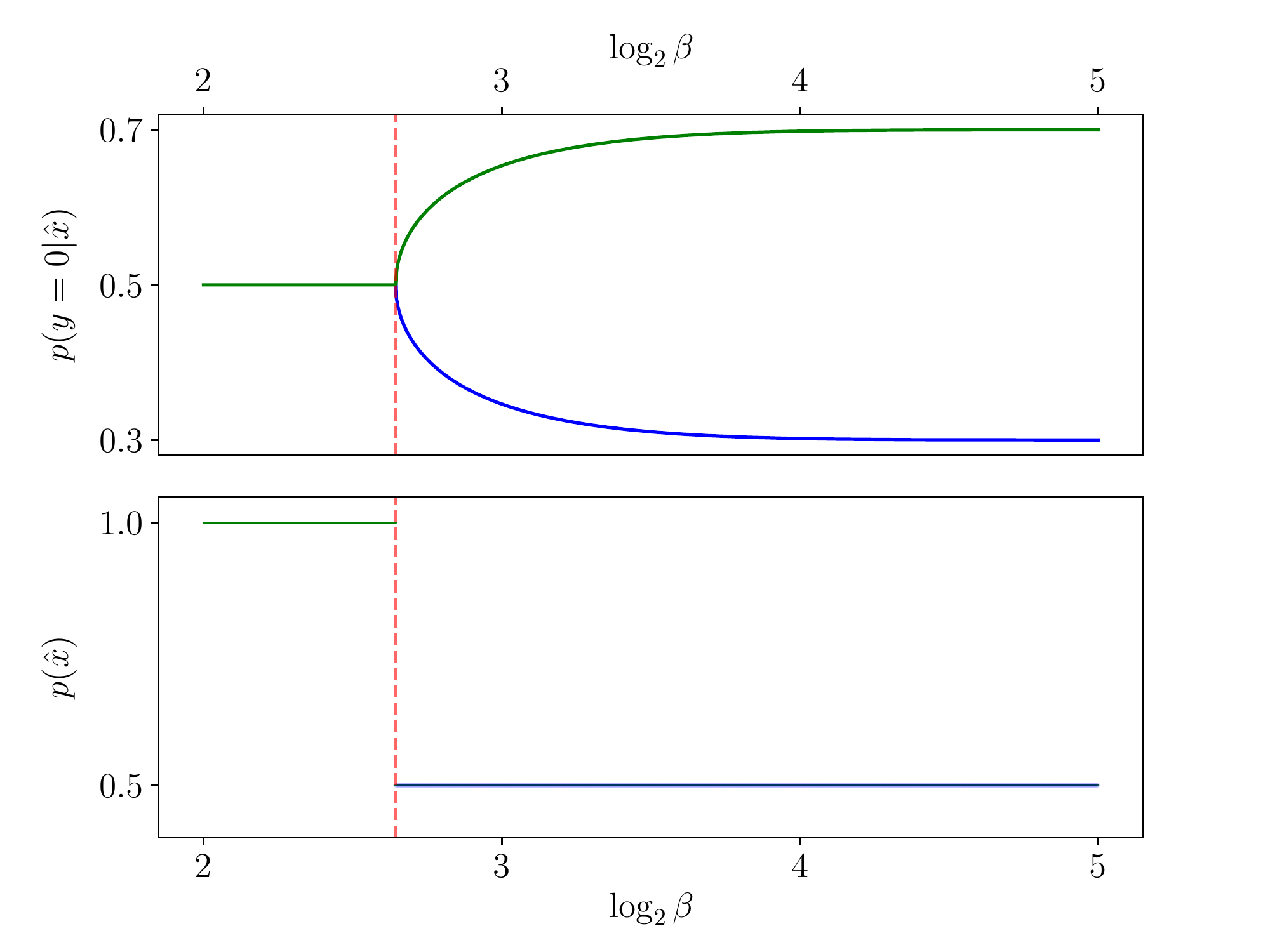}
	\else
	\includegraphics[width=.75\textwidth]{figs/empty_figure}
	\fi
	\caption{
		\textbf{A cluster-merging bifurcation.} The reduced form of the optimal IB root in decoder coordinates as a function of $\beta$, for the exact solution of BSC$(0.3)$ with a uniform source, as in Figure \ref{fig:IBRT1-IB-curve-for-several-densitires} (see Appendix \ref{sec:analytical-IB-sol-for-BSC-appendix}).
		At high enough $\beta$, the root is comprised of two clusters (in green and blue), each of a marginal probability $\tfrac{1}{2}$. The clusters collide at $\beta_c = 6\nicefrac{1}{4}$ (dashed red vertical) and merge to one, yielding the trivial solution --- a single cluster of probability 1 at $p_Y$. 
		Carefully note that only a \textit{single} IB root is plotted here, in its reduced form, with one cluster to the left of $\beta_c$ and two to the right. 
		The violation of clusters' differentiability at $\beta_c$ can be observed visually (top), and the root is otherwise real-analytic in $\beta$, as can be deduced from Figure \ref{fig:eigenvals-of-BA-IB-jacobian-for-BSC}.
		Since the trivial solution is an IB root for every $\beta > 0$ (not shown), then $\beta_c$ is indeed a bifurcation, where the trivial and non-trivial roots intersect. 
		To see this, consider the degenerate form of the trivial solution on \textit{two} copies of $p_Y$, each of probability $\tfrac{1}{2}$. 
		The marginals $p(\xhat)$ (bottom) appear to be discontinuous at $\beta_c$ because the root was reduced before plotted (the latter degenerate form of the trivial root is \textit{not} plotted to the left of $\beta_c$). 
	}
	\label{fig:exact-sol-of-BSC}
\end{figure}

Consider an optimal IB root in its reduced form (see Section \ref{sub:IB-as-an-RD-problem-and-non-singularity-conj}). Namely, consider the reduced form of a root that achieves the IB curve \eqref{eq:IB-curve-def}.
Suppose that its decoders $r_{\xhat}$ and weights $p(\xhat)$ are continuous in $\beta$. 
Then, a qualitative change in the root can occur only if either (i) two (or more) of its clusters collide or (ii) the marginal probability $p(\xhat)$ of a cluster $\xhat$ vanishes. In either case, the minimal number of points in $\Delta[\mathcal{Y}]$ required to represent the root decreases. That is, its effective cardinality decreases\footnote{ A qualitative change where a root's effective cardinality increases is obtained by merely reversing the dynamics in $\beta$ of (i) or (ii) above.}.
We call the first a \textit{cluster-merging bifurcation} and the second a \textit{cluster-vanishing bifurcation}. Or, \textit{continuous bifurcations} collectively.
Both types were observed already at \cite[Section IV.C]{rose1994mapping} in the related setting of RD problems with a continuous source alphabet. 

At a continuous bifurcation, IB roots of distinct effective cardinalities collide and merge into one, as discussed in Section \ref{sub:discontinuous-IB-bifs} below. 
Specifically, one root achieves the minimal value of the IB Lagrangian and so is stable, while the other root is sub-optimal.
Thus, continuous IB bifurcations are \textit{pitchfork bifurcations}\footnote{ 
	Strictly speaking, several copies of the root of larger effective cardinality collide at a continuous bifurcation. 
	When two clusters $r \neq r'$ collide in (i), the root itself is invariant to interchanging their coordinates after the collision but not before it, breaking the IB's first structural symmetry discussed in Subsection \ref{sub:IB-as-an-RD-problem-and-non-singularity-conj}. 
	Interchanging the coordinates of $r$ and $r'$ (and their marginals) before the collision yields two distinct copies of essentially the same root. 
	For (ii), the IB's functionals (e.g., its $X$ and $Y$-information) do not depend on the coordinates $\left( r(y) \right)_y$ of a vanished cluster $r$, rendering these redundant. cf., \citep[Section 3.1]{gedeon2012mathematical}. 
	Before the cluster $r$ vanishes, there is one copy of the root for each index $\xhat$, with $r$ placed at its $\xhat$ coordinates. 
	Considered in reduced coordinates, these coincide to a single copy after the cluster vanishes. 
	This breaks the second structural symmetry. 
}, e.g., \cite[Section 3.4]{strogatz2018nonlinear}, in accordance with \cite{parker2022symmetry_breaking}. 
Even though the optimal root is continuous in $\beta$ (by assumption), its differentiability is lost at the point of bifurcation, as noted after Theorem \ref{thm:IB-ODE} and demonstrated in Figure \ref{fig:exact-sol-of-BSC}. 
Among the two, cluster-vanishing bifurcations are more frequent in practice than cluster-merging. This can be understood by considering cluster trajectories in the simplex. In a general position, one might expect clusters to seldom be at the same ``time'' and place ($\beta$ and $r\in \Delta[\mathcal{Y}]$). 

With that, we note that cluster vanishing bifurcations \textit{cannot} be detected directly by standard local techniques (i.e., considering the derivative's kernel directions at the bifurcation), whether considering the Hessian of the IB's loss function as in \citep{gedeon2012mathematical} or the Jacobian of the IB operator \eqref{eq:IB-operator-def} as here. 
The technical reason for this is as follows, while the root cause underlying it is discussed in Subsection \ref{sub:discontinuous-IB-bifs} (after Proposition \ref{prop:bif-is-detectable-only-on-enough-clusters}). 
Observe that the $I(Y; \hat{X})$ and $I(X; \hat{X})$ functionals do not depend on the coordinates $\left( r(y) \right)_y$ of clusters $r$ of zero mass. 
Thus, the directions corresponding to these coordinates are always in the kernel regardless of whether evaluating at a bifurcation or not, and so cannot be used to detect a bifurcation\footnote{ The direction corresponding to a cluster's marginal is useless when one does not know which cluster $\left( r(y) \right)_y$ to pick.}. 
Indeed, with its dynamics in $\beta$ reversed, ``a new symbol grows continuously from zero mass'' in a cluster-vanishing bifurcation, as \cite[IV.C]{rose1994mapping} comments in a related setting. 
It is then not clear a priori which point in $\Delta[\mathcal{Y}]$ should be chosen for the new symbol, rendering the perturbative condition at Equation \eqref{eq:second-variational-deriv-of-Lagrangian} difficult to test. 
In accordance with this, \cite[Section 5]{gedeon2012mathematical} offers a perturbative condition for detecting arbitrary IB bifurcations, while \cite[3.2 in Part III]{zaslavsky2019thesis} offers a condition for detecting cluster-merging bifurcations by analyzing cluster stability. 
However, both conditions are equivalent (Appendix \ref{sec:equivalent-conds-for-cluster-splitting-appendix}), and so must detect the same type of bifurcations. 
In contrast, a cluster-splitting (or merging) bifurcation is straightforward to detect because the stability of a particular cluster $\xhat$ is a property of the root itself --- see Appendix \ref{sec:equivalent-conds-for-cluster-splitting-appendix} and the references therein for details. 

One may wonder whether bifurcations exist in the IB for the same reason as they do in RD.
As in the IB, RD problems typically have many sub-optimal curves, \cite[Section 6.1]{agmon2022RTRD}. 
While bifurcations in the IB stem from restricting the effective cardinality\footnote{ At least for continuous IB bifurcations.}, \cite[Section 3.4]{tishby1999}, in RD they stem from the various restrictions that a reproduction alphabet has. e.g., a reproduction alphabet $\hat{\mathcal{X}} := \{r_1, r_2, r_3\}$ of an RD problem may be restricted to the distinct subsets $\{r_1, r_2\}$ and $\{r_2, r_3\}$, usually yielding distinct sub-optimal RD curves; cf., \cite[Figure 6.1]{agmon2022RTRD}.
In contrast to RD, the IB's distortion $d_{IB}$ \eqref{eq:d_IB} defined by a root's clusters is determined a posteriori by the problem's solution rather than a priori by the problem's definition.
As a result, both reasons for the existence of bifurcations coincide. 
To see this, consider the IB as an RD problem whose reproduction symbols $\hat{\mathcal{X}}$ are a finite subset of $\Delta[\mathcal{Y}]$ which is allowed to vary (e.g., as if defining the tangent RD problem anew at each $\beta$). 
Distinct restrictions of a reproduction alphabet $\hat{\mathcal{X}}$ can be forced to agree by altering the symbols themselves, so long that they are of the same size. 
For example, restricting the set $\{r_1, r_2, r_3\}$ of reproduction symbols to $\{r_1, r_2\}$ is the same as restricting it to $\{r_2, r_3\}$ instead, and then replacing $r_3$ with $r_1\in \Delta[\mathcal{Y}]$ in the restricted problem\footnote{ This is not to be confused with cluster permutations, which change the order among clusters but do \textit{not} alter the symbols themselves.}.

\begin{algorithm}
	\caption{Root reduction for the IB}
	\begin{algorithmic}[1]
		\Function{Reduce root}{$p(\y|\xhat), p(\xhat); \delta_1, \delta_2$}
		\Input
		\Statex An approximate IB root $\big( p(\y|\xhat), p(\xhat) \big)$ in decoder coordinates,
		\Statex a cluster-mass threshold $0 < \delta_1 < 1$ and a cluster-merging threshold $0 < \delta_2 < 1$.
		\Output{An approximate IB root $\big( \tilde{p}(\y|\xhat), \tilde{p}(\xhat) \big)$ at its effective cardinality.}
		\For{$\xhat$}
		\If{$p(\xhat) < \delta_1$}										\label{algo:root-reduction:nearly-vanished-cluster}
		\Comment{Delete clusters of near-zero mass.}
		\State delete the coordinates of $\xhat$, from $p(\xhat)$ and $p(\y|\xhat)$.
		\EndIf
		\EndFor 
		\State $p(\xhat) \gets $ normalize $p(\xhat)$
		\Comment{Preserve normalization, in case clusters were removed.}
		\newline
		\For{$\xhat \neq \xhatp$}
		\If{$\| p(\y|\xhat) - p(\y|\xhatp)\|_{\infty} < \delta_2$}		\label{algo:root-reduction:distance-between-points}
		\Comment{Merge nearly-identical points in $\Delta[\mathcal{Y}]$.}
		\State $p(\xhat) \gets p(\xhat) + p(\xhatp)$
		\State delete the coordinates of $\xhatp$, from $p(\xhat)$ and $p(\y|\xhat)$.
		\EndIf
		\EndFor \newline
		\State \Return $\big( p(\y|\xhat), p(\xhat) \big)$
		\EndFunction
	\end{algorithmic}
	\label{algo:root-reduction}
\end{algorithm}

The dynamical point of view above, considering an IB root as weighted points traversing $\Delta[\mathcal{Y}]$, offers a straightforward way to identify and handle continuous IB bifurcations. It is spelled out as our root-reduction Algorithm \ref{algo:root-reduction}. 
For cluster-vanishing bifurcations, one can set a small threshold value $\delta_1 > 0$ and consider the cluster $\xhat$ as vanished if $p(\xhat) < \delta_1$ (step \algref{algo:root-reduction}{algo:root-reduction:nearly-vanished-cluster}), as in \cite[Section 3.1]{agmon2022RTRD}. 
Similarly, for cluster-merging bifurcations, one can set a small threshold $\delta_2 > 0$ and consider the clusters $\xhat \neq \xhatp$ to have merged if $\| r_{\xhat} - r_{\xhatp} \|_{\infty} < \delta_2$ (step \algref{algo:root-reduction}{algo:root-reduction:distance-between-points}).
A vanished cluster is then erased (and merged clusters replaced by one), resulting in an approximate IB root on fewer clusters.
This not only identifies continuous IB bifurcations but also handles them, since the output of the root-reduction Algorithm \ref{algo:root-reduction} is a numerically-reduced root, represented in its effective cardinality.
To re-gain accuracy, we shall later invoke the BA-IB Algorithm \ref{algo:BA-IB} on the reduced root, as part of Algorithm \ref{algo:IBRT1} (in Section \ref{sec:IBRT1-algo}).
We note that one should pick the thresholds $\delta_1$ and $\delta_2$ small enough to avoid false detections, and yet not too small so as to cause mis-detections. Mis-detections are handled later, using the heuristic Algorithm \ref{algo:handle-ODE-singularity} (Subsection \ref{sub:IBRT1-algo-spec}). 

Using the root-reduction Algorithm \ref{algo:root-reduction} allows one to stop early in the vicinity of a bifurcation when following the path of an IB root. 
As mentioned in Section \ref{sec:euler-method}, early-stopping restricts the computational difficulty of root-tracking, \citep{agmon2022RTRD}. 
Further, reducing the root \textit{before} invoking BA-IB (Algorithm \ref{algo:BA-IB}) allows to avoid BA's critical slowing down, \citep{agmon2021critical}. 
For, reduction removes the nearly-vanished Jacobian eigenvalues that pertain to the nearly-vanished (or nearly-merged) cluster(s), which are the cause of BA's critical slowing down. 
cf., Proposition \ref{prop:bif-is-detectable-only-on-enough-clusters} (Section \ref{sub:discontinuous-IB-bifs}) and the discussion around it. 
See also \cite[Figure 3.1(C) and Section 3.2]{agmon2022RTRD} for the respective behavior in RD. 
Finally, we comment that the root-reduction Algorithm \ref{algo:root-reduction} can also be implemented in the other two coordinate systems of Section \ref{sec:coords-exchange-for-the-IB}.

\medskip
\subsection{Discontinuous IB bifurcations and linear curve segments}
\label{sub:discontinuous-IB-bifs}

In the previous Subsection \ref{sub:continuous-IB-bifs}, we considered continuous IB bifurcations --- namely, when the clusters $r_{\xhat} \in \Delta[\mathcal{Y}]$ and weights $p(\xhat)$ of an IB root are continuous functions of $\beta$. 
By exploiting the intimate relations between the IB and RD (Section \ref{sub:IB-as-an-RD-problem-and-non-singularity-conj}), we now consider IB bifurcations where these \textit{cannot} be written as a continuous function of $\beta$.
Although in our experience discontinuous bifurcations are infrequent in practice, the theory they evoke has several subtle yet important consequences, with practical implications for computing IB roots (in Section \ref{sec:IBRT1-algo}). 
We start with several examples before diving into the theory.

\begin{figure}[h]
	\centering
	\vspace*{10pt}
	\ifdefined\compilefigs
	\includegraphics[width=1\textwidth]{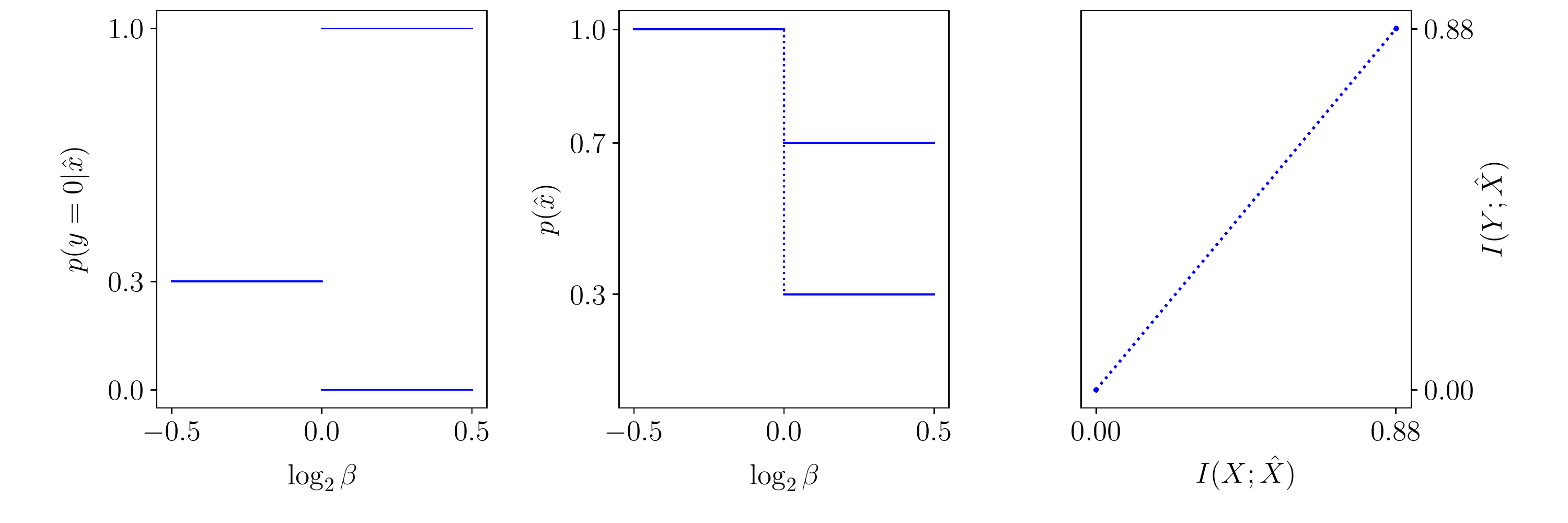}
	\else
	\includegraphics[width=1\textwidth]{figs/empty_figure}
	\fi
	\caption{
		\textbf{A discontinuous IB bifurcation} at $\beta_c = 1$, of the problem defined by $p_{Y|X} \; p_X = \mat{0.3 & \\ & 0.7}$.
		\textbf{Left}: to the left of $\beta_c$, the optimal solution is the trivial one, supported on the IB cluster $p_Y$. To the right it is supported on the boundary points $(1, 0)$ and $(0, 1)$ of $\Delta[\mathcal{Y}]$. 
		\textbf{Middle}: the marginals are constant, except at the point of bifurcation. 
		Any convex combination of the trivial and non-trivial roots is optimal there (dotted).
		That is, this is a \textbf{support-switching bifurcation} as in RD, \cite[Figure 6.2]{agmon2022RTRD}.
		\textbf{Right}: the IB curve exhibits a linear segment of slope $\nicefrac{1}{\beta_c} = 1$, connecting the image of the trivial solution in the information plane (bottom-left) to that of the non-trivial one (top-right).
		See comments in the main text.
	}
	\label{fig:support-switching-bif-in-IB}
\end{figure}

\medskip 
The examples of discontinuous IB bifurcations of which we are aware can be understood in RD context as follows. 
Consider the IB as an RD problem on the continuous reproduction alphabet $\Delta[\mathcal{Y}]$, with IB roots parametrized by points in $\Delta\left[\Delta[\mathcal{Y}]\right]$ (see Section \ref{sub:IB-as-an-RD-problem-and-non-singularity-conj}). 
In RD, the existence of linear curve segments is well known, \citep{berger71}. See, for example, Figure 2.7.6 in the latter and its reproduction \cite[Figure 6.2]{agmon2022RTRD}. 
\cite[Section 6.5]{agmon2022RTRD} offers an explanation of these in terms of a \textit{support-switching bifurcation}. 
Namely, a bifurcation where two RD roots of distinct supports exchange optimality at a particular multiplier value $\beta_c$. 
Both roots evolve smoothly in $\beta$, while only exchanging optimality at the bifurcation. 
At $\beta_c$ itself, every convex combination of these two roots is also an RD root. 
This is manifested by a linear segment of slope $-\beta_c$ in the RD curve (see panels E and F in \cite[Figure 6.2]{agmon2022RTRD}).
In particular, the optimal RD root \textit{cannot} be written as a continuous function of $\beta$.
The sudden emergence of an entire segment of RD roots is best understood in light of B{\'e}zout's Theorem; cf., point (i) in the discussion following Conjecture \ref{conj:BA-IB-Jacob-in-decoder-coords-is-nonsingular-at-reduced-root}. 

For one example of linear curve segments in the IB, say that a matrix $M$ \textit{decomposes} if it can be written (non-trivially) as a block matrix by permuting its rows or columns. In light of the above, we have the following refinement of \cite[Theorem 2.6]{witsenhausen1975conditional},

\begin{thm}			\label{thm:linear-curve-section-for-decomposable-probs-wit}
	The IB curve \eqref{eq:IB-curve-def} has a linear segment at $\beta = 1$ if and only if the problem's definition $p_{Y|X} \; p_X$ decomposes.
\end{thm}

Figure \ref{fig:support-switching-bif-in-IB} demonstrates a simple decomposable problem, exhibiting a support-switching bifurcation at $\beta_c = 1$ between the trivial and non-trivial roots there. 
For other examples, a symmetric binary erasure channel can also be seen to exhibit a support-switching bifurcation, \cite[IV.B]{witsenhausen1975conditional}, which is manifested by a linear segment of slope $\nicefrac{1}{\beta_c}$, for $\beta_c > 1$. 
Similarly, also for Hamming channels with a uniform input, \cite[IV.E]{witsenhausen1975conditional} and \citep{witsenhausen1974entropy}, whose problem definition $p_{Y|X} \; p_X$ is of full support. 
We argue that in the IB, support-switching bifurcations exhibit the same behavior as in RD. 
That is, two roots that evolve smoothly in $\beta$ and exchange optimality at the bifurcation. 
While the sequel can justify this in general, there is a simple way to see this in practice. 
Namely, following the two roots of Figure \ref{fig:support-switching-bif-in-IB} through the bifurcation\footnote{ That is, following the trivial root of Figure \ref{fig:support-switching-bif-in-IB} from left to right and the non-trivial one from right to left, through the bifurcation at $\beta_c = 1$ there.} by using BA-IB with deterministic annealing, \citep{rose1990deterministic}. 
As deterministic annealing usually follows a solution branch continuously, this immediately reveals either root at the region where it is sub-optimal (not displayed). 

A support-switching bifurcation evidently has similar characteristics\footnote{ Strictly speaking, the two roots do not intersect as in a classical transcritical, and so a support-switching bifurcation should perhaps be classified as an \textit{imperfect transcritical}.} to a \textit{transcritical} bifurcation; e.g., \cite[Section 3.2]{strogatz2018nonlinear}. 
This extends the results of \citep{parker2022symmetry_breaking}, who conclude that\footnote{Theorem 5 in \citep{parker2022symmetry_breaking} says that the bifurcations detected by their Theorem 3 are degenerate rather than transcritical. It is then concluded ``that the bifurcation guaranteed by Theorem 3 is [generically] pitchfork-like''.} bifurcations in the IB ``are only of pitchfork type''. 
To see the reason for this discrepancy, note that they employ the mathematical machinery in \citep{golubitsky1988singularities} of bifurcations under symmetry. 
As pitchfork bifurcations are ``common in physical problems that have a symmetry'', \citep[Section 3.4]{strogatz2018nonlinear}, then detecting only pitchforks by using the above machinery might not come as a surprise. 
Both \citep{gedeon2012mathematical} and its sequel \citep{parker2022symmetry_breaking} consider the IB's symmetry to interchanging the coordinates of identical clusters\footnote{ e.g., \cite[Definition 1(1)]{parker2022symmetry_breaking}.}. 
However, this is a structural symmetry of the IB which stems from representing IB roots by finite-dimensional vectors (Subsection \ref{sub:IB-as-an-RD-problem-and-non-singularity-conj}), and is broken at continuous IB bifurcations (Subsection \ref{sub:continuous-IB-bifs}). 
On the other hand, discontinuous IB bifurcations need not break this symmetry, as can be seen by inspecting the roots of Figure \ref{fig:support-switching-bif-in-IB} closely\footnote{ The trivial solution to the left of $\beta_c$ (Figure \ref{fig:support-switching-bif-in-IB}, left panel) may be given a degenerate bi-clustered representation, which is fully supported on $p_Y$ but has a second cluster $r$ of zero-mass. One may choose $r \neq p_Y$, in which case the root possesses no symmetry to interchanging cluster coordinates, at either side of the bifurcation there.}. 

A few convexity results from rate-distortion theory are needed to consider discontinuous bifurcations in general. 
These have subtle practical implications, which are of interest in their own right.

\begin{thm}[Theorem 2.4.2 in \cite{berger71}]				\label{thm:convexity-of-achieving-dists-in-RD-berger}
	The set of conditional probability distributions $p(\xhat|\x)$ which achieve a point $(D, R(D))$ on the rate-distortion curve \eqref{eq:RD-curve-def} is convex.
\end{thm}

Viewing the IB as an RD problem as in \citep{harremoes2007information} immediately yields an identical result for the IB,

\begin{cor}											\label{cor:convexity-of-achieving-dists-at-point-on-IB-curve}
	The set of IB encoders that achieve a point $(I_X, I_Y)$ on the IB curve \eqref{eq:IB-curve-def} is convex.
\end{cor}

The proof is provided below for completeness. 
We note that a version of Corollary \ref{cor:convexity-of-achieving-dists-at-point-on-IB-curve} in inverse-encoder coordinates can also be synthesized from the ideas leading to Theorem 2.3 in \cite{witsenhausen1975conditional}.

\begin{proof}[Proof of Corollary \ref{cor:convexity-of-achieving-dists-at-point-on-IB-curve}]
	Consider a finite IB problem $p_{Y|X} \; p_X$ as an RD problem $\left(d_{IB}, p_X\right)$ on the continuous reproduction alphabet $\Delta[\mathcal{Y}]$, as defined by \eqref{eq:d_IB} in Section \ref{sub:IB-as-an-RD-problem-and-non-singularity-conj}. 
	As noted above, its encoders (or test channels) are conditional probability distributions $p(r|\x)$, with $r \in \Delta[\mathcal{Y}]$, supported on finitely many coordinates $(r, \x)$. 

	Let $p_1(r|\x)$ and $p_2(r|\x)$ be encoders achieving a point $(I_X, I_Y)$ on the IB curve \eqref{eq:IB-curve-def}. 
	By \cite[Theorem 5]{harremoes2007information}, these may be considered as test channels achieving the curve \eqref{eq:RD-curve-def} of the RD problem $\left(d_{IB}, p_X\right)$. 
	The reproduction symbols $r\in \Delta[\mathcal{Y}]$ supporting\footnote{ Defined here $\supp p(r|\x) := \supp p(r)$, where $p(r)$ is defined from $p(r|\x)$ by marginalization, as in \eqref{eq:IB-eq-marginal}. } a convex combination $p_\lambda := \lambda \cdot p_1 + (1 - \lambda) \cdot p_2$, $0 \leq \lambda \leq 1$, are contained in the the supports of $p_1$ and $p_2$, $\supp p_\lambda \subseteq \supp p_1 \cup \supp p_2$, and so $p_\lambda$ is finitely-supported. 
	Although \citeauthor{berger71}'s Theorem \ref{thm:convexity-of-achieving-dists-in-RD-berger} assumes that the reproduction alphabet is finite, one can readily see that its proof works just as well when the distributions involved are finitely-supported. 
	Thus, by Theorem \ref{thm:convexity-of-achieving-dists-in-RD-berger}, $p_\lambda$ achieves the above point on the RD curve \eqref{eq:RD-curve-def}. 
	Since this point $(I_X, I_Y)$ is on the IB curve \eqref{eq:IB-curve-def}, then $p_\lambda$ is an optimal IB root. 
\end{proof}

The RD curve \eqref{eq:RD-curve-def} is the envelope of lines of slope $-\beta$ and intercept $\min_{p(\xhat|\x)} \big( I(X; \hat{X}) + \beta \; \bb{E}[d(\x, \xhat)] \big)$ along the $R$-axis, e.g., \citep{berger71}. 
Thus, Theorem \ref{thm:convexity-of-achieving-dists-in-RD-berger} can be generalized by considering the achieving distributions that pertain to a particular slope value rather than to a particular curve point $(D, R(D))$ --- see \cite[Section 6.3]{agmon2022RTRD}.

\begin{thm}[Theorem 20 in \cite{agmon2022RTRD}]			\label{thm:convexity-of-achieving-distribution-from-RTRD}
	For any $\beta > 0$ value, the set of distributions achieving the RD curve \eqref{eq:RD-curve-def} that correspond to $\beta$ is convex.
\end{thm}

As with Corollary \ref{cor:convexity-of-achieving-dists-at-point-on-IB-curve}, we immediately have an identical result for roots achieving the IB curve \eqref{eq:IB-curve-def},

\begin{cor}												\label{cor:convexity-of-IB-optimal-roots-at-beta}
	For any $\beta > 0$ value, the set of optimal IB encoders that correspond to $\beta$ is convex.
\end{cor}

See also \cite[IV]{witsenhausen1975conditional} for an argument in inverse-encoder coordinates. 
In particular, note the duality technique leading to (b) and (c) in Theorem 4.1 there.
This duality boils down to describing a compact convex set in the plane by its lines of support, as in the observation leading to Theorem \ref{thm:convexity-of-achieving-distribution-from-RTRD}. 
Commensurate with the IB being a special case of RD, Corollary \ref{cor:convexity-of-IB-optimal-roots-at-beta} can also be proven directly from the IB's definitions in direct-encoder terms, \cite{benger2019}.
Note that the requirement that the IB root indeed achieves the curve is necessary. Otherwise one could take convex combinations with the trivial IB root\footnote{ One can verify directly that this satisfies the IB Equations \eqref{eq:IB-eq-encoder}-\eqref{eq:IB-eq-marginal} for every $\beta > 0$.} $p(r|x) = \delta_{r, p_Y}$. This yields absurd since the trivial root contains no information on either $X$ or $Y$.

As in \cite[Section 6.3]{agmon2022RTRD}, the convexity of optimal IB roots (Corollary \ref{cor:convexity-of-IB-optimal-roots-at-beta}) has several important consequences.
For one, unlike the (local) bifurcations we have considered so far, bifurcation theory also has \textit{global bifurcations}. 
These are ``bifurcations that cannot be detected by looking at small neighborhoods of fixed points'', \citep[Section 2.3]{kuznetsov2004elements}.
From convexity, it immediately follows that

\begin{cor}				\label{cor:no-global-bifs-in-IB}
	There are no global bifurcations in finite IB problems.
\end{cor}

Indeed, if at a given $\beta$ value there exists more than one optimal root, then the Jacobian of the IB operator $Id - BA_\beta$ \eqref{eq:IB-operator-def} must have a kernel vector pointing along the line connecting these optimal roots, by Corollary \ref{cor:convexity-of-IB-optimal-roots-at-beta}.

With that comes an important practical caveat.	
Corollaries \ref{cor:convexity-of-achieving-dists-at-point-on-IB-curve} and \ref{cor:convexity-of-IB-optimal-roots-at-beta} hold for the IB when parametrized by points in $\Delta\left[\Delta[\mathcal{Y}]\right]$. 
However, the above kernel vector (which exists due to convexity) may not be detectable if an IB root is improperly represented by a finite-dimensional vector.
For example, consider the bifurcation in Figure \ref{fig:support-switching-bif-in-IB}, where a line segment at $\beta_c$ connects the trivial (single-clustered) root to the 2-clustered root. 
Obviously, the bifurcation there cannot be detected by the Jacobian of the IB operator \eqref{eq:IB-operator-def} when it is computed on $T = 1$ clusters (Jacobian of order $1\cdot (|\mathcal{Y}| + 1)$). 
Indeed, the root of effective cardinality two cannot be represented on a single cluster, and so the line segment connecting it to the trivial root does not exist in a 1-clustered representation.
This is demonstrated in Figure \ref{fig:support-switching-bif-eigs}, which compares Jacobian eigenvalues at reduced representations to those at 2-clustered representations. 
The same reasoning gives the following necessary condition,

\begin{prop}[A necessary condition for detectability of IB bifurcations]			\label{prop:bif-is-detectable-only-on-enough-clusters}
	A bifurcation at $\beta_c$ in a finite IB problem which involves roots of effective cardinalities $T_1$ and $T_2$ is detectable by a non-zero vector in $\ker (I - D_{\log p(\y|\xhat), \log p(\xhat)} BA_{\beta_c})$ only if the latter is evaluated at a representation on at least $\max\{T_1, T_2\}$ clusters. 
\end{prop}

Indeed, suppose that $T_1 \lneqq T_2$ (the conclusion is trivial if $T_1 = T_2$). 
By definition, a root of effective cardinality $T_2$ does not exist in representations with less than $T_2$ clusters. Thus, there is no bifurcation in a $T$-clustered representation if $T < T_2$, and so there is then nothing to detect.
As a special case of this argument, note that Conjecture \ref{conj:BA-IB-Jacob-in-decoder-coords-is-nonsingular-at-reduced-root} (Section \ref{sub:IB-as-an-RD-problem-and-non-singularity-conj}) implies that the Jacobian is non-singular in a $T_1$-clustered representation of the $T_1$-clustered root (namely, at its reduced representation). 
With that, we have observed numerically that the eigenvalues of $D_{\log p(\y|\xhat), \log p(\xhat)} BA_{\beta}$ do \textit{not} depend on the representation's dimension if computed on strictly\footnote{ Computing on one cluster more than the effective cardinality makes sense considering \cite[Theorem 2.6.1]{berger71} or \cite[Lemma 2.2(i)]{witsenhausen1975conditional}, for example.} more clusters than the effective cardinality. Rather, only the eigenvalues' multiplicities vary by dimension.
We omit practical caveats on exchanging between the coordinate systems of Section \ref{sec:coords-exchange-for-the-IB} for brevity. 

\begin{figure}[h]
	\centering
	\vspace*{10pt}
	\ifdefined\compilefigs
	\includegraphics[width=.7\textwidth]{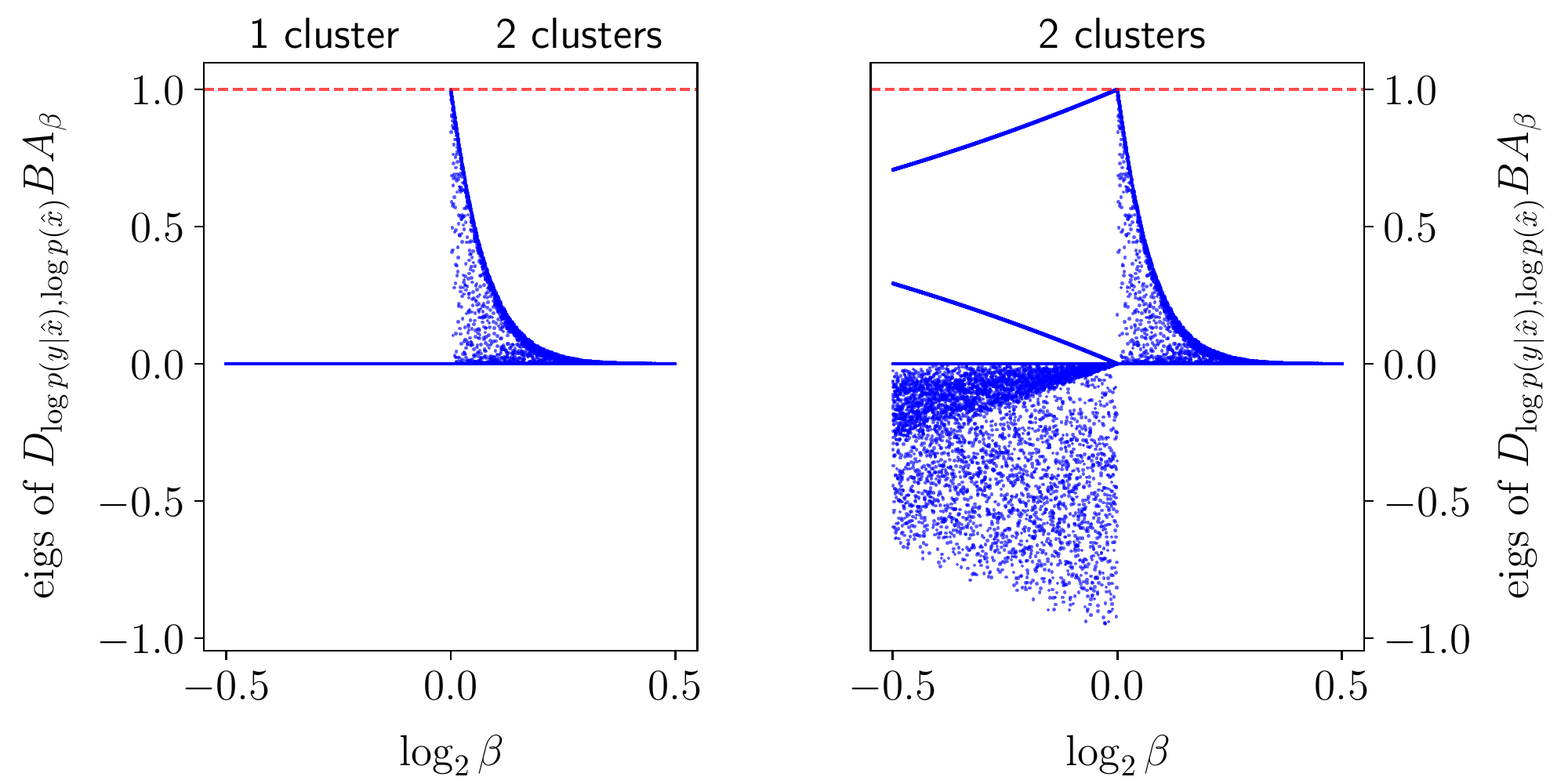}
	\else
	\includegraphics[width=.7\textwidth]{figs/empty_figure}
	\fi
	\caption{
		\textbf{Bifurcations can be detected by $BA_\beta$'s Jacobian only if computed on enough clusters}.
		The approximate eigenvalues of $D_{\log p(\y|\xhat), \log p(\xhat)} BA_\beta$ are plotted by the representation's dimension for the problem in Figure \ref{fig:support-switching-bif-in-IB}.
		The eigenvalues are evaluated at solutions obtained by the BA-IB Algorithm \ref{algo:BA-IB} (stopping condition = $10^{-9}$), initialized anew at random for each $\beta$. 
		While the random initializations account for much of the eigenvalues' spread, they reveal the solution's behavior through its various approximations. 
		Other factors which contribute to this spread are the degeneracy of the solutions (when $\beta < 1$, to the right), BA's loss of accuracy near the bifurcation (Figure \ref{fig:err-in-numerical-derivs-for-BSC-and-BAs-accuracy} bottom), and the decoders' proximity to the simplex boundaries (see Equation \eqref{eq:BA-Jacob-wrt-decoder-coords-at-main-text}).
		\textbf{Left}: 
		when computed at reduced representations (on $T=1$ clusters to the left, $T=2$ to the right), then the eigenvalues at the trivial solution give no indication of the upcoming bifurcation (at $\beta < 1$), unlike the eigenvalues at the 2-clustered root ($\beta > 1$). 
		\textbf{Right}: 
		the bifurcation's presence is clearly noticed also at the trivial solution ($\beta < 1$) when evaluated at its degenerate 2-clustered representations. 
		Indeed, the trivial solution is then represented on the same number of clusters ($T = 2$) as the root to the right ($\beta > 1$) --- see Proposition \ref{prop:bif-is-detectable-only-on-enough-clusters}. 
		However, due to the bifurcation, the eigenvalues' trajectories are not smooth at $\beta_c = 1$. 
		\textbf{Both}: 
		a similar dependency on the representation's dimension also exists in the other bifurcation examples in this paper (though without a spread of eigenvalues).
	}
	\label{fig:support-switching-bif-eigs}
\end{figure}

To complete the discussion on continuous bifurcations (Section \ref{sub:continuous-IB-bifs}), we argue that cluster-merging and cluster-vanishing are indeed bifurcations, where IB roots of distinct effective cardinalities collide and merge into one. 
We offer two ways to see this.
First, using the inverse-encoder\footnote{ Inverse-encoder and decoder coordinates are interchangeable here. Indeed, as noted in Section \ref{sec:coords-exchange-for-the-IB}, the inverse-encoders of an IB root (with no zero-mass clusters) are in bijective correspondence with its decoders. } formulation of the IB in \citep[Section II.A]{witsenhausen1975conditional}, one can consider an optimization problem in which the number of IB clusters is constrained explicitly. 
By the arguments therein, the constrained problem has an optimal root (due to compactness), which achieves the optimal curve of the constrained problem. The latter curve must be sub-optimal if fewer clusters are allowed than needed to achieve the IB curve \eqref{eq:IB-curve-def}. 
Thus, whenever the effective cardinality of an optimal root (in the un-constrained problem) decreases, it must therefore collide with an optimal root of the constrained IB problem, by Corollary \ref{cor:convexity-of-IB-optimal-roots-at-beta}. 
This accords with \cite[Section 3.4]{tishby1999}, which describes IB bifurcations as a separation of optimal and sub-optimal IB curves according to their effective cardinalities. 
Second, consider the reduced form of an IB root at the point of a continuous bifurcation.
Since its effective cardinality decreases there strictly, say from $T_2$ to $T_1$, then the root can be represented on $T_1$ clusters at the bifurcation itself. 
However, the Jacobian of the IB operator \eqref{eq:IB-operator-def} in log-decoder coordinates is non-singular when represented on $T_1$ clusters, as noted after Proposition \ref{prop:bif-is-detectable-only-on-enough-clusters}. 
Thus, by the Implicit Function's Theorem, there is a unique IB root on $T_1$ clusters through this point. It exists at both sides of the bifurcation (above and below the critical point). 
When represented on $T_2$ clusters, however, the latter intersects at the bifurcation with the root of effective cardinality $T_2$, and so the two roots collide and merge there to one. 
This argument is identical to \cite[Section 6.2]{agmon2022RTRD}, which proves that distinct RD roots collide and merge at cluster-vanishing bifurcations in RD. 

The arguments above imply that cluster vanishing bifurcations cannot be detected directly by considering kernel directions of the IB operator \eqref{eq:IB-operator-def} at the bifurcation, as argued in Subsection \ref{sub:continuous-IB-bifs}. 
Indeed, consider a continuous bifurcation, where roots $\bm{p}_1$ and $\bm{p}_2$ of respective effective cardinalities $T_1 < T_2$ intersect. 
These are paths in $\Delta\left[\Delta[\mathcal{Y}]\right]$ that coincide at the bifurcation itself, $\bm{p}_1(\beta_c) = \bm{p}_2(\beta_c)$, and so in particular are of the same effective cardinality $T_1$ there. 
Asking whether a bifurcation is detectable amounts to considering the evaluation of $\ker D(Id - BA_\beta)$ at a finite-dimensional representation (or ``projection'') of $\bm{p}$. 
The Jacobian $D(Id - BA_\beta)$ of the IB operator \eqref{eq:IB-operator-def} is non-singular at $\beta_c$ when evaluated on $T_1$-clusters in log-decoder coordinates, as noted after Proposition \ref{prop:bif-is-detectable-only-on-enough-clusters}. 
We argue that evaluating it on representations with more clusters $T \gneqq T_1$ does \textit{not} allow to detect the bifurcation, not even if $T \geq T_2$. 
See Appendix \ref{sec:degeneracies-of-IB-operator} for a formal argument. 
Intuitively, this follows because picking a degenerate representation amounts to duplicating clusters of the reduced representation or adding clusters of zero mass;	 cf., \textit{reduction} in Subsection \ref{sub:IB-as-an-RD-problem-and-non-singularity-conj}. 
Introducing degeneracies to a reduced root adds no information about the problem at hand. 

Due to the above, cluster-vanishing bifurcations cannot be detected by following a root $\bm{p}_1$ of effective cardinality $T_1$ through a ``cluster growing'' bifurcation, but only by following a root $\bm{p}_2$ with $T_2 > T_1$ till its collision with $\bm{p}_1$. 
As discussed after Conjecture \ref{conj:BA-IB-Jacob-in-decoder-coords-is-nonsingular-at-reduced-root} (Subsection \ref{sub:IB-as-an-RD-problem-and-non-singularity-conj}), the Jacobian of $Id - BA_\beta$ in reduced log-decoder coordinates can then be used to indicate an upcoming collision of $\bm{p}_2$ with $\bm{p}_1$, in addition to the root-reduction Algorithm \ref{algo:root-reduction}. 
The exact same arguments as above apply also to cluster-merging bifurcations. However, as noted in Subsection \ref{sub:continuous-IB-bifs} (and Appendix \ref{sec:equivalent-conds-for-cluster-splitting-appendix}), the stability of a particular IB cluster $\xhat$ is a property of the root itself. Thus, these are detectable by standard local techniques at the point bifurcation. 
Unlike continuous bifurcations, discontinuous bifurcations are inherently detectable due to the line segment in $\Delta\left[\Delta[\mathcal{Y}]\right]$ connecting the roots at the bifurcation (Corollary \ref{cor:convexity-of-IB-optimal-roots-at-beta}), so long that the IB root is represented on sufficiently many clusters (Proposition \ref{prop:bif-is-detectable-only-on-enough-clusters}) --- see Figure \ref{fig:support-switching-bif-eigs}. 
These results make sense, considering that cluster-vanishing bifurcations seem to appear more frequently in practice than other types. 
Intuitively, branching from a suboptimal root $\bm{p}_1$ to an optimal one $\bm{p}_2$ is harder than the other way around, just as learning new relevant information is harder than discarding it. 
Cases where both directions are equally difficult are the exception, as one might expect. 
This is consistent with the later discussion in Subsection \ref{sub:IBRT1-discussion} on the stability of optimal IB roots (Appendix \ref{sec:optimal-IB-root-is-Lyapunov-stable-in-decreasing-beta}). 

When following the path of a reduced IB root (as in Section \ref{sec:euler-method}), one would like to ensure that its bifurcations are indeed detectable by BA's Jacobian. 
Due to the caveats above, it is computationally preferable\footnote{ While one can compute BA's Jacobian on more clusters than necessary, that increases computational costs and may introduce numerical subtleties. } to follow its path as the effective cardinality decreases rather than increases.
As a result, we take only negative step sizes $\Delta \beta < 0$, since the effective cardinality of an optimal IB root cannot decrease with $\beta$. 
To see this, first note that the IB curve $I_Y(I_X)$ \eqref{eq:IB-curve-def} is concave, and so its slope $\nicefrac{1}{\beta}$ cannot increase with $I_X$. That is, $\beta$ cannot decrease with $I_X$.
Second, note that allowing more clusters cannot decrease the $X$-information $\sum_{\xhat} p(\xhat) H\big( p(\x|\xhat) \big)$ achieved by the IB's optimization variables. Indeed, a $T$-clustered variable $\big( p(\x|\xhat), p(\xhat)\big)$ (\textit{not} necessarily a root) can always be considered as $(T+1)$-clustered, by adding a cluster of zero mass.
cf., the construction at \cite[II.A]{witsenhausen1975conditional}.
Thus, the effective cardinality of an optimal root cannot decrease as the constraint $I_X$ on the $X$-information is relaxed.
When both points are combined, the effective cardinality cannot decrease with $\beta$, as argued. 
In contrast to the IB, we note that the behavior of RD problems is more complicated, e.g., \cite[Example 2.7.3 and Problems 2.8-2.10]{berger71}, since the distortion of each reproduction symbol is fixed a priori.

Finally, we proceed with the argument of Section \ref{sub:continuous-IB-bifs} for the case of discontinuous IB bifurcations. 
That is, consider the reduced form of an optimal IB root, and suppose that either its decoders or its weights (or both) cannot be written as a continuous function of $\beta$ in the vicinity of $\beta_c$. 
Write $r^+_{\xhat}$ and $r^-_{\xhat}$ for its distinct decoders as $\beta \to \beta_c^+$ and $\beta \to \beta_c^-$, respectively.
Similarly, $p^+(\xhat)$ and $p^-(\xhat)$ for its non-zero weights. 
Consider the tangent RD problem on the reproduction alphabet $\hat{\mathcal{X}} := \{r^+_{\xhat}\}_{\xhat} \cup \{r^-_{\xhat}\}_{\xhat} \subset \Delta[\mathcal{Y}]$, as in Section \ref{sub:IB-as-an-RD-problem-and-non-singularity-conj}; cf., \cite[Section V]{agmon2021critical}, upon which this argument is based.
By construction, the IB coincides with its tangent RD problem at the two points $\big(r^+_{\xhat}, p^+(\xhat)\big)$ and $\big(r^-_{\xhat}, p^-(\xhat)\big)$.
Since both points achieve the optimal curve at the same slope value $\nicefrac{1}{\beta_c}$, then the linear segment of distributions connecting these points is also optimal, by Theorem \ref{thm:convexity-of-achieving-distribution-from-RTRD}. 
Alternatively, one could apply Corollary \ref{cor:convexity-of-IB-optimal-roots-at-beta} directly to the IB problem. 
Either way, there exists a line segment of optimal IB roots, which pertain to the given slope value. 
In summary,

\begin{thm}
	Let a finite IB problem have a discontinuous bifurcation at $\beta_c > 0$. 
	Then, its IB curve \eqref{eq:IB-curve-def} has a linear segment of slope $\nicefrac{1}{\beta_c}$.
\end{thm}

Unless the decoder sets $\{r^+_{\xhat}\}_{\xhat}$ and $\{r^-_{\xhat}\}_{\xhat}$ are identical, then this is a support-switching bifurcation, as in Figure \ref{fig:support-switching-bif-in-IB}; cf., \cite[Section 6.5]{agmon2022RTRD}.
A priori, the IB roots $\big(r^+_{\xhat}, p^+(\xhat)\big)$ and $\big(r^-_{\xhat}, p^-(\xhat)\big)$ \textit{may} achieve the same point in the information plane, in which case the linear curve segment is of length zero. However, we are unaware of such examples.
Yet, even if such bifurcations exist, they would be detectable by the Jacobian of BA-IB (when represented on enough clusters), subject to Conjecture \ref{conj:BA-IB-Jacob-in-decoder-coords-is-nonsingular-at-reduced-root}.

\newpage
\section{First-order root-tracking for the Information Bottleneck}
\label{sec:IBRT1-algo}

Gathering the results of Sections \ref{sec:coords-exchange-for-the-IB} through \ref{sec:IB-bifurcations}, we can now not only follow the evolution of an IB root along the first-order equation \eqref{eq:IB-beta-ODE-in-decoder-coords}, but can also identify and handle IB bifurcations. 
This is summarized by our First-order Root-Tracking Algorithm \ref{algo:IBRT1} for the IB (IBRT1) in Section \ref{sub:IBRT1-algo-spec}, with some numerical results in Section \ref{sub:IBRT1-numerical-results}. 
Section \ref{sub:IBRT1-discussion} discusses the basic properties of IBRT1, and mainly the surprising quality of approximations of the IB curve \eqref{eq:IB-curve-def} that it produces, as seen in Figure \ref{fig:IBRT1-IB-curve-for-several-densitires}. 
We focus on continuous bifurcations (Section \ref{sub:continuous-IB-bifs}), as in our experience, these are far more frequent than discontinuous ones and are straightforward to handle. cf., Section \ref{sub:IBRT1-discussion}. 

\medskip 
\subsection{The IBRT1 Algorithm \ref{algo:IBRT1}}
\label{sub:IBRT1-algo-spec}

To assist the reader, we first present a simplified version in Algorithm \ref{algo:sIBRT1}, with edge-cases handled at Algorithm \ref{algo:handle-ODE-singularity} --- clarifications follow. 
These two combined form our IBRT1 Algorithm \ref{algo:IBRT1}, specified below. 

\medskip 
\begin{algorithm}
	\caption{Simplified First-order Root-Tracking for the IB}
	\begin{algorithmic}[1]
		\Function{sIBRT1}{$p_{Y|X} \; p_X, \beta_0, p_{\beta_0}(\xhat|\x); \Delta\beta, \delta_1, \delta_2$}
		\Input
		\Statex An IB problem definition $p_{Y|X} \; p_X$ with $\forall x \; p_X(x) > 0$. 
		\Statex A reduced IB-optimal root $p_{\beta_0}(\xhat|\x)$ at $\beta_0$. A step size $\Delta \beta < 0$.
		\Statex Cluster-mass threshold $\delta_1$ and cluster-merging threshold $\delta_2$, with $0 < \delta_i < 1$.
		\Output{Approximations $\tilde{\bm{p}}_{\beta_n}$ of the optimal IB roots $\bm{p}_{\beta_n}$ at $\beta_n := \beta_0 + n \Delta\beta$.}
		\State Initialize $\beta \gets \beta_0$ and $results \gets \{\}$.
		\State Initialize $\tilde{\bm{p}} := \big(\tilde{p}(\xhat|\x), \tilde{p}(\x|\xhat), \tilde{p}(\y|\xhat), \tilde{p}(\xhat)\big)$ from $p_{\beta_0}(\xhat|\x)$, via Equations \algref{algo:BA-IB}{eq:IB-BA-cluster_marginal}-\algref{algo:BA-IB}{eq:IB-BA-decoder-eq}.		\label{algo:sIBRT1:dists-init}
		\While{$\beta > |\Delta \beta|$ and $\left| \supp \tilde{p}(\xhat) \right| > 1$}			\label{algo:sIBRT1:stopping-cond}
			\Comment{See main text on stopping condition.}
			\State Append $\tilde{\bm{p}}$ to $results$.
			\State $\bm{v} := \Big( \tfrac{d \log \tilde{p}(\y|\xhat) }{d\beta}, \tfrac{d \log \tilde{p}(\xhat) }{d\beta} \Big) \gets $ solve the IB ODE \eqref{eq:IB-beta-ODE-in-decoder-coords} at $\tilde{\bm{p}}$.		\label{algo:sIBRT1:solve-ODE}
			\State $\tilde{p}(\y|\xhat) \gets \tilde{p}(\y|\xhat) \exp\left(\Delta\beta \cdot \tfrac{d \log \tilde{p}(\y|\xhat) }{d\beta} \right)$			\label{algo:sIBRT1:exponentiating-linear-approx1}
			\State $\tilde{p}(\xhat) \gets \tilde{p}(\xhat) \exp\left(\Delta\beta \cdot \tfrac{d \log \tilde{p}(\xhat) }{d\beta} \right)$					\label{algo:sIBRT1:exponentiating-linear-approx2}
			\Comment{Exponentiate the linear approximations \eqref{eq:first-order-approx}.}
			\State $\big(\tilde{p}(\y|\xhat), \tilde{p}(\xhat)\big) \gets \text{normalize } \big(\tilde{p}(\y|\xhat), \tilde{p}(\xhat)\big)$		\label{algo:sIBRT1:normalize-after-linear-approx}
			\State $old\_dim \gets \dim \tilde{p}(\xhat)$
			\State $\big(\tilde{p}(\y|\xhat), \tilde{p}(\xhat)\big) \gets \Call{Reduce root}{\tilde{p}(\y|\xhat), \tilde{p}(\xhat); \delta_1, \delta_2}$.		\label{algo:sIBRT1:reducing-root}
			\Comment{Root reduction Algorithm \ref{algo:root-reduction}.}
			\If{$old\_dim \neq \dim \tilde{p}(\xhat)$}
			\Comment{Root was reduced due to bifurcation.}
			\State $\tilde{p}(\xhat|\x) \gets$ the encoder defined by $\big(\tilde{p}(\y|\xhat), \tilde{p}(\xhat)\big)$, via Equations \algref{algo:BA-IB}{eq:IB-BA-partition-func}-\algref{algo:BA-IB}{eq:IB-BA-new-direct-enc}.	\label{algo:sIBRT1:enc-defined-by-reduced-root}
			\State $\tilde{\bm{p}} \gets \Call{BA-IB}{\tilde{p}(\xhat|\x); p_{Y|X} \; p_X, \beta + \Delta \beta}$. \label{algo:sIBRT1:BA-IB-regain-accuracy-after-reduction}
			\begin{flushright}
				\Comment{Ensure accuracy of the reduced root, using BA-IB Algorithm \ref{algo:BA-IB} till convergence.}
			\end{flushright}
			\EndIf
			\State $\beta \gets \beta + \Delta \beta$.
			\State $\tilde{\bm{p}} \gets BA_\beta(\tilde{p}(\y|\xhat), \tilde{p}(\xhat))$		\label{algo:sIBRT1:BA-step-to-enforce-Markovity}
			\Comment{A single BA-IB iteration in decoder coordinates.}
		\EndWhile
		\State Append $\tilde{\bm{p}}$ to $results$.
		\State \Return $results$.
		\EndFunction
	\end{algorithmic}
	\label{algo:sIBRT1}
\end{algorithm}

We now elaborate on the main steps of the Simplified First-order Root-Tracking for the IB (Algorithm \ref{algo:sIBRT1}), which follows Root-Tracking for RD, Algorithm 3 in \citep{agmon2022RTRD}. 
Its purpose is to follow the path of a given IB root $p_{\beta_0}(\xhat|\x)$ in a finite IB problem. 
The initial condition $p_{\beta_0}(\xhat|\x)$ is required to be reduced and IB-optimal. 
Its optimality is needed below to ensure that the path traced by the algorithm is indeed optimal. 
The step-size $\Delta \beta$ is negative, for reasons explained in Section \ref{sub:discontinuous-IB-bifs} (Proposition \ref{prop:bif-is-detectable-only-on-enough-clusters} ff.). 
The cluster-mass and cluster-merging thresholds are as in the root-reduction Algorithm \ref{algo:root-reduction} (Section \ref{sub:continuous-IB-bifs}). 

Denote $\tilde{\bm{p}}$ (line \ref{algo:sIBRT1:dists-init} of Algorithm \ref{algo:sIBRT1}) for the distributions generated from an encoder (cf., Equation \eqref{eq:coordinate-sets-parameterizing-an-IB-root} in Section \ref{sec:coords-exchange-for-the-IB}). 
Algorithm \ref{algo:sIBRT1} iterates over grid points $\tilde{\bm{p}}$, with each \textbf{while} iteration generating the reduced form of the next grid point, as follows. 
On line \ref{algo:sIBRT1:solve-ODE}, evaluate the IB ODE \eqref{eq:IB-beta-ODE-in-decoder-coords} at the current root $\tilde{\bm{p}}$, solving the linear equations numerically. 
By Conjecture \ref{conj:BA-IB-Jacob-in-decoder-coords-is-nonsingular-at-reduced-root} (Section \ref{sub:IB-as-an-RD-problem-and-non-singularity-conj}), the IB ODE has a unique numerical solution $\bm{v}$ if $\tilde{\bm{p}}$ is a reduced root and not a bifurcation. 
Lines \ref{algo:sIBRT1:exponentiating-linear-approx1} and \ref{algo:sIBRT1:exponentiating-linear-approx2} approximate the root at the next grid point at $\beta + \Delta \beta$, by exponentiating Euler-method's step \eqref{eq:first-order-approx} (Section \ref{sec:euler-method}). 
Normalization is enforced on line \ref{algo:sIBRT1:normalize-after-linear-approx}, since it is assumed throughout. 
Off-grid points can be generated by repeating lines \ref{algo:sIBRT1:exponentiating-linear-approx1} through \ref{algo:sIBRT1:normalize-after-linear-approx} for intermediate $\Delta \beta$ values if desired. 
The approximate root at $\beta + \Delta \beta$ is reduced on line \ref{algo:sIBRT1:reducing-root}, by invoking the root-reduction Algorithm \ref{algo:root-reduction} (Section \ref{sub:continuous-IB-bifs}). 
Note that Algorithm \ref{algo:root-reduction} returns its input root unmodified unless reducing it numerically. 
If reduced, then the root is a vector of a lower dimension --- either a cluster mass $p(\xhat)$ has nearly vanished or distinct clusters have nearly merged. 
To re-gain accuracy, we invoke (on line \ref{algo:sIBRT1:BA-IB-regain-accuracy-after-reduction}) the Blahut-Arimoto Algorithm \ref{algo:BA-IB} for the IB till convergence, on the encoder defined at line \ref{algo:sIBRT1:enc-defined-by-reduced-root} by the reduced root. 
Although BA-IB is invoked near a bifurcation, this does \textit{not} incur a hefty computational cost due to its critical slowing-down, \citep{agmon2021critical}  --- see comments at the bottom of Section \ref{sub:continuous-IB-bifs}. 
Invoking BA (on line \ref{algo:sIBRT1:BA-IB-regain-accuracy-after-reduction}) \textit{before} reducing (on line \ref{algo:sIBRT1:reducing-root}) would have inflicted a hefty computational cost to BA-IB due to the nearby bifurcation. 
Finally, a single BA-IB iteration in decoder coordinates is invoked on the approximate root (line \ref{algo:sIBRT1:BA-step-to-enforce-Markovity}), whether reduced earlier or not. 
This enforces Markovity while improving the algorithm's order (see Section \ref{sec:euler-method}, and Figure \ref{fig:Euler-method-for-IB-dec} in particular). 
Algorithm \ref{algo:sIBRT1} continues this way (line \ref{algo:sIBRT1:stopping-cond}) until the approximate solution is trivial (single-clustered), or $\beta$ is non-positive. 
In the IB, the trivial solution is always optimal for tradeoff values $\beta < 1$. However, here $\beta$ plays the role of the ODE's independent variable instead. 
Thus, we allow Algorithm \ref{algo:sIBRT1} to continue beyond $\beta = 1$, so long that\footnote{ The condition $\beta > |\Delta \beta|$ is required on line \ref{algo:sIBRT1:stopping-cond}, to ensure that the target $\beta$ value of the \textit{next} grid point is non-negative. } $\beta > 0$ (which we assume throughout). 
This shall be useful for overshooting --- see below. 

With that, there are caveats in Algorithm \ref{algo:sIBRT1}, which stem from passing too far or close to a bifurcation. 
For one, suppose that the error accumulated from the true solution is too large for a bifurcation to be detected. 
The approximations generated by the algorithm will then overshoot the bifurcation. Namely, proceeding with more clusters than needed until the conditions for reduction are met later on (see Section \ref{sub:IBRT1-discussion} below), as demonstrated by the two sparse grids in Figure \ref{fig:IBRT1-decoders-for-several-densities} (Section \ref{sub:IBRT1-numerical-results}). 
For another, suppose that the current grid point $\tilde{\bm{p}}$ is too close to a bifurcation. 
This might happen due to a variety of numerical reasons --- e.g., thresholds $\delta_1, \delta_2$ too small, or due to the particular grid layout. 
The coefficients matrix\footnote{ That is, the Jacobian $D \left( Id - BA_\beta \right)$ of the IB operator \eqref{eq:IB-operator-def} in log-decoder coordinates.} $I - D_{\log p(\y|\xhat), \log p(\xhat)} BA_\beta$ of the IB ODE \eqref{eq:IB-beta-ODE-in-decoder-coords} would then be ill-conditioned (cf., Conjecture \ref{conj:BA-IB-Jacob-in-decoder-coords-is-nonsingular-at-reduced-root} ff. in Section \ref{sub:IB-as-an-RD-problem-and-non-singularity-conj}), typically resulting in very large implicit numerical derivatives $\bm{v}$ (on line \ref{algo:sIBRT1:solve-ODE}). 
Any inaccuracy\footnote{ e.g., due to the accumulated approximation error or due to the error caused by computing implicit derivatives in the vicinity of a bifurcation (see Figure \ref{fig:err-in-numerical-derivs-for-BSC-and-BAs-accuracy} top, in Section \ref{sec:IB-ODE}).} in  $\bm{v}$ might then send the next grid point astray, derailing the algorithm from there on. 
Indeed, the derivatives $\tfrac{d\bm{x}}{d\beta} = - (D_{\bm{x}} F)^{-1} D_\beta F$ defined by\footnote{ Note that $D_{\bm{x}} F$ here is always non-singular outside bifurcations, due to Conjecture \ref{conj:BA-IB-Jacob-in-decoder-coords-is-nonsingular-at-reduced-root} and the use of reduced coordinates. } the implicit ODE \eqref{eq:implicit-beta-ODE} are in general unbounded near a bifurcation of $F$. 
This can be seen in Figure \ref{fig:norm-of-analytical-derivs-for-BSC} (Section \ref{sec:coords-exchange-for-the-IB}) for example, where the derivatives ``explode'' at the bifurcation's vicinity. 
See also \cite[Section 7.2]{agmon2022RTRD} on the computational difficulty incurred by a bifurcation. 
While overshooting a bifurcation is not a significant concern for our purposes (see Section \ref{sub:IBRT1-discussion}), passing too close to one is. 
The latter is important, especially when the step size $|\Delta \beta|$ is small. 
While decreasing $|\Delta \beta|$ generally improves the error of Euler's method, it also makes it easier for the approximations to come close to a bifurcation, thus potentially worsening the approximation dramatically if it derails. 
This motivates one to consider how singularities of the IB ODE \eqref{eq:IB-beta-ODE-in-decoder-coords} should be handled. 

\begin{algorithm}
	\caption{A heuristic for handling singularities of the IB ODE \eqref{eq:IB-beta-ODE-in-decoder-coords}}
	\begin{algorithmic}[1]
		\Function{Handle singularity}{$p_{Y|X} \; p_X, \big( \tilde{p}(\y|\xhat), \tilde{p}(\xhat) \big), \bm{v}, \beta$}
		\Input
		\Statex An IB problem definition $p_{Y|X} \; p_X$, with $\forall x \; p_X(x) > 0$.
		\Statex An approximate root $\big( \tilde{p}(\y|\xhat), \tilde{p}(\xhat) \big)$ of the given problem, near a singularity of the IB ODE \eqref{eq:IB-beta-ODE-in-decoder-coords}. 
		\Statex Approximate numerical derivatives $\bm{v} := \Big( \tfrac{d \log \tilde{p}(\y|\xhat) }{d\beta}, \tfrac{d \log \tilde{p}(\xhat) }{d\beta} \Big)$ at the given root.
		\Statex The $\beta > 0$ value of the next (output) grid point. 
		\Output{An approximate IB root $\tilde{\bm{p}}$ at $\beta$ on one fewer cluster.}
		\State $\xhatp, \xhatpp \gets$ the two indices $\xhat$ of largest $\left\| \tfrac{d \log \tilde{p}(\y|\xhat) }{d\beta} \right\|_\infty$ value (norm of $\y$-indexed vectors).		\label{algo:handle-ODE-singularity:identification}
		\State $\tilde{p}(\y|\xhatp) \gets \tfrac{1}{2} \cdot \big( \tilde{p}(\y|\xhatp) + \tilde{p}(\y|\xhatpp) \big)$
		\Comment{Replace fastest-moving clusters by their mean.}		\label{algo:handle-ODE-singularity:merge-two-clusters}
		\State Erase $\xhatpp$ from the decoder $\tilde{p}(\y|\xhat)$.
		\State $\tilde{p}(\xhatp) \gets \tilde{p}(\xhatp) + \tilde{p}(\xhatpp)$
		\State Erase $\xhatpp$ from the marginal $\tilde{p}(\xhat)$.		\label{algo:handle-ODE-singularity:erase-leftover-coords}
		\State $\tilde{p}(\xhat|\x) \gets$ the encoder generated from $(\tilde{p}(\y|\xhat), \tilde{p}(\xhat))$, via Equations \algref{algo:BA-IB}{eq:IB-BA-partition-func}-\algref{algo:BA-IB}{eq:IB-BA-new-direct-enc}.		\label{algo:handle-ODE-singularity:new-enc-after-merging}
		\begin{flushright}
			\Comment{A new encoder on one cluster \textit{less} than the input.}
		\end{flushright}
		\State $\tilde{\bm{p}} \gets \Call{BA-IB}{\tilde{p}(\xhat|\x); p_{Y|X} \; p_X, \beta}$.			\label{algo:handle-ODE-singularity:BA-IB-after-merging}
		\Comment{Re-gain accuracy, by the BA-IB Algorithm \ref{algo:BA-IB}.}
		\State \Return $\tilde{\bm{p}}$
		\EndFunction
	\end{algorithmic}
	\label{algo:handle-ODE-singularity}
\end{algorithm}

Next, we elaborate on our heuristic for handling singularities of the IB ODE \eqref{eq:IB-beta-ODE-in-decoder-coords}, brought as Algorithm \ref{algo:handle-ODE-singularity}. 
The inputs of this heuristic are defined as in Algorithm \ref{algo:sIBRT1}. 
It starts with the assumption that the coefficients matrix $I - D_{\log p(\y|\xhat), \log p(\xhat)} BA_\beta$ of the IB ODE \eqref{eq:IB-beta-ODE-in-decoder-coords} is nearly-singular at the current grid point $\tilde{\bm{p}}$ due to\footnote{ While a priori the Jacobian $D_{\log p(\y|\xhat), \log p(\xhat)} (Id - BA_\beta)$ may be singular also due to other reasons, by Conjecture \ref{conj:BA-IB-Jacob-in-decoder-coords-is-nonsingular-at-reduced-root} it is non-singular at the approximations generated so far since they are assumed to be in their reduced form. cf., Section \ref{sub:IB-as-an-RD-problem-and-non-singularity-conj}. } a nearby bifurcation. 
As a result, the implicit derivatives $\bm{v}$ at $\tilde{\bm{p}}$ are not to be used directly to extrapolate the next grid point, as explained above. 
Instead, we use them to identify the two\footnote{ While this can be refined to handle more than two fast-moving clusters at once, that is not expected to be necessary for typical bifurcations. } fastest moving clusters, on line \ref{algo:handle-ODE-singularity:identification} of Algorithm \ref{algo:handle-ODE-singularity}. 
These are replaced by a single cluster (lines \ref{algo:handle-ODE-singularity:merge-two-clusters} through \ref{algo:handle-ODE-singularity:erase-leftover-coords}), resulting in an approximate root on one fewer cluster. 
To re-gain accuracy, the BA-IB Algorithm \ref{algo:BA-IB} is then invoked (at line \ref{algo:handle-ODE-singularity:BA-IB-after-merging}) on the encoder generated (at line \ref{algo:handle-ODE-singularity:new-enc-after-merging}) from the latter root, thereby generating the next grid point. 
If the fast-moving clusters have merged (in the true solution) by the following grid point, then the output of Algorithm \ref{algo:handle-ODE-singularity} will be an IB-optimal root if its input grid point is so. 
Namely, the branch followed by the algorithm remains an optimal one. 
Otherwise, if these clusters merge shortly after the next grid point, then Algorithm \ref{algo:handle-ODE-singularity} yields a sub-optimal branch. 
However, optimality is re-gained shortly afterward since the sub-optimal branch collides and merges with the optimal one in continuous IB bifurcations (Section \ref{sub:discontinuous-IB-bifs}). 
Figure \ref{fig:IBRT1-decoders-for-several-densities} below demonstrates Algorithm \ref{algo:handle-ODE-singularity}. 
cf., the similar heuristic \cite[Section 3.2]{agmon2022RTRD} in root-tracking for RD, which may also lose optimality near a bifurcation and re-gain it shortly after. 

The heuristic Algorithm \ref{algo:handle-ODE-singularity} is motivated by cluster-merging bifurcations. 
In these, the implicit derivatives are very large only\footnote{ Note that cluster masses barely change in the vicinity of a cluster-merging, till the point of bifurcation itself. } at the coordinates $\tfrac{d \log {p}(\y|\xhat) }{d\beta}$ of the points colliding in $\Delta[\mathcal{Y}]$. 
While intended for cluster-merging bifurcations, this heuristic works nicely in practice also for cluster-vanishing ones. 
To see why, note that one can always add a cluster of zero-mass to an IB root without affecting the root's essential properties, regardless of its coordinates in $\Delta[\mathcal{Y}]$; cf., Section \ref{sub:IB-as-an-RD-problem-and-non-singularity-conj} on reduction in the IB. 
Therefore, a numerical algorithm may, in principle, do anything with the coordinates $p(\y|\xhat) \in \Delta[\mathcal{Y}]$ of a nearly-vanished cluster $\xhat$, $p(\xhat) \simeq 0$, without affecting the approximation's quality too much. 
Thus, for numerical purposes, one may treat a cluster-vanishing bifurcation as a cluster-merging one. 
Conversely, in a cluster-merging bifurcation, a numerical algorithm may, in principle, zero the mass of one cluster while adding it to the remaining cluster. Again, without affecting the approximation's quality too much. 
To conclude, for numerical purposes, cluster-vanishing is very similar to cluster-merging. 
A variety of treatments between these extremities may be possible by a numerical algorithm. 
Empirically, we have observed that our ODE-based algorithm treats both as cluster-merging bifurcations. 
To our understanding, this is because our algorithm operates in decoder coordinates, unlike the BA-IB Algorithm \ref{algo:BA-IB}, for example, which operates in encoder coordinates. 

Finally, we combine the simplified root-tracking Algorithm \ref{algo:sIBRT1} with the heuristic Algorithm \ref{algo:handle-ODE-singularity} for handling singularities, yielding our IBRT1 Algorithm \ref{algo:IBRT1}. 
It follows the lines of simplified Algorithm \ref{algo:sIBRT1}, except that after solving for the implicit derivatives on line \ref{algo:IBRT1:solve-ODE}, we test the IB ODE \eqref{eq:IB-beta-ODE-in-decoder-coords} for singularity. 
To that end, we propose to use the matrix $S$ \eqref{eq:smaller-matrix-for-ker-of-IB-operator-in-dec-coords-main-text} (from Lemma \ref{lem:smaller-matrix-for-ker-of-J-in-decoder-coords} in Section \ref{sec:IB-ODE}), since its order $T\cdot |\mathcal{Y}|$ is smaller than the order $T\cdot \left( |\mathcal{Y}| + 1 \right)$ of the ODE's coefficients matrix. 
This might make it computationally cheaper to test for singularity (on lines \ref{algo:IBRT1:S-matrix-to-compute-kernel} and \ref{algo:IBRT1:singularity-test}). 
Our heuristic Algorithm \ref{algo:handle-ODE-singularity} is invoked (on line \ref{algo:IBRT1:handle-sigularity}) if the ODE \eqref{eq:IB-beta-ODE-in-decoder-coords} is found to be nearly-singular, otherwise proceeding as in Algorithm \ref{algo:sIBRT1}.

\begin{algorithm}
	\caption{First-order Root-Tracking for the IB (IBRT1)}
	\begin{algorithmic}[1]
		\Function{IBRT1}{$p_{Y|X} \; p_X, \beta_0, p_{\beta_0}(\xhat|\x); \Delta\beta, \delta_1, \delta_2, \delta_3$}
		\Input
			\Statex An IB problem definition $p_{Y|X} \; p_X$ with $\forall x \; p_X(x) > 0$. 
			\Statex A reduced IB-optimal root $p_{\beta_0}(\xhat|\x)$ at $\beta_0$. A step size $\Delta \beta < 0$.
			\Statex Thresholds $0 < \delta_1, \delta_2 < 1$ for the root-reduction Algorithm \ref{algo:root-reduction} (cluster mass and merging).
			\Statex A threshold $0 < \delta_3 < 1$ for eigenvalues' singularity.
		\Output{Approximations $\tilde{\bm{p}}_{\beta_n}$ of the optimal IB roots $\bm{p}_{\beta_n}$ at $\beta_n := \beta_0 + n \Delta\beta$.}
		\State Initialize $\beta \gets \beta_0$ and $results \gets \{\}$.
		\State Initialize $\tilde{\bm{p}} := \big(\tilde{p}(\xhat|\x), \tilde{p}(\x|\xhat), \tilde{p}(\y|\xhat), \tilde{p}(\xhat)\big)$ from $p_{\beta_0}(\xhat|\x)$, via Equations \algref{algo:BA-IB}{eq:IB-BA-cluster_marginal}-\algref{algo:BA-IB}{eq:IB-BA-decoder-eq}.
		\While{$\beta > |\Delta \beta|$ and $\left| \supp \tilde{p}(\xhat) \right| > 1$}
			\State Append $\tilde{\bm{p}}$ to $results$.
			\State $\bm{v} := \Big( \tfrac{d \log \tilde{p}(\y|\xhat) }{d\beta}, \tfrac{d \log \tilde{p}(\xhat) }{d\beta} \Big) \gets $ solve the IB ODE \eqref{eq:IB-beta-ODE-in-decoder-coords} at $\tilde{\bm{p}}$.		\label{algo:IBRT1:solve-ODE}
			\State $eigs \gets \eig \left( I - S \right)\big\rvert_{\tilde{\bm{p}}}$		\label{algo:IBRT1:S-matrix-to-compute-kernel}
			\Comment{Test ODE for singularity, using $S$ \eqref{eq:smaller-matrix-for-ker-of-IB-operator-in-dec-coords-main-text} from Lemma \ref{lem:smaller-matrix-for-ker-of-J-in-decoder-coords}.}
			\If{$\left(\min_{v\in eigs} |v|\right) < \delta_3$}			\label{algo:IBRT1:singularity-test}
				\Comment{ODE is nearly-singular.}
				\State $\tilde{\bm{p}} \gets \Call{Handle singularity}{p_{Y|X} \; p_X, \big( \tilde{p}(\y|\xhat), \tilde{p}(\xhat) \big), \bm{v}, \beta + \Delta \beta}$ 		\label{algo:IBRT1:handle-sigularity}
				\begin{flushright}
					\Comment{Handle otherwise undetected singularity using Algorithm \ref{algo:handle-ODE-singularity}.}
				\end{flushright}
			\Else
				\State $\tilde{p}(\y|\xhat) \gets \tilde{p}(\y|\xhat) \exp\left(\Delta\beta \cdot \tfrac{d \log \tilde{p}(\y|\xhat) }{d\beta} \right)$
				\State $\tilde{p}(\xhat) \gets \tilde{p}(\xhat) \exp\left(\Delta\beta \cdot \tfrac{d \log \tilde{p}(\xhat) }{d\beta} \right)$
				\State $\big(\tilde{p}(\y|\xhat), \tilde{p}(\xhat)\big) \gets \text{normalize } \big(\tilde{p}(\y|\xhat), \tilde{p}(\xhat)\big)$
				\State $old\_dim \gets \dim \tilde{p}(\xhat)$
				\State $\big(\tilde{p}(\y|\xhat), \tilde{p}(\xhat)\big) \gets \Call{Reduce root}{\tilde{p}(\y|\xhat), \tilde{p}(\xhat); \delta_1, \delta_2}$.
				\If{$old\_dim \neq \dim \tilde{p}(\xhat)$}
					\State $\tilde{p}(\xhat|\x) \gets$ encoder defined from $\big(\tilde{p}(\y|\xhat), \tilde{p}(\xhat)\big)$, via Equation \algref{algo:BA-IB}{eq:IB-BA-partition-func}-\algref{algo:BA-IB}{eq:IB-BA-new-direct-enc}.
					\State $\tilde{\bm{p}} \gets \Call{BA-IB}{\tilde{p}(\xhat|\x); p_{Y|X} \; p_X, \beta + \Delta \beta}$.
				\EndIf
			\EndIf
			\State $\beta \gets \beta + \Delta \beta$.
			\State $\tilde{\bm{p}} \gets BA_\beta(\tilde{p}(\y|\xhat), \tilde{p}(\xhat))$
		\EndWhile
		\State Append $\tilde{\bm{p}}$ to $results$.
		\State \Return $results$.
		\EndFunction
	\end{algorithmic}
	\label{algo:IBRT1}
\end{algorithm}

\flushbottom 
\newpage 

\medskip 
\subsection{Numerical results for the IBRT1 Algorithm \ref{algo:IBRT1}}
\label{sub:IBRT1-numerical-results}

To demonstrate the IBRT1 Algorithm \ref{algo:IBRT1}, we present the numerical results used to approximate the IB curve in Figure \ref{fig:IBRT1-IB-curve-for-several-densitires} (Section \ref{sec:introduction}) --- see Section \ref{sub:IBRT1-discussion} below on the approximation quality and the algorithm's basic properties. 
This example was chosen both because it has an analytical solution (Appendix \ref{sec:analytical-IB-sol-for-BSC-appendix}) and because it allows one to get a good idea of the bifurcation handling added (in Section \ref{sub:IBRT1-algo-spec}) on top of the modified Euler method (from Section \ref{sec:euler-method}). 
The source code used to generate these results is provided for readers who wish to examine the details (bottom of Section \ref{sec:introduction}). 

\begin{figure}[h!]
	\centering
	\vspace*{10pt}
	\ifdefined\compilefigs
	\includegraphics[width=1\textwidth]{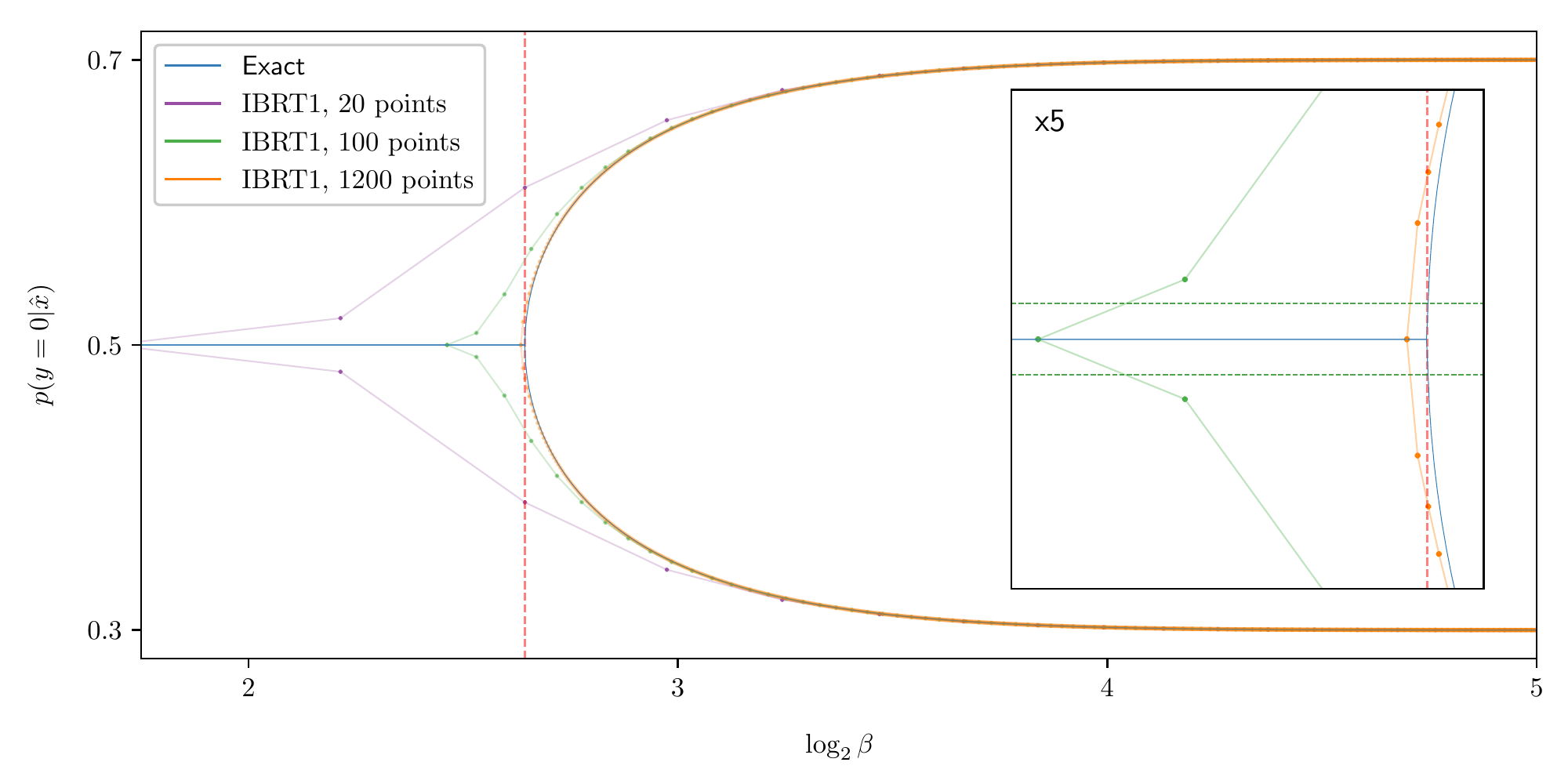}
	\else
	\includegraphics[width=1\textwidth]{figs/empty_figure}
	\fi
	\caption{
		\textbf{Clusters of the approximate IB roots generated by the IBRT1 Algorithm \ref{algo:IBRT1} for several step-sizes}, on top of the exact solutions of BSC$(0.3)$ with a uniform source (Appendix \ref{sec:analytical-IB-sol-for-BSC-appendix}). 
		Carefully note that only a \textit{single} IB root is plotted here; its two clusters merge at $\beta_c$, as seen in Figure \ref{fig:exact-sol-of-BSC} (Section \ref{sub:continuous-IB-bifs}). 
		At 20 and 100 grid points, the approximations overshoot the bifurcation, terminating due to (approximate) cluster collision. 
		While on 1200 grid points, the approximations pass too close to the bifurcation, terminating due to the nearby singularity. 
		This can be seen in \textbf{the inset to the right}: 
		The leftmost green marker has passed the cluster-merging threshold (dashed-green lines), and so was numerically reduced to the trivial (single-clustered) solution by the root-reduction Algorithm \ref{algo:root-reduction}. 
		On the other hand, the orange markers to the right are still far from the cluster-merging threshold; the leftmost one was reduced by the singularity-handling heuristic Algorithm \ref{algo:handle-ODE-singularity} since the IB ODE \eqref{eq:IB-beta-ODE-in-decoder-coords} is nearly-singular there.
		Indeed, the numerical derivative is about five orders of magnitude larger there than at the algorithm's initial condition (see Figure \ref{fig:norm-of-analytical-derivs-for-BSC}) due to the bifurcation's proximity. 
		The leftmost green and orange markers were drawn \textit{after} the reductions took place. 
		See main text and Section \ref{sub:IBRT1-algo-spec} for details, Figure \ref{fig:IBRT1-err-from-exact-sol-for-several-densities} for errors, and Figure \ref{fig:IBRT1-IB-curve-for-several-densitires} (in Section \ref{sec:introduction}) for the approximate IB curves.
		\newline
		The marginals $p(\xhat)$ are not shown, as these barely deviate from their true value in this problem. 
		For each step-size $\Delta \beta$, the algorithm was initialized at the problem's exact solution at $\beta = 2^5$, with thresholds set to $\delta_i = 10^{-2}$, for $i = 1, 2, 3$. 
		The lines connecting consecutive markers are for visualization only. 
	}
	\label{fig:IBRT1-decoders-for-several-densities}
\end{figure}

We discuss the numerical examples of this Section in light of the explanations provided in the previous Section \ref{sub:IBRT1-algo-spec}. 
The error of the IBRT1 Algorithm \ref{algo:IBRT1} generally improves as the step-size $|\Delta \beta|$ becomes smaller, as expected. 
The single BA-IB iteration added to Euler's method (in Section \ref{sec:euler-method}) typically allows one to achieve the same error by using much fewer grid points, thus lowering computational costs. 
For example, the two denser grids in Figure \ref{fig:IBRT1-decoders-for-several-densities} require about an order of magnitude fewer points to achieve the same error compared to Euler's method for the IB; this can be seen from Figure \ref{fig:Euler-method-for-IB-dec} (Section \ref{sec:euler-method}). 

In sparse grids, the approximations often pass too far away from a bifurcation for the root-reduction Algorithm \ref{algo:root-reduction} to detect it.
When overshooting it, the conditions for numerical reduction are generally met later on, as discussed in Section \ref{sub:IBRT1-discussion} below. 
Decreasing $|\Delta \beta|$ further often leads the approximations too close to a bifurcation, as can be seen in the densest grid of Figure \ref{fig:IBRT1-decoders-for-several-densities}. 
The implicit derivatives are typically very large at the proximity of a bifurcation, while the least accurate there (see Section \ref{sub:IBRT1-algo-spec}). 
As these might send subsequent grid points off-track, the heuristic Algorithm \ref{algo:handle-ODE-singularity} is invoked to handle the nearby singularity (see inset of Figure \ref{fig:IBRT1-decoders-for-several-densities}). 
As noted earlier, the computational difficulty in tracking IB roots (or root-tracking in general) stems from the presence of a bifurcation, manifested here by large approximation errors in its vicinity. 
While the algorithm's error peaks at the bifurcation, it typically decreases afterward when overshooting, as seen in Figure \ref{fig:IBRT1-err-from-exact-sol-for-several-densities}. See Section \ref{sub:IBRT1-discussion} for details.

\begin{figure}[h!]
	\centering
	\vspace*{10pt}
	\ifdefined\compilefigs
	\includegraphics[width=1\textwidth]{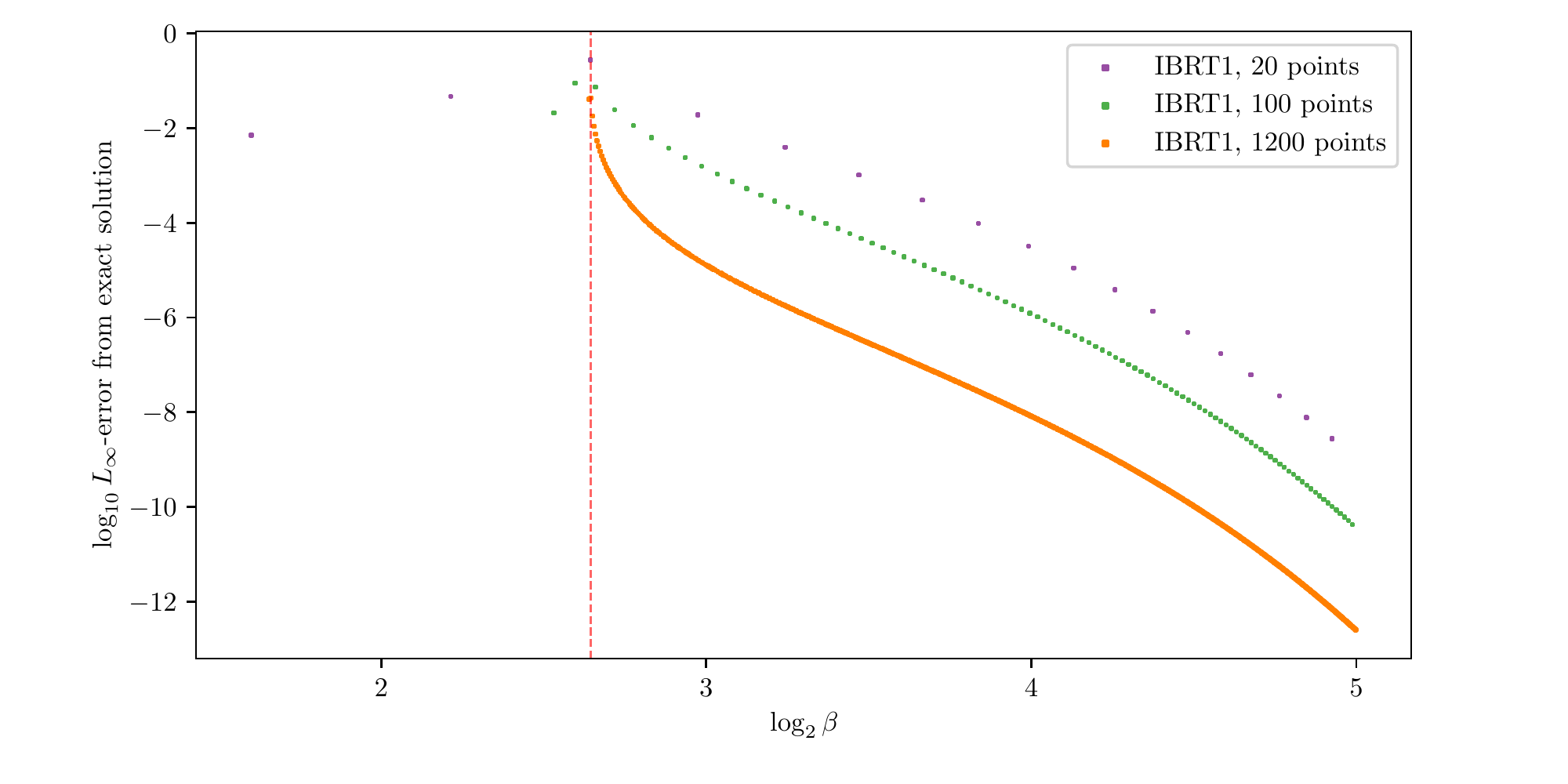}
	\else
	\includegraphics[width=1\textwidth]{figs/empty_figure}
	\fi
	\caption{
		\textbf{The error of the IBRT1 Algorithm \ref{algo:IBRT1} from the exact solution for several step-sizes}. 
		The figure shows the (log-) $L_\infty$-error of the numerical approximations in Figure \ref{fig:IBRT1-decoders-for-several-densities} from the exact solutions; the error is measured as in Figure \ref{fig:err-in-numerical-derivs-for-BSC-and-BAs-accuracy} (bottom). 
		Increasing the grid density decreases the error, as one might expect.
		While the error peaks at the bifurcation, it decreases afterward --- see main text and Section \ref{sub:IBRT1-discussion} below. 
		The rightmost marker for each grid density is missing since the initial error is zero. 
	}
	\label{fig:IBRT1-err-from-exact-sol-for-several-densities}
\end{figure}

\medskip 
\subsection{Basic properties of the IBRT1 Algorithm \ref{algo:IBRT1} and why does it work}
\label{sub:IBRT1-discussion}

Apart from presenting the basic properties of the IBRT1 Algorithm \ref{algo:IBRT1}, the primary purpose of this section is to understand why does it approximate the problem's true IB curve \eqref{eq:IB-curve-def} so well, despite its apparent errors in approximating the IB roots? 
While shown here only in Figures \ref{fig:IBRT1-IB-curve-for-several-densitires} and \ref{fig:IBRT1-decoders-for-several-densities} (Sections \ref{sub:IBRT1-numerical-results} and \ref{sec:introduction}), this behavior is consistent in the few numerical examples that we have tested. We offer an explanation why this may be true in general. 

\medskip 
To understand why the IBRT1 Algorithm \ref{algo:IBRT1} approximates the true IB curve \eqref{eq:IB-curve-def} so well, we first explain why overshooting is not a significant concern, as noted earlier in Section \ref{sub:IBRT1-algo-spec}. 
To that end, consider the implicit ODE \eqref{eq:implicit-beta-ODE}
\begin{equation*}
	\tfrac{d\bm{x}}{d\beta} = - (D_{\bm{x}} F)^{-1} D_\beta F \;,
\end{equation*}
from Section \ref{sec:introduction}. 
So long that $D_{\bm{x}} F$ and $D_\beta F$ at its right-hand side are well-defined, it defines a vector field on the entire \textit{phase space} of admissible $\bm{x}$ values, at least when $D_{\bm{x}} F$ is non-singular. 
That is, even for $\bm{x}$'s which are \textit{not} roots \eqref{eq:root-of-functional-eq-implicit} of $F$. 
Ignoring several technicalities, the IB ODE \eqref{eq:IB-beta-ODE-in-decoder-coords} therefore defines a vector field also \textit{outside} IB roots. 
Indeed, due to Conjecture \ref{conj:BA-IB-Jacob-in-decoder-coords-is-nonsingular-at-reduced-root}, the Jacobian of the IB operator $Id - BA_\beta$ \eqref{eq:IB-operator-def} is non-singular in the vicinity of a reduced root, \footnote{ By Equation \eqref{eq:BA-Jacob-wrt-decoder-coords-at-main-text} (in Section \ref{sec:IB-ODE}), $D_{\log p(\y|\xhat), \log p(\xhat)} BA_\beta$ is continuous in the distributions defining it, under mild assumptions. cf., Lemma \ref{lemma:sufficient-condition-for-IB-decoder-to-never-vanish} in Appendix \ref{sec:BA-IB-in-decoder-coords-appendix}. Thus, so are its eigenvalues. }. 
Now, suppose that $\bm{p}_\beta$ is an optimal IB root, and consider a point $\bm{p}' \neq \bm{p}_{\beta}$ in its vicinity. 
An argument based on a strong notion of Lyapunov stability (in Appendix \ref{sec:optimal-IB-root-is-Lyapunov-stable-in-decreasing-beta}) shows that $\bm{p}'$ flows along the IB's vector field towards $\bm{p}_{\beta}$ in regions that do not contain a bifurcation, though only if flowing in \textit{decreasing} $\beta$ as done by our IBRT Algorithm \ref{algo:IBRT1}. 
An approximation $\bm{p}'$ would then be ``pulled'' towards the true root. 
Stability at this direction of $\beta$ is very reasonable, considering that $\bm{p}_\beta$ follows a path of decreasingly informative representations as $\beta$ decreases. 
Indeed, all the paths to oblivion lead to one place --- the trivial solution, whose representation in reduced coordinates is unique. 
As a result, a numerical approximation $\bm{p}'$ would gradually settle in the vicinity of the true root $\bm{p}_\beta$ as seen in Figures \ref{fig:IBRT1-decoders-for-several-densities} and \ref{fig:IBRT1-err-from-exact-sol-for-several-densities}, so long that $\bm{p}_\beta$ does not change much and the step-size $|\Delta \beta|$ is small enough. 
While this explanation obviously breaks near a bifurcation, it does suggest that the approximation error should decrease when overshooting it (see Section \ref{sub:IBRT1-algo-spec}), once the true reduced root has settled down. 
In a sense, overshooting is similar to being in the right place but at the wrong time. 

The above suggests that the IBRT1 Algorithm \ref{algo:IBRT1} should generally approximate the true IB curve \eqref{eq:IB-curve-def} well, despite its errors in approximating IB roots. 
To see this, note that while $\beta^{-1}$ is the slope of the optimal curve \eqref{eq:IB-curve-def} of the IB, \cite[Equation (32)]{tishby1999}, for the IB ODE \eqref{eq:IB-beta-ODE-in-decoder-coords} it is merely a ``time-like'' independent variable. 
When solving for the optimal curve \eqref{eq:IB-curve-def}, one is not interested in an optimal root or at its $\beta$ value, but rather at its image $\big( I(X; \hat{X}), I(Y; \hat{X}) \big)$ in the information plane. 
As a result, achieving the optimal roots but on the wrong $\beta$ values does yield the true IB curve \eqref{eq:IB-curve-def}, as required. 
This is the reason that the true curve \eqref{eq:IB-curve-def} is achieved in Figure \ref{fig:IBRT1-IB-curve-for-several-densitires} (Section \ref{sec:introduction}) even on sparse grids, despite the apparent approximation errors in Figures \ref{fig:IBRT1-decoders-for-several-densities} and \ref{fig:IBRT1-err-from-exact-sol-for-several-densities} (Section \ref{sub:IBRT1-numerical-results}). 
With that, expect the approximate IB curve produced by the IBRT1 Algorithm \ref{algo:IBRT1} to be of lesser quality when there are more than two possible labels $y$. 
To see this, note that the space $\Delta[\mathcal{Y}]$ traversed by the approximate clusters is not one-dimensional then, and so it is possible to maneuver around clusters of an optimal IB root. 

\medskip 
Next, we briefly discuss the basic properties of the IBRT1 Algorithm \ref{algo:IBRT1}. 
Its \textbf{computational complexity} is determined by the complexity of a single grid point. 
The latter is readily seen to be dominated by the complexity $O\big(T^2 \cdot |\mathcal{Y}|^2 \cdot \left( |\mathcal{X}| + T\cdot |\mathcal{Y}| \right) \big)$ of computing the coefficients matrix of the IB ODE \eqref{eq:IB-beta-ODE-in-decoder-coords} and of solving it numerically (on line \ref{algo:IBRT1:solve-ODE}). 
To that, one should add the complexity of the BA-IB Algorithm \ref{algo:BA-IB} each time a root is reduced. 
However, the critical slowing down of BA-IB \citep{agmon2021critical} is avoided since we reduce the root before invoking BA-IB (see Section \ref{sub:continuous-IB-bifs}). 
The complexity is only linear in $|\mathcal{X}|$ thanks to the choice of decoder coordinates. Had we chosen one of the other coordinate systems in Section \ref{sec:coords-exchange-for-the-IB}, then solving the ODE would have been cubic in $|\mathcal{X}|$ rather than linear (see there).
The \textbf{computational difficulty} in following IB roots stems from the existence of bifurcations (Section \ref{sec:euler-method}), as it generally is with following an operator's root, \cite[Section 7.2]{agmon2022RTRD}. 

As noted in Section \ref{sec:euler-method}, \textbf{convergence guarantees} can be derived for Euler's method for the IB when away of bifurcation, in terms of the step-size $|\Delta \beta|$, in a manner similar to \cite[Theorem 5]{agmon2022RTRD} for RD. 
These imply similar guarantees for the IBRT1 Algorithm \ref{algo:IBRT1}, as adding a single BA-IB iteration in our modified Euler method improves its order (see there). 
These details are omitted for brevity, however. 

For a numerical method of order $d > 0$ (see Section \ref{sec:euler-method}) with a fixed step-size $|\Delta \beta|$ and a fixed computational cost per grid point, the \textbf{cost-to-error tradeoff} is given by
\begin{equation}
	error \propto cost^{-d} \;,
\end{equation}
as in \cite[Equation (3.6)]{agmon2022RTRD}, when $|\Delta \beta|$ is small enough. See \cite{butcher2016numerical} for example. 
Figure 3.4 in \citep{agmon2022RTRD} demonstrates for RD that methods of higher order achieve a better tradeoff, as expected, as in the fixed-order Taylor methods they employ. 
Since computing implicit derivatives of higher orders requires the calculation of many more derivative tensors of $Id - BA_\beta$ \eqref{eq:IB-operator-def} than done here, \cite[Section 2.2]{agmon2022RTRD}, we have used only first-order derivatives for simplicity. 
However, while the vanilla Euler method for the IB is of order $d = 1$, the discussion in Section \ref{sec:euler-method} (and Figure \ref{fig:Euler-method-for-IB-dec} in particular) suggests that the order $d$ of the modified Euler method used by the IBRT1 Algorithm \ref{algo:IBRT1} is nearly twice than that. 
cf., Section \ref{sub:IBRT1-numerical-results}. 

\medskip 
With that, we comment on the behavior of the IBRT1 Algorithm \ref{algo:IBRT1} at \textbf{discontinuous bifurcations}.
Consider the problem in Figure \ref{fig:support-switching-bif-in-IB} (Section \ref{sub:discontinuous-IB-bifs}), for example. 
When Algorithm \ref{algo:IBRT1} follows the optimal 2-clustered root there, the Jacobian's singularity (in Figure \ref{fig:support-switching-bif-eigs}) is detectable by it because the step size $\Delta	\beta$ is negative. 
cf., the discussion in Section \ref{sub:discontinuous-IB-bifs} there. 
Indeed, due to Conjecture \ref{conj:BA-IB-Jacob-in-decoder-coords-is-nonsingular-at-reduced-root} ff., the algorithm can detect discontinuous bifurcations in general. 
Whether a particular discontinuous bifurcation is detected by Algorithm \ref{algo:IBRT1} in practice depends on the details\footnote{ e.g., on the threshold value $\delta_3$ for detecting singularity and on the precise grid points layout.}, of course, as with continuous bifurcations. 
Indeed, the details may or may not cause a particular example to be detected by the conditions on lines \ref{algo:IBRT1:S-matrix-to-compute-kernel} and \ref{algo:IBRT1:singularity-test} (in Algorithm \ref{algo:IBRT1}). 
If missed, Algorithm \ref{algo:IBRT1} will continue to follow the 2-clustered root in Figure \ref{fig:support-switching-bif-in-IB} to the left of the bifurcation, where it is sub-optimal, just as BA-IB with reverse deterministic annealing would. 
Once detected, though, one may wonder whether the heuristic Algorithm \ref{algo:handle-ODE-singularity} works well also for discontinuous bifurcations. 
The example of Figure \ref{fig:support-switching-bif-in-IB} has just one single-clustered root to the left of the bifurcation. 
Thus, the BA-IB Algorithm \ref{algo:BA-IB} invoked on line \ref{algo:handle-ODE-singularity:BA-IB-after-merging} (of Algorithm \ref{algo:handle-ODE-singularity}) must converge to it. 
However, there may generally be more than a single root of smaller effective cardinality to the left of the bifurcation, to which BA-IB may converge. 
The handling of discontinuous bifurcations is left to future work. 
Such handling is expected to be easier in the IB than in RD. 
Since, in contrast to RD, the effective cardinality of an optimal IB root cannot decrease with $\beta$ (bottom of Section \ref{sub:discontinuous-IB-bifs}). 
See \cite[Problems 2.8-2.10]{berger71} for counter-examples in RD. 
This makes detecting discontinuous bifurcations easier in the IB and is also expected to assist with their handling. 

%

\medskip 
We list the \textbf{assumptions} used along the way for reference. 
These are needed to guarantee the optimality of the IBRT1 Algorithm \ref{algo:IBRT1} at the limit of small step-sizes $|\Delta \beta |$, except at a bifurcation's vicinity. 
In Section \ref{sec:introduction}, it was assumed without loss of generality\footnote{ Otherwise, one may remove symbols $\x$ with $p_X(\x) = 0$ from the source alphabet.} that the input distribution $p_X$ is of full support, $p(\x) > 0$ for every $\x$. 
The requirement $p(\y | \x) > 0$ was added in Section \ref{sec:IB-ODE} as a sufficient technical condition for exchanging to logarithmic coordinates (Lemma \ref{lemma:sufficient-condition-for-IB-decoder-to-never-vanish} in Appendix \ref{sec:BA-IB-in-decoder-coords-appendix}), and could perhaps be alleviated in alternative derivations. 
Together, these are equivalent to having a never-vanishing IB problem definition, $p(\y|\x) p(\x) > 0 $ for every $\x$ and $\y$. 
The algorithm's initial condition is assumed to be a reduced and optimal IB root, as reduction is needed by Conjecture \ref{conj:BA-IB-Jacob-in-decoder-coords-is-nonsingular-at-reduced-root} in Section \ref{sub:IB-as-an-RD-problem-and-non-singularity-conj}. 
Finally, the given IB problem is assumed to have only continuous bifurcations, except perhaps for its first (leftmost) one. 
While these assumptions are sufficient to guarantee optimality, we note that milder conditions might do in a particular problem.

\newpage
\section{Concluding remarks}
\label{sec:discussion}

The IB is intimately related to several problems in adjacent fields, \cite{zaidi2020information}, including coding problems, inference, and representation learning. 
Despite its importance, there are surprisingly few techniques to solve it numerically. 
This work attempts to fill this gap by exploiting the dynamics of IB roots.

The end result of this work is a new numerical algorithm for the IB, which follows the path of a root along the IB's optimal tradeoff curve \eqref{eq:IB-curve-def}. 
A combination of several novelties was required to achieve this goal. 
First, the dynamics underlying the IB-curve \eqref{eq:IB-curve-def} obeys an ODE, \cite{agmon2022thesis}. 
Following the discussion around Conjecture \ref{conj:BA-IB-Jacob-in-decoder-coords-is-nonsingular-at-reduced-root} (in Section \ref{sub:IB-as-an-RD-problem-and-non-singularity-conj}), the existence of such a dynamics stems from the analyticity of the IB's fixed-point Equations \eqref{eq:IB-eq-encoder}-\eqref{eq:IB-eq-marginal}, thus typically resulting in piece-wise smooth dynamics of IB roots. 
Several natural choices of a coordinate system for the IB were considered, both for computational purposes and to facilitate a clean treatment of IB bifurcations below. 
The IB's ODE \eqref{eq:IB-beta-ODE-in-decoder-coords} was derived anew in appropriate coordinates, allowing an efficient computation of implicit derivatives at an IB root. 
Combining BA-IB with Euler's method yields a modified numerical method whose order is higher than either. 

Second, one needs to understand where the IB ODE \eqref{eq:IB-beta-ODE-in-decoder-coords} is \textit{not} obeyed, thereby violating the differentiability of an optimal root with respect to $\beta$. 
To that end, one not only needs to detect IB bifurcations but also needs to identify their type in order to handle them properly. 
Unlike standard techniques, our approach is to remove redundant coordinates, following root-tracking for RD, \cite{agmon2022RTRD}; cf., Section \ref{sec:introduction}. 
To achieve a reduction, we follow the arguably better definition of the IB in \cite{harremoes2007information}. 
Namely, a finite IB problem is an RD problem on the continuous reproduction alphabet $\Delta[\mathcal{Y}]$. 
Therefore, the IB may be intuitively considered as a method of lossy compression of the information on $Y$ embedded in $X$. 
Viewing a finite IB problem as an infinite RD problem suggests a particular choice of a coordinate system for the IB, which enables reduction in the IB; this extends reduction in RD, \cite{agmon2022RTRD}. 
Furthermore, this point of view highlights subtleties due to computing finite-dimensional representations of IB roots. 
To our understanding, these subtleties hindered the understanding of IB bifurcations throughout the years. 

Combining the above allows us to translate an understanding of IB bifurcations to a new numerical algorithm for the IB (the IBRT1 Algorithm \ref{algo:IBRT1}). 
There are several directions that one could consider to improve our algorithm. 
Near bifurcations, one could improve its handling of discontinuous bifurcations. 
While we used implicit derivatives only of the first order for simplicity, higher-order derivatives generally offer a better cost-to-error tradeoff when away of bifurcations. 
See also \cite[Section 3.4]{agmon2022RTRD} on possible improvements for following an operator's root.

\newpage
\appendix
\part*{Appendix}

\medskip
\section{The BA-IB operator in decoder coordinates}
\label{sec:BA-IB-in-decoder-coords-appendix}

For reference, we give an explicit expression for the BA-IB operator in decoder coordinates, defined in Section \ref{sec:coords-exchange-for-the-IB}.

\medskip
Denote by $\bm{p}_{Y|\hat{X}}$ and $\bm{p}_{\hat{X}}$ the vectors whose coordinates are $\big( p(\y|\xhat), p(\xhat) \big)$.
We denote the evaluation of $BA_\beta$ at this point by $BA_\beta[\bm{p}_{Y|\hat{X}}, \bm{p}_{\hat{X}}]$. 
Its output is again a decoder-marginal pair, whose coordinates are denoted respectively $BA_\beta[\bm{p}_{Y|\hat{X}}, \bm{p}_{\hat{X}}](\y|\xhat)$ and $BA_\beta[\bm{p}_{Y|\hat{X}}, \bm{p}_{\hat{X}}](\xhat)$.
Explicitly, $BA_\beta$ in decoder coordinates is given by,
\begin{equation}		\label{eq:BA-operator-def-in-decoder-coords-appendix}
\begin{split}
	BA_\beta[\bm{p}_{Y|\hat{X}}, \bm{p}_{\hat{X}}](\y|\xhat) &:= 
	\sum_{\x} \frac{p(\y|\x)  p(\x)}{Z(\x, \beta)} \exp \big\{ -\beta \; D_{KL}\big[p(\yp|\x) || p(\yp|\xhat)\big] \big\}
	\quad \text{and}
	\\
	BA_\beta[\bm{p}_{Y|\hat{X}}, \bm{p}_{\hat{X}}](\xhat) &:= 
	\sum_{\x} \frac{p(\xhat) p(\x) }{Z(\x, \beta)} \exp \big\{ -\beta\;D_{KL}\big[p(\yp|\x) || p(\yp|\xhat)\big] \big\} \;,
\end{split}
\end{equation}
where $Z(\x, \beta)$ is defined in terms of $p(\y|\xhat)$ and $p(\xhat)$ as in the IB's encoder Equation \eqref{eq:IB-eq-encoder} (Section \ref{sec:introduction}).

\medskip
The following lemma is handy when exchanging to logarithmic coordinates in Section \ref{sec:IB-ODE}.
\begin{lemma}		\label{lemma:sufficient-condition-for-IB-decoder-to-never-vanish}
	Let $p(\y|\x) p(\x)$ define a finite IB problem, such that $p(\y|\x) > 0$ for every $\x$ and $\y$. 
	Let $p(\y|\xhat)$ be the decoder of an IB root, and $\xhatp$ such that $p(\xhatp) > 0$. 
	Then $p(\y|\xhatp) > 0$ for every $\y$.
\end{lemma}

\begin{proof}[Proof of Lemma \ref{lemma:sufficient-condition-for-IB-decoder-to-never-vanish}]
	This follows immediately from the IB's decoder Equation \eqref{eq:IB-eq-decoder}, since $p(\x|\xhatp)$ is a well-defined normalized conditional probability distribution if $p(\xhatp) > 0$. 
\end{proof}

\medskip
\section{The first-order derivative tensors of Blahut-Arimoto for the IB}
\label{sec:first-order-deriv-tensors-of-BA-IB-appendix}

We calculate the first-order derivative tensors of the Blahut-Arimoto operator $BA_\beta$ in log-decoder coordinates (see Sections \ref{sec:coords-exchange-for-the-IB} and \ref{sec:IB-ODE}).
Namely, its Jacobian matrix $D_{\log p(\xhat|\x), \log p(\xhat)} BA_\beta$, and the vector $D_{\beta} BA_\beta$ of its partial derivatives with respect to $\beta$. 
cf., Appendix \ref{sec:BA-IB-in-decoder-coords-appendix} for explicit formulae of $BA_\beta$ in decoder coordinates. 

While these are ``just'' differentiations, many subtleties are involved in getting the math right. 
For example, one needs to correctly identify the inputs and outputs of $BA_\beta$, when considered as an operator on log-decoder coordinates. 
For another, one must take special care as to which variable depends on which, and especially on which does it not depend, as multiple variables are involved. 
Above all, these calculations require a deep understanding of the chain rule. 
With that, a common caveat in such calculations is that the $BA_\beta$ operator (and the equations defining it) should be differentiated \textit{before} they are evaluated. 
While this is obvious for real functions, where $f'(3)$ stands for the derivative function of $f(x)$ evaluated at $x=3$, for the $BA_\beta$ operator, this might get obfuscated by the myriad of variables and variable-dependencies of which it is comprised. 
Although calculating the derivative of $BA_\beta$ (at an arbitrary point) first and only then evaluating at a fixed point might appear as a mere technical necessity, it is required by this work. For example, when considering the vector field defined by the IB operator \eqref{eq:IB-operator-def} at Section \ref{sub:IBRT1-discussion}. 
cf., \cite[Section 5]{agmon2022RTRD}, for the derivative tensors of Blahut's algorithm \citep{blahut1972} for RD, of arbitrary order. 

The subtitles involved in these differentiations are discussed in Appendix \ref{sub:BA-IB-jacobian-appendix:calculation-setup-and-goals}, with the bulk of the calculations carried out in \ref{subsub:BA-IB-jacobian-appendix:differentiating-along-depend-graph}. 
The latter are gathered and simplified in Appendix \ref{sub:BA-IB-jacobian-appendix:decoder-deriv-matrix} to obtain the Jacobian matrix $D_{\log p(\xhat|\x), \log p(\xhat)} BA_\beta$, and in Appendix \ref{sub:BA-IB-beta-deriv-appendix-wrt-decoder-coords} to obtain the partial-derivatives vector $D_{\beta} BA_\beta$. 
The results provided here naturally depend on the choice of coordinate system. 
To compare results between log-decoder and log-encoder coordinates in Section \ref{sec:coords-exchange-for-the-IB} (e.g., in Figure \ref{fig:norm-of-analytical-derivs-for-BSC}), we derive in Appendix \ref{sub:coordinates-exchange-jacobians-appendix} the coordinate-exchange Jacobians between these coordinate systems.

\medskip
\subsection{Calculation setups and partial derivatives of unnamed functions}
\label{sub:BA-IB-jacobian-appendix:calculation-setup-and-goals}

We explain the mathematical subtitles relevant to the sequel. 

\medskip 
As we are interested in the derivatives of the Blahut-Arimoto Algorithm \ref{algo:BA-IB} for the IB (in Section \ref{sec:introduction}), we shall follow its notation. 
Namely, distributions are subscripted $i$ or $i+1$ by the algorithm's iteration number. 
A subscript $i$ is usually considered an input distribution, and a subscript $i+1$ is usually considered an output distribution. 
e.g., $p_i(\xhat)$ or $p_{i+1}(\y|\xhat)$. 
These need \textit{not} be IB roots but rather are arbitrary distributions. 
On the other hand, a subscript $\beta$ denotes a distribution of an IB root at a tradeoff value $\beta$, as in $p_{\beta}(\y|\xhat)$ for a root's decoders. 
To avoid subtleties due to zero-mass clusters, we usually assume $p_i(\xhat) \neq 0$ in the sequel, for any $\xhat$. 
cf., Sections \ref{sec:coords-exchange-for-the-IB} and \ref{sub:IB-as-an-RD-problem-and-non-singularity-conj} on root-reduction in the IB. 

It is important to distinguish which variables are dependent and which are independent in a particular calculation. 
e.g., in Appendix \ref{sub:BA-IB-beta-deriv-appendix-wrt-decoder-coords}. 
Since this task is easier for a single real variable (as opposed to distributions, for example), we consider simplifications to the real case. 
Note that each of the equations \algref{algo:BA-IB}{eq:IB-BA-cluster_marginal} through \algref{algo:BA-IB}{eq:IB-BA-new-direct-enc} defining the BA-IB Algorithm \ref{algo:BA-IB} yields a new distribution in terms of already-specified ones. 
These define \textit{unnamed functions}, whose variables and values are probability distributions. For example, one could have formally defined $p_i(x|\hat{x})$ in \algref{algo:BA-IB}{eq:IB-BA-bayes-for-computing-inverse-enc} by the function
\begin{equation}			\label{eq:bayes-for-inverse-enc-formal}
	\mathcal{F}\left[ p_i(\hat{x}|x), p_i(\hat{x}) \right]\left(x, \hat{x}\right) := 
	\nicefrac{p_i(\hat{x}|x)p(\x)}{p_i(\hat{x})} \;,
\end{equation}
where $p_i(\hat{x}|x)$ and $p_i(\hat{x})$ are the variables of $\mathcal{F}$, and its output is a conditional probability distribution, with $x$ conditioned upon $\hat{x}$. 
As the input and representation alphabets $\mathcal{X}$ and $\hat{\mathcal{X}}$ are finite, $N := |\mathcal{X}|$ and $T := |\hat{\mathcal{X}}|$, the arguments $p_i(\hat{x}|x), p_i(\hat{x})$ and values $p_i(x|\hat{x})$ of $\mathcal{F}$ \eqref{eq:bayes-for-inverse-enc-formal} are merely real vectors. 
Thus, enumerating the variables $x_1, \dots, x_N$ and $\hat{x}_1, \dots, \hat{x}_T$ allows to spell-out \eqref{eq:bayes-for-inverse-enc-formal} by its coordinates, 
\begin{equation}			\label{eq:bayes-for-inv-enc-formal-explicit}
	\mathcal{F}\left[ p_i(\hat{x}_1|x_1), p_i(\hat{x}_1|x_2), \dots, p_i(\hat{x}_1|x_N), \dots, p_i(\hat{x}_T|x_N),  p_i(\hat{x}_1), \dots, p_i(\hat{x}_T) \right]\left(x, \hat{x}\right) :=
	\nicefrac{p_i(\hat{x}|x)p(\x)}{p_i(\hat{x})} \;.
\end{equation}
While \eqref{eq:bayes-for-inv-enc-formal-explicit} is too cumbersome to work with, it does highlight that $\mathcal{F}$ is merely a vector of $N\cdot T$ real vector-valued functions, in $T + N\cdot T$ real variables. 
This allows us to use partial derivatives rather than their infinite-dimensional counterparts (namely, variational derivatives), as in 
\begin{equation}			\label{eq:partial-deriv-formal-wrt-encoder-coord}
	\frac{\partial \mathcal{F}\left[ p_i(\hat{x}|x), p_i(\hat{x}) \right]}{\partial p_i(\hat{x}_j|x_k)} := 
	\lim_{h\to 0} \frac{\mathcal{F}\left[p_i(\hat{x}_1|x_1), \dots, p_i(\hat{x}_j|x_k) + h, \dots, p_i(\hat{x}_T) \right] - \mathcal{F}\left[\dots, p_i(\hat{x}_j|x_k), \dots \right]}{h} \;.
\end{equation}
This is the derivative of $\mathcal{F}$ \eqref{eq:bayes-for-inv-enc-formal-explicit} with respect to a particular $(j, k)$-entry of its argument, by definition. 
However, to maintain a concise notation, we shall carry on with un-named function definitions, writing $\nicefrac{\partial p_i(x|\hat{x})}{\partial p_i(\hat{x}_j|x_k)}$ for the partial derivative of \eqref{eq:bayes-for-inverse-enc-formal} rather than its explicit form \eqref{eq:partial-deriv-formal-wrt-encoder-coord}.
If disoriented, the reader is encouraged to return to the definitions \eqref{eq:partial-deriv-formal-wrt-encoder-coord}. 

We often exchange variables implicitly to logarithmic coordinates, as in Section \ref{sec:IB-ODE}. 
For example, $\frac{\partial \mathcal{F}\left[ p_i(\hat{x}|x), p_i(\hat{x}) \right]}{\plog p_i(\hat{x}_i|x_j)}$ is to be understood as exchanging variables to $u_i(\hat{x}, x) := \log p_i(\hat{x}|x)$, with $\mathcal{G}\left[u_i(\hat{x},x), u_i(\hat{x}) \right] := \mathcal{F}\left[ \exp u_i(\hat{x},x), \exp u_i(\hat{x}) \right]$ now differentiated with respect to its variables $u_i(\hat{x},x)$ and $u_i(\hat{x})$,
\begin{equation}
	\frac{\partial \mathcal{F}\left[ p_i(\hat{x}|x), p_i(\hat{x}) \right]}{\plog p_i(\hat{x}_i|x_j)} = \frac{\partial \mathcal{F}\left[ \exp u_i(\hat{x},x), \exp u_i(\hat{x}) \right]}{\partial u_i(\hat{x},x)} =:
	\frac{\partial \mathcal{G}\left[u_i(\hat{x},x), u_i(\hat{x}) \right]}{\partial u_i(\hat{x},x)} 
\end{equation}
The output of $\mathcal{F}$ may similarly be exchanged to logarithmic coordinates, as in $\log \mathcal{F}\left[ \exp u_i(\hat{x},x), \exp u_i(\hat{x}) \right]$. 

To proceed, carefully note the dependencies between the various variables in a BA-IB iteration, at \algref{algo:BA-IB}{eq:IB-BA-cluster_marginal} through \algref{algo:BA-IB}{eq:IB-BA-new-direct-enc}. These are summarized compactly by the following diagram,
\begin{equation}			\label{eq:dependencies-graph-for-BA-IB-variables}
	\xymatrix@C=1.8em{	
		\\
		\dots\ar[r] & p_i(\hat{x}|x)\ar[r]\ar@(u,u)[rr]|(0.72)\hole & p_i(\hat{x})\ar[r]\ar@(u,u)[rrr]\ar@(d,l)[rrd] & p_i(x|\hat{x})\ar[r] & p_i(y|\hat{x})\ar[r]\ar[d] & p_{i+1}(\hat{x}|x)\ar[r] & \dots \\
		&  &  &  & Z_i(x, \beta)\ar@(r,d)[ur]
	}
\end{equation}
by their order of appearance in the BA-IB Algorithm \ref{algo:BA-IB}. 
This diagram proceeds to both sides by the iteration number $i$. 
Each node in \eqref{eq:dependencies-graph-for-BA-IB-variables} serves both as a function of the nodes preceding it and as a variable for those succeeding it, and so it is a ``function-variable''. 

To differentiate along the dependencies graph \eqref{eq:dependencies-graph-for-BA-IB-variables}, we shall need the multivariate chain rule 
\begin{equation}			\label{eq:multivariate-chain-rule}
	\frac{df}{dy} = \frac{\partial f}{\partial y} + \frac{\partial f}{\partial z}\frac{d z}{d y} \;,
\end{equation}
for a function $f\big(y, z(y)\big)$.
As the dependencies graph \eqref{eq:dependencies-graph-for-BA-IB-variables} involves multiple function-variables, such as $z(y)$, we pause on the definition's subtleties. 
The partial derivative of a function $g$ in several variables $x_1, \dots, x_N$ with respect to its $i$-th entry is defined by
\begin{equation}			\label{eq:partial-deriv-def}
	\frac{\partial g}{\partial x_i} := \lim_{h\to 0} \frac{g(x_1, \dots, x_i + h, \dots, x_N) - g(x_1, \dots, x_i, \dots, x_N)}{h} \;.
\end{equation}
We emphasize that variables $x_1, \dots, x_{i-1}, x_{i+1}, \dots, x_N$ \textit{other} than $x_i$ are fixed when calculating $\frac{\partial g}{\partial x_i}$. And so, it makes no difference in \eqref{eq:partial-deriv-def} whether or not they depend on $x_i$, as in $x_j = x_j(x_i)$ for $j\neq i$. 

Next, suppose we would like to calculate how changing an input distribution affects some output distribution. 
This is relevant in Appendix \ref{sub:BA-IB-jacobian-appendix:decoder-deriv-matrix} for example, when considering how does a change in a coordinate of an input decoder $p_i(\y|\xhat)$ or marginal $p_i(\xhat)$ affect a particular coordinate of the output decoder or marginal. 
For exposition's simplicity, though, suppose that we would like to calculate how a change in the $(k_1, k_2)$ coordinate $p_i(\hat{x}_{k_1}|x_{k_2})$ of an input encoder affects the $(j_1, j_2)$ coordinate $p_{i+1}(\hat{x}_{j_1}|x_{j_2})$ of the output encoder. 
That is, deriving the rightmost node in \eqref{eq:dependencies-graph-for-BA-IB-variables} with respect to a coordinate of the leftmost one, 
\begin{equation}			\label{eq:computation-goal-log-encoder-deriv}
	\frac{d\log p_{i+1}(\hat{x}_{j_1}|x_{j_2})}{d\log p_i(\hat{x}_{k_1}|x_{k_2})} \;,
\end{equation}
where we have exchanged to logarithmic coordinates to simplify calculations. 
To calculate \eqref{eq:computation-goal-log-encoder-deriv}, one needs to apply the multivariate chain rule \eqref{eq:multivariate-chain-rule} along all the possible dependencies of the output $\log p_{i+1}(\hat{x}_{j_1}|x_{j_2})$ on the input coordinate $\log p_i(\hat{x}_{k_1}|x_{k_2})$. 
This amounts to following all the paths in \eqref{eq:dependencies-graph-for-BA-IB-variables} connecting these two nodes, summing the contributions of \textit{every} possible path. 
For example, traversing from the input $p_i(\hat{x}_{k_1}|x_{k_2})$ rightwards at \eqref{eq:dependencies-graph-for-BA-IB-variables} to $p_i(\hat{x})$, then downwards to $Z_i(x, \beta)$ and then to the output $p_{i+1}(\hat{x}_{j_1}|x_{j_2})$ yields the term
\begin{equation*}
	\frac{\partial\log p_i(\hat{x}'')}{\partial\log p_i(\hat{x}_{k_1}| x_{k_2})}
	\frac{\partial\log Z_i(x, \beta)}{\partial\log p_i(\hat{x}'')} 
	\frac{\plog p_{i+1}(\hat{x}_{j_1}|x_{j_2})}{\plog Z_i(x, \beta)} 
\end{equation*}
corresponding to this path, at particular $x$ and $\hat{x}''$ coordinates.
To collect the contribution from every intermediate function-variable coordinate, we need to sum the latter over $x$ and $\hat{x}''$. 
Writing down all such paths, one has for \eqref{eq:computation-goal-log-encoder-deriv},
\begin{multline}			\label{eq:log-encoder-deriv-chain-rule}
	\frac{d\log p_{i+1}(\hat{x}_{j_1}|x_{j_2})}{d\log p_i(\hat{x}_{k_1} | x_{k_2})} \\ = 
	\frac{\partial\log p_i(\hat{x}'')}{\partial\log p_i(\hat{x}_{k_1}| x_{k_2})} \cdot \Bigg\{ \frac{\plog p_{i+1}(\hat{x}_{j_1}|x_{j_2})}{\plog Z_i(x, \beta)} \cdot \frac{\partial\log Z_i(x, \beta)}{\partial\log p_i(\hat{x}'')} 
	+ \frac{\partial\log p_{i+1}(\hat{x}_{j_1} | x_{j_2})}{\partial\log p_i(\hat{x}'')} \\
	+ \left[\frac{\plog p_{i+1}(\hat{x}_{j_1}|x_{j_2})}{\plog Z_i(x, \beta)} \cdot \frac{\plog Z_i(x, \beta)}{\plog p_i(y|\hat{x})} + \frac{\plog p_{i+1}(\hat{x}_{j_1}|x_{j_2})}{\plog p_i(y|\hat{x})} \right] \cdot
	\frac{\plog p_i(y|\hat{x})}{\plog p_i(x'|\hat{x}')}\cdot \frac{\plog p_i(x'|\hat{x}')}{\plog p_i(\hat{x}'')} \Bigg\} \\
	+ \left[\frac{\plog p_{i+1}(\hat{x}_{j_1}|x_{j_2})}{\plog Z_i(x, \beta)} \cdot \frac{\plog Z_i(x, \beta)}{\plog p_i(y|\hat{x})} + \frac{\plog p_{i+1}(\hat{x}_{j_1}|x_{j_2})}{\plog p_i(y|\hat{x})} \right] \cdot
	\frac{\plog p_i(y|\hat{x})}{\plog p_i(x'|\hat{x}')}\cdot \frac{\plog p_i(x'|\hat{x}')}{\plog p_i(\hat{x}_{k_1}|x_{k_2})}
\end{multline}
Repeated unbounded variables are understood to be summed over, as in Einstein's summation convention.

\medskip
\subsubsection{Differentiating along the dependencies graph}

\label{subsub:BA-IB-jacobian-appendix:differentiating-along-depend-graph}

Next, we differentiate each edge in (the logarithm of) the dependency graph \eqref{eq:dependencies-graph-for-BA-IB-variables}. 
These are necessary to evaluate derivatives along dependency paths, that underlie the subsequent sections' calculations.

\medskip
Equation \algref{algo:BA-IB}{eq:IB-BA-cluster_marginal} in the BA-IB Algorithm \ref{algo:BA-IB} defines the cluster marginal in terms of the direct encoder,
\begin{multline}			\label{eq:cluster-marginal-by-direct-enc-deriv}
	\frac{\plog p_{i}(\hat{x})}{\plog p_{i}(\hat{x}'|x')} \overset{\algref{algo:BA-IB}{eq:IB-BA-cluster_marginal}}{=}
	\frac{1}{p_{i}(\hat{x})} \sum_{\x} p(\x) \frac{\partial }{\plog p_{i}(\hat{x}'|x')} p_{i}(\hat{x}|x) \\ =
	\frac{1}{p_{i}(\hat{x})} \sum_{\x} p(\x) p_{i}(\hat{x}|x) \frac{\plog p_{i}(\hat{x}|x)}{\plog p_{i}(\hat{x}'|x')} \overset{\algref{algo:BA-IB}{eq:IB-BA-bayes-for-computing-inverse-enc}}{=}
	p_{i}(x'|\hat{x}) \cdot \delta_{\hat{x}, \hat{x}'}
\end{multline}
In the first and second equalities we have used the identity $\tfrac{\partial}{\partial x} y = y \tfrac{\partial}{\partial x} \log y$ for the differentiation of a function's logarithm, when $y$ is a function of $x$.

Following the comments around the definition \eqref{eq:partial-deriv-def} of a partial derivative, note that  \algref{algo:BA-IB}{eq:IB-BA-bayes-for-computing-inverse-enc} defines the inverse encoder $\log p_{i}(x|\hat{x})$ as a function of the variables $\log p_{i}(\hat{x}|x)$ and $\log p_{i}(\hat{x})$ (and $p(\x)$, which we ignore under differentiation).
Thus, differentiating this equation with respect to an entry of the variable $\log p_{i}(\hat{x}|x)$ implies that the entries of the other variable $\log p_{i}(\hat{x})$ are held fixed, and vice versa. 
So, for the Bayes rule \algref{algo:BA-IB}{eq:IB-BA-bayes-for-computing-inverse-enc} we have
\begin{equation}			\label{eq:inv-encoder-deriv-wrt-cluster-marginal}
	\frac{\plog p_{i}(x_{j_1}|\hat{x}_{j_2})}{\plog p_{i}(\hat{x})} = \frac{\partial }{\plog p_{i}(\hat{x})} \big[ \cancel{\log p_{i}(\hat{x}_{j_2}|x_{j_1})} - \log p_{i}(\hat{x}_{j_2}) \big] = -\delta_{\hat{x}, \hat{x}_{j_2}}
\end{equation}
where $\log p_{i}(\hat{x}_{j_2}|x_{j_1})$ at the right-hand side is different from the variable $\log p_{i}(\hat{x})$ of differentiation, and so its partial derivative vanishes.
Next, differentiating \algref{algo:BA-IB}{eq:IB-BA-bayes-for-computing-inverse-enc} with respect to a coordinate of its other variable $\log p_{i}(\hat{x}|x)$,
\begin{equation}
	\frac{\plog p_{i}(x_{j_1}|\hat{x}_{j_2})}{\plog p_{i}(\hat{x}'|x')} = 
	\frac{\plog p_{i}(\hat{x}_{j_2}|x_{j_1})}{\plog p_{i}(\hat{x}'|x')} - \cancel{\frac{\plog p_{i}(\hat{x}_{j_2})}{\plog p_{i}(\hat{x}'|x')}} = \delta_{x_{j_1}, x'} \cdot \delta_{\hat{x}_{j_2}, \hat{x}'}
\end{equation}

Using again the logarithmic derivative identity $\tfrac{\partial}{\partial x} y = y \tfrac{\partial}{\partial x} \log y$, by the decoder Equation \algref{algo:BA-IB}{eq:IB-BA-decoder-eq} we have
\begin{multline}			\label{eq:deriving-decider-by-inv-enc}
	\frac{\plog p_{i}(y|\hat{x}'')}{\plog p_i(x_{k_1}|\hat{x}_{k_2})} =
	\frac{1}{p_{i}(y|\hat{x}'')} \sum_{x'''} p(y|x''') \frac{\partial }{\plog p_i(x_{k_1}|\hat{x}_{k_2})} p_i(x'''|\hat{x}'') \\
	= \frac{1}{p_{i}(y|\hat{x}'')} \sum_{x'''} p(y|x''') p_i(x'''|\hat{x}'') \; \delta_{\hat{x}_{k_2}, \hat{x}''} \cdot \delta_{x_{k_1}, x'''}
	= \delta_{\hat{x}_{k_2}, \hat{x}''}  \cdot \frac{p(y|x_{k_1}) p_i(x_{k_1}|\hat{x}'')}{p_{i}(y|\hat{x}'')}
\end{multline}

Next, consider the KL-divergence term in the definition \algref{algo:BA-IB}{eq:IB-BA-partition-func} of the partition function $Z_i$,
\begin{equation}			\label{eq:KL-decoder-deriv}
	\frac{\partial }{\plog p_{i}(y|\hat{x}'')} D_{KL}\big[p(y|x'') || p_{i}(y|\hat{x})\big] 
	= -\sum_{y'} p(y'|x'') \underset{\delta_{\hat{x}, \hat{x}''}\cdot \delta_{y, y'}}{\underbrace{\frac{\partial }{\plog p_{i}(y|\hat{x}'')} \log p_{i}(y'|\hat{x})}}
	= -\delta_{\hat{x}, \hat{x}''} \cdot p(y|x'')
\end{equation}
Since the partition function \algref{algo:BA-IB}{eq:IB-BA-partition-func} depends on the decoder $p_i(y|\hat{x})$ only via the KL-divergence,
\begin{multline}
	\frac{\partial Z_{i}(x'', \beta)}{\plog p_{i}(y|\hat{x}'')} = 
	\frac{\partial }{\plog p_{i}(y|\hat{x}'')}
	\sum_{\hat{x}} p_{i}(\hat{x}) \exp{\big\{ -\beta \; D_{KL}\big[p(y|x'') || p_{i}(y|\hat{x})\big] \big\}} \\ =
	-\beta \sum_{\hat{x}} p_{i}(\hat{x}) \exp{\big\{ -\beta \; D_{KL}\big[p(y|x'') || p_{i}(y|\hat{x})\big] \big\}} \frac{\partial }{\plog p_{i}(y|\hat{x}'')} D_{KL}\big[p(y|x'') || p_{i}(y|\hat{x})\big] \\ \overset{\eqref{eq:KL-decoder-deriv}}{=}
	\beta \; p_{i}(\hat{x}'') \exp{\big\{ -\beta \; D_{KL}\big[p(y|x'') || p_{i}(y|\hat{x}'')\big] \big\}} p(y|x'')
	\\ \overset{\algref{algo:BA-IB}{eq:IB-BA-new-direct-enc}}{=}
	\beta \; p_{i+1}(\hat{x}''|x'') Z_i(x'', \beta) p(y|x'')
\end{multline}
Hence,
\begin{equation}			\label{eq:partit_func_deriv_wrt_decoder}
	\frac{\plog Z_{i}(x'', \beta)}{\plog p_{i}(y|\hat{x}'')} = 
	\beta \; p_{i+1}(\hat{x}''|x'') p(y|x'')
\end{equation}

For the derivative of the partition function with respect to the marginal $p_i(\hat{x})$,
\begin{multline}
	\frac{\partial Z_{i}(x,\beta)}{\plog p_i(\hat{x}')} \overset{\algref{algo:BA-IB}{eq:IB-BA-partition-func}}{=}
	\frac{\partial}{\plog p_i(\hat{x}')}\sum_{\hat{x}} p_{i}(\hat{x}) \exp{\big\{ -\beta \; D_{KL}\big[p(\y|\x) || p_{i}(y|\hat{x})\big] \big\}} \\ =
	\sum_{\hat{x}} p_{i}(\hat{x}) \exp{\big\{ -\beta \; D_{KL}\big[p(\y|\x) || p_{i}(y|\hat{x})\big] \big\}} \frac{\plog p_i(\hat{x})}{\plog p_i(\hat{x}')} \\ =
	\sum_{\hat{x}} p_{i}(\hat{x}) \exp{\big\{ -\beta \; D_{KL}\big[p(\y|\x) || p_{i}(y|\hat{x})\big] \big\}} \cdot \delta_{\hat{x}, \hat{x}'} \overset{\algref{algo:BA-IB}{eq:IB-BA-new-direct-enc}}{=}
	Z_i(x, \beta) \cdot p_{i+1}(\hat{x}'|x)
\end{multline}
Where the second equality follows from the logarithmic derivative identity. Hence,
\begin{equation}			\label{eq:log-partit-func-wrt-marginal-deriv}
	\frac{\plog Z_{i}(x,\beta)}{\plog p_i(\hat{x}')} = p_{i+1}(\hat{x}'|x)
\end{equation}

Finally, for the encoder Equation \algref{algo:BA-IB}{eq:IB-BA-new-direct-enc},
\begin{equation}			\label{eq:deriving-direct-enc-wrt-dec}
	\log p_{i+1}(\hat{x}'|x') :=
	\log p_{i}(\hat{x}') - \log Z_{i}(x',\beta) - \beta \; D_{KL}\big[p(y|x') || p_{i}(y|\hat{x}')\big]
\end{equation}
The first two terms to the right, $p_{i}(\hat{x})$ and $Z_{i}(x, \beta)$, take the role of a variable in Equation \algref{algo:BA-IB}{eq:IB-BA-new-direct-enc}.
In contrast, we consider the last divergence term as a shorthand for summing over $p_{i}(y|\hat{x})$. Thus, the latter \textit{is} a variable of \eqref{eq:deriving-direct-enc-wrt-dec}.
With \eqref{eq:KL-decoder-deriv}, we thus have
\begin{equation}
	\frac{\plog p_{i+1}(\hat{x}'|x')}{\plog p_{i}(y|\hat{x}'')} = 
	\beta \; \delta_{\hat{x}', \hat{x}''} \cdot p(y|x') \;.
\end{equation}
For the other derivatives of the encoder equation \algref{algo:BA-IB}{eq:IB-BA-new-direct-enc},
\begin{equation}
	\frac{\plog p_{i+1}(\hat{x}'|x')}{\plog Z_{i}(x'', \beta)} = 
	-\frac{\partial\log Z_i(x', \beta)}{\plog Z_i(x'', \beta)} = -\delta_{x', x''}
\end{equation}
And,
\begin{equation}
	\frac{\plog p_{i+1}(\hat{x}|x)}{\plog p_i(\hat{x}')} = 
	\frac{\plog p_{i}(\hat{x})}{\plog p_i(\hat{x}')} - \cancel{\frac{\plog Z_{i}(x, \beta)}{\plog p_i(\hat{x}')}} -\beta \cancel{\frac{\partial D_{KL}\big[p(\y|\x) || p_{i}(y|\hat{x})\big]}{\plog p_i(\hat{x}')}} =
	\delta_{\hat{x}, \hat{x}'}
\end{equation}
where the variable $p_i(\hat{x})$ of Equation \algref{algo:BA-IB}{eq:IB-BA-new-direct-enc} differs from the variables $Z_i$ and $p_{i}(y|\hat{x})$, on which the crossed-out terms depend.

\medskip
We summarize the calculations of this subsection in the following diagram:
\begin{equation}			\label{eq:log-variable-dep-w-derivatives-no-normalization}
	\xymatrix@C=2em@R=4em{
		\log p_i(\hat{x}|x)\ar[rr]_{p_{i}(x'|\hat{x})\; \delta_{\hat{x}, \hat{x}'}}\ar@(u,u)[rrr]^{\delta_{x, x'} \; \delta_{\hat{x}, \hat{x}'}}|(0.75)\hole & &
		\log p_i(\hat{x})\ar[r]_{-\delta_{\hat{x}, \hat{x}'}}\ar@(u,u)[rrrrr]^{\delta_{\hat{x}, \hat{x}'}}\ar@(d,l)[rrrd]_{p_{i+1}(\hat{x}'|x)} & 
		\log p_i(x|\hat{x})\ar[rr]^{\frac{p(y|x') p_i(x'|\hat{x})}{p_{i}(y|\hat{x})} \cdot \delta_{\hat{x}', \hat{x}}} & & 
		\log p_i(y|\hat{x})\ar[rr]_{\beta \; \delta_{\hat{x}, \hat{x}'} p(y'|x)}\ar[d]_{\beta \; p_{i+1}(\hat{x}'|x) p(y'|x)} & &
		\log p_{i+1}(\hat{x}|x)\\
		&  &  &	&	& \log Z_i(x, \beta)\ar@(r,d)[urr]_{-\delta_{x, x'}}
	}
\end{equation} 
A differentiation variable is denoted with commas, at an arrow's source in this diagram. 
A coordinate of the function which we differentiate is written without commas, at an arrow's end. e.g.,
\begin{equation*}
	\xymatrix{
		\log p_i(\hat{x}|x)\ar[rrr]_{\frac{\plog p_i(x|\hat{x})}{\plog p_i(\hat{x}'|x')} \;=\; \dots} & & & \log p_i(x|\hat{x})
	}
\end{equation*}

\medskip
\subsection{The Jacobian matrix of BA-IB in log-decoder coordinates}

\label{sub:BA-IB-jacobian-appendix:decoder-deriv-matrix}

By gathering the results of Appendix \ref{subsub:BA-IB-jacobian-appendix:differentiating-along-depend-graph} and following the lines of \ref{sub:BA-IB-jacobian-appendix:calculation-setup-and-goals}, we calculate the Jacobian matrix \eqref{eq:BA-Jacob-wrt-decoder-coords-at-main-text} (in Section \ref{sec:IB-ODE}) of the Blahut-Arimoto operator $BA_\beta$ in log-decoder coordinates, defined in Section \ref{sec:coords-exchange-for-the-IB}. 

\medskip
The derivative of $BA_\beta$ in decoder coordinates boils down to the four quantities: 
the effect $\tfrac{d\log p_{i+1}(\y|\xhat)}{d \log p_i(\yp|\xhatp)}$ that varying a coordinate $\log p_i(\yp|\xhatp)$ of an input cluster has on a coordinate $\log p_{i+1}(\y|\xhat)$ of an output cluster, 
the effect $\frac{d\log p_{i+1}(\y|\xhat)}{d \log p_i(\xhatp)}$ that varying an input marginal coordinate $\log p_i(\xhatp)$ has on a coordinate $\log p_{i+1}(\y|\xhat)$ of an output cluster, and so forth. 
And so, the Jacobian $D_{\log p(\y|\xhat), \log p(\xhat)} BA_\beta$ it is a block matrix,
\begin{equation}		\label{eq:BA-Jacob-wrt-decoder-coords-as-block-matrix-implicit}
	\left(
		\begin{array}{c|c}
			\\	\frac{d\log p_{i+1}(\y|\xhat)}{d \log p_i(\yp|\xhatp)}	&
			\frac{d\log p_{i+1}(\y|\xhat)}{d \log p_i(\xhatp)}	\\	\\
			\hline	\\
			\frac{d\log p_{i+1}(\xhat)}{d \log p_i(\yp|\xhatp)}	&
			\frac{d\log p_{i+1}(\xhat)}{d \log p_i(\xhatp)}	\\	\;
		\end{array}
	\right)
\end{equation}
Its rows correspond to the output coordinates of $BA_\beta$.
We index its upper rows by $\y \in \mathcal{Y}$ and $\xhat \in \{1, \dots, T\}$, while its lower rows are indexed by $\xhat$ alone. Similarly, its columns correspond to the input coordinates of $BA_\beta$.
We index its leftmost columns by $\yp$ and $\xhatp$, and its rightmost columns by $\xhatp$ alone.
Each block in \eqref{eq:BA-Jacob-wrt-decoder-coords-as-block-matrix-implicit} is comprised of contributions along all the distinct paths connecting two vertices in the dependencies graph \eqref{eq:dependencies-graph-for-BA-IB-variables}.
For example, the lower-left block in \eqref{eq:BA-Jacob-wrt-decoder-coords-as-block-matrix-implicit} is comprised of the contributions along all the paths in \eqref{eq:dependencies-graph-for-BA-IB-variables} connecting $p_i(\yp|\xhatp)$ to $p_{i+1}(\xhat)$.

We now spell out the paths contributing to each block in \eqref{eq:BA-Jacob-wrt-decoder-coords-as-block-matrix-implicit}, with repeated dummy indices understood to be summed over.
Afterward, we shall calculate the contributing paths explicitly, carrying out the summations. The upper-left block of \eqref{eq:BA-Jacob-wrt-decoder-coords-as-block-matrix-implicit} is comprised of
\begin{multline}		\label{eq:BA-Jacob-wrt-decoder-coords-upper-left-block}
	\frac{d\log p_{i+1}(\y|\xhat)}{d \log p_i(\yp|\xhatp)} =
	\frac{\plog p_{i+1}(\y|\xhat)}{\plog p_{i+1}(x_1|\hat{x}_2)} \cdot \left[
		\frac{\plog p_{i+1}(x_1|\hat{x}_2)}{\plog p_{i+1}(\hat{x}_3)}
		\frac{\plog p_{i+1}(\hat{x}_3)}{\plog p_{i+1}(\hat{x}_4|x_5)} +
		\frac{\plog p_{i+1}(x_1|\hat{x}_2)}{\plog p_{i+1}(\hat{x}_4|x_5)}
	\right]	\\ \cdot \left[
		\frac{\plog p_{i+1}(\hat{x}_4|x_5)}{\plog p_i(\yp|\xhatp)} +
		\frac{\plog p_{i+1}(\hat{x}_4|x_5)}{\plog Z_i(x_6, \beta)}
		\frac{\plog Z_i(x_6, \beta)}{\plog p_i(\yp|\xhatp)}
	\right]
\end{multline}
This Equation \eqref{eq:BA-Jacob-wrt-decoder-coords-upper-left-block} encodes the fours paths connecting the vertex $p_i(\yp|\xhatp)$ to $p_{i+1}(\y|\xhat)$ in \eqref{eq:dependencies-graph-for-BA-IB-variables}. 
When accumulating the contributions in \eqref{eq:BA-Jacob-wrt-decoder-coords-upper-left-block}, one must carefully sum only over repeated dummy indices that appear in the given term.  
e.g., the two paths in \eqref{eq:BA-Jacob-wrt-decoder-coords-upper-left-block} which traverse the edge $\tfrac{\plog p_{i+1}(x_1|\hat{x}_2)}{\plog p_{i+1}(\hat{x}_4|x_5)}$ (pointing from $p_{i+1}(\xhat|x)$ to $p_{i+1}(x|\xhat)$) do \textit{not} involve a summation over $\hat{x}_3$. In contrast, the two paths involving $\tfrac{\plog p_{i+1}(x_1|\hat{x}_2)}{\plog p_{i+1}(\hat{x}_3)} \tfrac{\plog p_{i+1}(\hat{x}_3)}{\plog p_{i+1}(\hat{x}_4|x_5)}$ there do entail a summation over $\hat{x}_3$.
This is relevant for the calculations below, as in \eqref{eq:BA-Jacob-wrt-decoder-coords-upper-left-block-intermid-calcs} for example.

Similarly, for the upper-right block of \eqref{eq:BA-Jacob-wrt-decoder-coords-as-block-matrix-implicit},
\begin{multline}		\label{eq:BA-Jacob-wrt-decoder-coords-upper-right-block}
	\frac{d\log p_{i+1}(\y|\xhat)}{d \log p_i(\xhatp)} =
	\frac{\plog p_{i+1}(\y|\xhat)}{\plog p_{i+1}(x_1|\hat{x}_2)} \cdot \left[
	\frac{\plog p_{i+1}(x_1|\hat{x}_2)}{\plog p_{i+1}(\hat{x}_3)}
	\frac{\plog p_{i+1}(\hat{x}_3)}{\plog p_{i+1}(\hat{x}_4|x_5)} +
	\frac{\plog p_{i+1}(x_1|\hat{x}_2)}{\plog p_{i+1}(\hat{x}_4|x_5)}
	\right]	\\ \cdot \left\{
		\frac{\plog p_{i+1}(\hat{x}_4|x_5)}{\plog p_i(\xhatp)} +
		\frac{\plog p_{i+1}(\hat{x}_4|x_5)}{\plog p_i(y_7|\hat{x}_8)} \frac{\plog p_i(y_7|\hat{x}_8)}{\plog p_i(x_9|\hat{x}_{10})} \frac{\plog p_i(x_9|\hat{x}_{10})}{\plog p_i(\xhatp)} \right. \\ \left. +
		\frac{\plog p_{i+1}(\hat{x}_4|x_5)}{\plog Z_i(x_6, \beta)} \left[
			\frac{\plog Z_i(x_6, \beta)}{\plog p_i(\xhatp)} +
			\frac{\plog Z_i(x_6, \beta)}{\plog p_i(y_7|\hat{x}_8)}
			\frac{\plog p_i(y_7|\hat{x}_8)}{\plog p_i(x_9|\hat{x}_{10})} \frac{\plog p_i(x_9|\hat{x}_{10})}{\plog p_i(\xhatp)}
		\right]
	\right\}
\end{multline}
For the lower-left block of \eqref{eq:BA-Jacob-wrt-decoder-coords-as-block-matrix-implicit},
\begin{equation}		\label{eq:BA-Jacob-wrt-decoder-coords-lower-left-block}
	\frac{d\log p_{i+1}(\xhat)}{d \log p_i(\yp|\xhatp)} =
	\frac{\plog p_{i+1}(\xhat)}{\plog p_{i+1}(\hat{x}_1|x_2)} \left[
		\frac{\plog p_{i+1}(\hat{x}_1|x_2)}{\plog p_i(\yp|\xhatp)} +
		\frac{\plog p_{i+1}(\hat{x}_1|x_2)}{\plog Z_i(x_3, \beta)} \frac{\plog Z_i(x_3, \beta)}{\plog p_i(\yp|\xhatp)}
	\right]
\end{equation}
Last, for the lower-right block of \eqref{eq:BA-Jacob-wrt-decoder-coords-as-block-matrix-implicit},
\begin{multline}		\label{eq:BA-Jacob-wrt-decoder-coords-lower-right-block}
	\frac{d\log p_{i+1}(\xhat)}{d \log p_i(\xhatp)} \\ =
	\frac{\plog p_{i+1}(\xhat)}{\plog p_{i+1}(\hat{x}_1|x_2)} \cdot \left\{
		\frac{\plog p_{i+1}(\hat{x}_1|x_2)}{\plog p_i(\xhatp)} +
		\frac{\plog p_{i+1}(\hat{x}_1|x_2)}{\plog p_i(y_3|\hat{x}_4)} 
		\frac{\plog p_i(y_3|\hat{x}_4)}{\plog p_i(x_5|\hat{x}_6)} \frac{\plog p_i(x_5|\hat{x}_6)}{\plog p_i(\xhatp)} \right. \\ \left. +
		\frac{\plog p_{i+1}(\hat{x}_1|x_2)}{\plog Z_i(x_7, \beta)} \left[
			\frac{\plog Z_i(x_7, \beta)}{\plog p_i(\xhatp)} +
			\frac{\plog Z_i(x_7, \beta)}{\plog p_i(y_3|\hat{x}_4)}
			\frac{\plog p_i(y_3|\hat{x}_4)}{\plog p_i(x_5|\hat{x}_6)} 
			\frac{\plog p_i(x_5|\hat{x}_6)}{\plog p_i(\xhatp)}
		\right]
	\right\}
\end{multline}

Next, by using the intermediate results summarized in \eqref{eq:log-variable-dep-w-derivatives-no-normalization} (Section \ref{subsub:BA-IB-jacobian-appendix:differentiating-along-depend-graph}), we calculate each of the four blocks of \eqref{eq:BA-Jacob-wrt-decoder-coords-as-block-matrix-implicit} explicitly. 
For the upper-left block \eqref{eq:BA-Jacob-wrt-decoder-coords-upper-left-block} we have
\begin{multline}		\label{eq:BA-Jacob-wrt-decoder-coords-upper-left-block-explicit}
	\frac{d\log p_{i+1}(\y|\xhat)}{d \log p_i(\yp|\xhatp)} =
	\frac{p(\y|x_1) p_{i+1}(x_1|\xhat)}{p_{i+1}(\y|\xhat)} \; \delta_{\xhat, \hat{x}_2} \cdot
	\left[
		\left( - \delta_{\hat{x}_2, \hat{x}_3} \right)
		p_{i+1}(x_5|\hat{x}_3) \delta_{\hat{x}_3, \hat{x}_4} +
		\delta_{x_1, x_5} \delta_{\hat{x}_2, \hat{x}_4}
	\right]	\\ \cdot \left[
		\beta \delta_{\hat{x}_4, \xhatp} p(\yp|x_5) +
		( -\delta_{x_5, x_6} )
		\beta p_{i+1}(\xhatp|x_6) p(\yp|x_6)
	\right]
\end{multline}
For clarity, we elaborate on each step needed to complete the calculation of the upper-left block \eqref{eq:BA-Jacob-wrt-decoder-coords-upper-left-block} while providing only the main steps for the other blocks.
To carry out the summations over the dummy variables $x_1, \hat{x}_2, \hat{x}_3, \hat{x}_4, x_5$ and $x_6$ in \eqref{eq:BA-Jacob-wrt-decoder-coords-upper-left-block-explicit}, we carefully sum only over repeated dummy indices, as explained after \eqref{eq:BA-Jacob-wrt-decoder-coords-upper-left-block}. 
We carry out one summation at a time, starting with $\hat{x}_2$. This yields,
\begin{multline}		\label{eq:BA-Jacob-wrt-decoder-coords-upper-left-block-intermid-calcs}
	\beta \; \frac{p(\y|x_1) p_{i+1}(x_1|\xhat)}{p_{i+1}(\y|\xhat)} \cdot
	\left[
		- \delta_{\xhat, \hat{x}_3} p_{i+1}(x_5|\hat{x}_3) \delta_{\hat{x}_3, \hat{x}_4} +
		\delta_{x_1, x_5} \delta_{\xhat, \hat{x}_4}
	\right]	\\ \cdot \left[
		\delta_{\hat{x}_4, \xhatp} p(\yp|x_5)
		-\delta_{x_5, x_6} \; p_{i+1}(\xhatp|x_6) p(\yp|x_6)
	\right]	\\ =
	\beta \cdot \left[
		- \delta_{\xhat, \hat{x}_3} p_{i+1}(x_5|\hat{x}_3) \delta_{\hat{x}_3, \hat{x}_4} +
		\delta_{\xhat, \hat{x}_4} \frac{p(\y|x_5) p_{i+1}(x_5|\xhat)}{p_{i+1}(\y|\xhat)}
	\right]	\\ \cdot \left[
		\delta_{\hat{x}_4, \xhatp} p(\yp|x_5)
		-\delta_{x_5, x_6} \; p_{i+1}(\xhatp|x_6) p(\yp|x_6)
	\right]	\\ =
	\beta \cdot p_{i+1}(x_5|\xhat) \; \left[
		- \delta_{\xhat, \hat{x}_4} +
		\delta_{\xhat, \hat{x}_4} \frac{p(\y|x_5) }{p_{i+1}(\y|\xhat)}
	\right]	\cdot \left[
		\delta_{\hat{x}_4, \xhatp} p(\yp|x_5)
		-\delta_{x_5, x_6} \; p_{i+1}(\xhatp|x_6) p(\yp|x_6)
	\right]	\\ =
	- \beta \cdot p_{i+1}(x_5|\xhat)\; \left[
		1 -	\frac{p(\y|x_5) }{p_{i+1}(\y|\xhat)}
	\right]	\cdot \left[
		\delta_{\xhat, \xhatp} p(\yp|x_5)
		-\delta_{x_5, x_6} \; p_{i+1}(\xhatp|x_6) p(\yp|x_6)
	\right]	\\ =
	- \beta \cdot p(\yp|x_5) p_{i+1}(x_5|\xhat) \; \left[
		1 -  \frac{p(\y|x_5) }{p_{i+1}(\y|\xhat)}
	\right]	\cdot \left[
		\delta_{\xhat, \xhatp} 
		- p_{i+1}(\xhatp|x_5) 
	\right]	\\ =
	-\beta \sum_{\x} p(\yp|x) p_{i+1}(x|\xhat) \cdot \left[
		1 -
		\frac{p(\y|x) }{p_{i+1}(\y|\xhat)}
	\right]	\cdot \Big[
		\delta_{\xhat, \xhatp} 
		- p_{i+1}(\xhatp|x) 
	\Big]	
\end{multline}
In the first equality above we carried out the summation over $x_1$, in the second over $\hat{x}_3$, in the third over $\hat{x}_4$, in the fourth over $x_6$, and in the fifth over $x_5$.

To simplify the notation, we replace summations over $x$ with definitions as in Equation \eqref{eq:B_defs-for-BA-Jacob-wrt-decoder-coords} (Section \ref{sec:IB-ODE}),
\begin{align}			\label{eq:B_C_defs-for-BA-Jacob-wrt-decoder-coords-appendix}
\begin{split}
	C(\xhat, \xhatp; i)_{\y, \yp} := 	&\sum_{\x} p(\y|x) p(\yp|x) p_{i}(\xhatp|x) p_{i}(x|\xhat)	\\
	B(\xhat, \xhatp; i)_{\y} := 		&\sum_{\x} p(\y|x) p_{i}(\xhatp|x) p_{i}(x|\xhat) 
	= \sum_{\yp} C(\xhat, \xhatp; i)_{\y, \yp} \\
	A(\xhat, \xhatp; i) := 				&\sum_{\x} p_{i}(\xhatp|x) p_{i}(x|\xhat) = 
	\sum_{\y} B(\xhat, \xhatp; i)_{\y}	\\
	D(\xhat; i)_{\y, \yp} := 			&\frac{1}{p_i(\y|\xhat) } \sum_{\x} p(\y|x) p(\yp|x) p_{i}(x|\xhat) =
	\frac{1}{p_i(\y|\xhat) } \sum_{\xhatp} C(\xhat, \xhatp; i)_{\y, \yp}
\end{split}
\end{align}
and note that
\begin{equation}			\label{eq:B-tensor-sums-up-to-decoder}
	\sum_{\yp, \xhatp} C(\xhat, \xhatp; i)_{\y, \yp} = p_i(\y|\xhat) \;. 
\end{equation}
The quantities $A, B$, and $C$ involve two IB clusters. They are a scalar, a vector, and a matrix, respectively. 
The definition of $D$ involves only one IB cluster and coincides with $C_Y$ in \cite[3.2 in Part III]{zaslavsky2019thesis}. 
The relations to the right of \eqref{eq:B_C_defs-for-BA-Jacob-wrt-decoder-coords-appendix} show that each can be expressed in terms of $C(\xhat, \xhatp; i)_{\y, \yp}$. 
Equation \eqref{eq:B-tensor-sums-up-to-decoder} shows that the latter can be rewritten as a right-stochastic matrix, up to trivial manipulations. 
As seen below, the Jacobian matrix \eqref{eq:BA-Jacob-wrt-decoder-coords-as-block-matrix-implicit} of a BA-IB step in log-decoder coordinates can be computed in terms of the quantities in \eqref{eq:B_C_defs-for-BA-Jacob-wrt-decoder-coords-appendix}.

With the latter definitions \eqref{eq:B_C_defs-for-BA-Jacob-wrt-decoder-coords-appendix}, \eqref{eq:BA-Jacob-wrt-decoder-coords-upper-left-block-intermid-calcs} can be rewritten as,
\begin{multline}				\label{eq:BA-Jacob-wrt-decoder-coords-upper-left-block-simplified}
	\frac{d\log p_{i+1}(\y|\xhat)}{d \log p_i(\yp|\xhatp)} \overset{\eqref{eq:BA-Jacob-wrt-decoder-coords-upper-left-block-intermid-calcs}}{=}
	-\beta \sum_{\x} \Big[
		\delta_{\xhat, \xhatp} \; p(\yp|x) p_{i+1}(x|\xhat) 
		- p(\yp|x) p_{i+1}(\xhatp|x) p_{i+1}(x|\xhat) 
	\Big. \\ \Big.
		- \delta_{\xhat, \xhatp} \; \tfrac{1}{p_{i+1}(\y|\xhat)} 
		p(\y|x) p(\yp|x) p_{i+1}(x|\xhat) +
		\tfrac{1}{p_{i+1}(\y|\xhat)} 
		p(\yp|x) p(\y|x) p_{i+1}(\xhatp|x) p_{i+1}(x|\xhat) 
	\Big] \\ \overset{\eqref{eq:B_C_defs-for-BA-Jacob-wrt-decoder-coords-appendix}}{=}
	-\beta \Big[
		\delta_{\xhat, \xhatp} \; p_{i+1}(\yp|\xhat) 
		- B(\xhat, \xhatp; i+1)_{\yp}
	\Big. \\ \Big. 
		- \delta_{\xhat, \xhatp} \; D(\xhat; i+1)_{\y, \yp}
		+ \tfrac{1}{p_{i+1}(\y|\xhat)} C(\xhat, \xhatp; i+1)_{\y, \yp}
	\Big] \\ 
	=
	\beta \sum_{\xhatpp, \ypp} 
		\left( \delta_{\xhatpp, \xhatp} - \delta_{\xhat, \xhatp} \right)
		\left( 1 - \tfrac{\delta_{\ypp, \y}}{p_{i+1}(\y|\xhat) } \right)
	C(\xhat, \xhatpp; i+1)_{\yp, \ypp} 
\end{multline}
The third equality above follows from \eqref{eq:B-tensor-sums-up-to-decoder}, the identities to the right of \eqref{eq:B_C_defs-for-BA-Jacob-wrt-decoder-coords-appendix} and simple algebra.

For the upper-right block \eqref{eq:BA-Jacob-wrt-decoder-coords-upper-right-block},
\begin{multline}		\label{eq:BA-Jacob-wrt-decoder-coords-upper-right-block-explicit}
	\frac{d\log p_{i+1}(\y|\xhat)}{d \log p_i(\xhatp)} =
	\frac{p(\y|x_1) p_{i+1}(x_1|\hat{x}_2) }{p_{i+1}(\y|\hat{x}_2)} \delta_{\hat{x}_2, \xhat}
	\cdot \Big[
		\left( - \delta_{\hat{x}_2, \hat{x}_3} \right)
		p_{i+1}(x_5|\hat{x}_3) \delta_{\hat{x}_3, \hat{x}_4} +
		\delta_{x_1, x_5} \delta_{\hat{x}_2, \hat{x}_4}
	\Big] \\ \cdot \left\{
		\delta_{\hat{x}_4, \xhatp} +
		\beta \delta_{\hat{x}_4, \hat{x}_8} p(y_7|x_5)
		\frac{p(y_7|x_9) p_i(x_9|\hat{x}_8) }{p_i(y_7|\hat{x}_8)} \delta_{\hat{x}_8, \hat{x}_{10}}
		\left( -\delta_{\hat{x}_{10}, \xhatp} \right)
		\right. \\ \left. +
		( -\delta_{x_5, x_6} ) \left[
			p_{i+1}(\xhatp|x_6) +
			\beta p_{i+1}(\hat{x}_8|x_6) p(y_7|x_6)
			\frac{p(y_7|x_9) p_i(x_9|\hat{x}_8) }{p_i(y_7|\hat{x}_8)} \delta_{\hat{x}_8, \hat{x}_{10}}
			\left( -\delta_{\hat{x}_{10}, \xhatp} \right)
		\right]
	\right\}
\end{multline}
In a manner similar to \eqref{eq:BA-Jacob-wrt-decoder-coords-upper-left-block-intermid-calcs}, summing over all ten dummy variables other than $x_1$ and $x_5$ yields,
\begin{multline}		\label{eq:BA-Jacob-wrt-decoder-coords-upper-right-block-intermid}
	\left( 1 - \beta \right) \cdot
	\frac{p(\y|x_1) p_{i+1}(x_1|\xhat) }{p_{i+1}(\y|\xhat)} \cdot 
	\Big( \delta_{x_1, x_5} - p_{i+1}(x_5|\xhat) \Big) \cdot 
	\Big( \delta_{\xhat, \xhatp} - p_{i+1}(\xhatp|x_5) \Big) \\ =
		\left( 1 - \beta \right) \cdot \Big(
		\tfrac{-1}{p_{i+1}(\y|\xhat)} \sum_{\x} p(\y|x) p_{i+1}(\xhatp|x) p_{i+1}(x|\xhat) 
		+ \sum_{\x} p_{i+1}(\xhatp|x) p_{i+1}(x|\xhat)
	\Big) \\ =
	\left( 1 - \beta \right) \cdot 
	\sum_{\x} \left( 1 - \tfrac{p(\y|x) }{p_{i+1}(\y|\xhat)} \right) p_{i+1}(\xhatp|x) p_{i+1}(x|\xhat)
\end{multline}
The two terms involving $\delta_{\xhat, \xhatp}$ cancel out when summing over $x_1$ and $x_5$ at the first equality. Rewriting with the definitions \eqref{eq:B_C_defs-for-BA-Jacob-wrt-decoder-coords-appendix} of $A$ and $B$ further simplifies \eqref{eq:BA-Jacob-wrt-decoder-coords-upper-right-block-intermid} to,
\begin{multline}			\label{eq:BA-Jacob-wrt-decoder-coords-upper-right-block-simplified}
	\left( 1 - \beta \right) \cdot \Big[ 
		A(\xhat, \xhatp; i+1) - 
		\tfrac{1}{p_{i+1}(\y|\xhat)} B(\xhat, \xhatp; i+1)_{\y}
	\Big] \\ =
	\left( 1 - \beta \right) \cdot \sum_{\ypp} \Big[ 
		1 - \tfrac{\delta_{\ypp, \y}}{p_{i+1}(\y|\xhat)} 
	\Big] B(\xhat, \xhatp; i+1)_{\ypp} 
\end{multline}

For the lower-left block \eqref{eq:BA-Jacob-wrt-decoder-coords-lower-left-block},
\begin{equation}		\label{eq:BA-Jacob-wrt-decoder-coords-lower-left-block-explicit}
	\frac{d\log p_{i+1}(\xhat)}{d \log p_i(\yp|\xhatp)} =
	p_{i+1}(x_2|\xhat) \delta_{\xhat, \hat{x}_1} \left[
		\beta \delta_{\hat{x}_1, \xhatp} p(\yp|x_2) +
		\left( -\delta_{x_2, x_3} \right)
		\beta p_{i+1}(\xhatp|x_3) p(\yp|x_3)
	\right]
\end{equation}
Summing over dummy variables and simplifying yields,
\begin{equation}		\label{eq:BA-Jacob-wrt-decoder-coords-lower-left-block-intermid}
	\beta \cdot \Big[
		\delta_{\xhat, \xhatp} \; p_{i+1}(\yp|\xhat) -
		\sum_{\x} p(\yp|x) p_{i+1}(\xhatp|x) p_{i+1}(x|\xhat) 
	\Big]
\end{equation}
In terms of definitions \eqref{eq:B_C_defs-for-BA-Jacob-wrt-decoder-coords-appendix}, this simplifies to
\begin{equation}		\label{eq:BA-Jacob-wrt-decoder-coords-lower-left-block-simplified}
	\beta \cdot \Big[
		\delta_{\xhat, \xhatp} \; p_{i+1}(\yp|\xhat) -
		B(\xhat, \xhatp; i+1)_{\yp}
	\Big]
\end{equation}

Finally, for the lower-right block \eqref{eq:BA-Jacob-wrt-decoder-coords-lower-right-block},
\begin{multline}		\label{eq:BA-Jacob-wrt-decoder-coords-lower-right-block-explicit}
	\frac{d\log p_{i+1}(\xhat)}{d \log p_i(\xhatp)} \\ =
	p_{i+1}(x_2|\xhat) \delta_{\xhat, \hat{x}_1} \cdot \left\{
		\delta_{\hat{x}_1, \xhatp} +
		\beta \delta_{\hat{x}_1, \hat{x}_4} p(y_3|x_2)
		\frac{p(y_3|x_5) p_i(x_5|\hat{x}_4) }{p_i(y_3|\hat{x}_4)} \delta_{\hat{x}_4, \hat{x}_6}
		\left( -\delta_{\hat{x}_6, \xhatp} \right)
		\right. \\ \left. +
		\left( -\delta_{x_2, x_7} \right) \left[
			p_{i+1}(\xhatp|x_7) +
			\beta p_{i+1}(\hat{x}_4|x_7) p(y_3|x_7)
			\frac{p(y_3|x_5) p_i(x_5|\hat{x}_4) }{p_i(y_3|\hat{x}_4)} \delta_{\hat{x}_4, \hat{x}_6}
			\left( -\delta_{\hat{x}_6, \xhatp} \right)
		\right]
	\right\}
\end{multline}
This simplifies to,
\begin{equation}		\label{eq:BA-Jacob-wrt-decoder-coords-lower-right-block-intermid}
	\left(1 - \beta \right)
	\left(\delta_{\xhat, \xhatp} - \sum_{\x} p_{i+1}(\xhatp|x) p_{i+1}(x|\xhat) \right)
\end{equation}
With definitions \eqref{eq:B_C_defs-for-BA-Jacob-wrt-decoder-coords-appendix}, this can be written as
\begin{equation}		\label{eq:BA-Jacob-wrt-decoder-coords-lower-right-block-simplified}
	\left(1 - \beta \right)
	\Big( \delta_{\xhat, \xhatp} - A(\xhat, \xhatp; i+1) \Big)
\end{equation}

\medskip
Collecting the results from \eqref{eq:BA-Jacob-wrt-decoder-coords-upper-left-block-simplified}, \eqref{eq:BA-Jacob-wrt-decoder-coords-upper-right-block-simplified}, \eqref{eq:BA-Jacob-wrt-decoder-coords-lower-left-block-simplified} and \eqref{eq:BA-Jacob-wrt-decoder-coords-lower-right-block-simplified} back into \eqref{eq:BA-Jacob-wrt-decoder-coords-as-block-matrix-implicit}, BA's Jacobian in these coordinates is
\begin{equation}		\label{eq:BA-Jacob-wrt-decoder-coords-as-block-matrix-explicit}
	\left(
	\begin{array}{c|c}
		\substack{
			\beta \sum_{\xhatpp, \ypp} 
			\left( \delta_{\xhatpp, \xhatp} - \delta_{\xhat, \xhatp} \right)
		\Big. \\ \Big.
			\cdot \left(1 - \tfrac{\delta_{\ypp, \y} }{p_{i+1}(\y|\xhat)} \right)
			C(\xhat, \xhatpp; i+1)_{\yp, \ypp} 
		} &
		\left( 1 - \beta \right) \cdot \sum_{\ypp} \Big[ 
		1 - \tfrac{\delta_{\ypp, \y}}{p_{i+1}(\y|\xhat)} 
		\Big] B(\xhat, \xhatp; i+1)_{\ypp}	\\	\\
		\hline	\\
		\beta \cdot \Big[
		\delta_{\xhat, \xhatp} \; p_{i+1}(\yp|\xhat) -
		B(\xhat, \xhatp; i+1)_{\yp}
		\Big] &
		\left(1 - \beta \right)
		\Big( \delta_{\xhat, \xhatp} - A(\xhat, \xhatp; i+1) \Big)
	\end{array}
	\right)
\end{equation}
When evaluated at an IB root, this is Equation \eqref{eq:BA-Jacob-wrt-decoder-coords-at-main-text} of Section \ref{sec:IB-ODE}. 
Equivalently, it can be written in the following form, which is more convenient for implementation
\begin{equation}		\label{eq:BA-Jacob-wrt-decoder-coords-as-block-matrix-explicit-implementation-friendly}
	\left(
	\begin{array}{c|c}
		\substack{
			\beta \Big[
				B(\xhat, \xhatp; i+1)_{\yp}
				- \delta_{\xhat, \xhatp} \; p_{i+1}(\yp|\xhat) 
			\Big. \\ \Big. 
				+ \delta_{\xhat, \xhatp} \; D(\xhat; i+1)_{\y, \yp}
				- \tfrac{1}{p_{i+1}(\y|\xhat)} C(\xhat, \xhatp; i+1)_{\y, \yp}
			\Big] 
		} &
		\left( 1 - \beta \right) \cdot \Big[ 
		A(\xhat, \xhatp; i+1) - 
		\tfrac{1}{p_{i+1}(\y|\xhat)} B(\xhat, \xhatp; i+1)_{\y}
		\Big]	\\	\\
		\hline	\\
		\beta \cdot \Big[
		\delta_{\xhat, \xhatp} \; p_{i+1}(\yp|\xhat) -
		B(\xhat, \xhatp; i+1)_{\yp}
		\Big] &
		\left(1 - \beta \right)
		\Big( \delta_{\xhat, \xhatp} - A(\xhat, \xhatp; i+1) \Big)
	\end{array}
	\right)
\end{equation}

\medskip 
\subsection{The partial $\beta$-derivatives of BA-IB in log-decoder coordinates}

\label{sub:BA-IB-beta-deriv-appendix-wrt-decoder-coords}

We calculate the vector $D_\beta BA_\beta$ of partial derivatives of the $BA_\beta$ operator in log-decoder coordinates (of Section \ref{sec:coords-exchange-for-the-IB}), which appears at the right-hand side of the IB-ODE \eqref{eq:IB-beta-ODE-in-decoder-coords} (in Section \ref{sec:IB-ODE}). 

\medskip 
To that end, we differentiate backward along the dependencies graph \eqref{eq:dependencies-graph-for-BA-IB-variables} (in Appendix \ref{sub:BA-IB-jacobian-appendix:calculation-setup-and-goals}) with respect to $\beta$, starting at the output coordinates $p_{i+1}(\y|\xhat)$ and $p_{i+1}(\xhat)$ of $BA_\beta$. 
After differentiating, we mind our independent variables. 
Here, these are $\beta$, and the input coordinates $p_{i}(\y|\xhat)$ and $p_{i}(\xhat)$ of $BA_\beta$. 
The differentiation of these with respect to $\beta$ vanishes (except for $\tfrac{d\beta}{d\beta} = 1$), as they are independent. 
Finally, we compose the differentiations to obtain the effect $D_\beta BA_\beta$ of changing $\beta$ on BA's output. 
We note that, in principle, one can differentiate the explicit formulae \eqref{eq:BA-operator-def-in-decoder-coords-appendix} of $BA_\beta$ in decoder coordinates (Appendix \ref{sec:BA-IB-in-decoder-coords-appendix}) with respect to $\beta$. 
However, we find that to be cumbersome and far more error-prone than our approach, and so proceed in the spirit of the previous Appendix \ref{sub:BA-IB-jacobian-appendix:decoder-deriv-matrix}. 

We start by differentiating each of the equations defining the Blahut-Arimoto Algorithm \ref{algo:BA-IB} with respect to $\beta$, as if all its variables are dependent.
For the cluster marginal Equation \algref{algo:BA-IB}{eq:IB-BA-cluster_marginal},
\begin{equation}			\label{eq:partial-beta-deriv-calcs-marginal-eq}
	\dbeta{} p_{i}(\xhat) = \sum_{\x} p(\x) \dbeta{} p_i(\xhat|x)
\end{equation}
For the inverse-encoder Equation \algref{algo:BA-IB}{eq:IB-BA-bayes-for-computing-inverse-enc},
\begin{equation}			\label{eq:partial-beta-deriv-calcs-inv-enc-eq}
	\dbeta{} p_i(\x|\xhat) = 
	\frac{p(\x)}{p_i(\xhat)} \dbeta{p_i(\xhat|\x)} -
	\frac{p_i(\xhat|\x)p(\x)}{p_i(\xhat)^2} \dbeta{p_i(\xhat)} 
\end{equation}
For the decoder Equation \algref{algo:BA-IB}{eq:IB-BA-decoder-eq}
\begin{equation}			\label{eq:partial-beta-deriv-calcs-decoder-eq}
	\dbeta{} p_{i}(\y|\xhat) =
	\sum_{\x} p(\y|x) \dbeta{} p_{i}(x|\xhat)
\end{equation}
For the KL-divergence,
\begin{equation}			\label{eq:partial-beta-deriv-calcs-KL-divergence}
	\dbeta{} D_{KL}\big[ p(\y|\x) || p_{i}(\y|\xhat) \big] =
	\dbeta{} \sum_{\ypp} p(\ypp|\x) \log \frac{p(\ypp|\x) }{p_{i}(\ypp|\xhat) } =
	- \sum_{\ypp} \frac{p(\ypp|\x) }{p_{i}(\ypp|\xhat) } \dbeta{} p_{i}(\ypp|\xhat) 
\end{equation}
And its exponent,
\begin{multline}			\label{eq:partial-beta-deriv-calcs-exp-beta-divergence}
	\dbeta{} \exp{\big\{ -\beta \; D_{\x, \xhat} \big\}} =
	- \left( D_{\x, \xhat} + \beta \dbeta{D_{\x, \xhat}} \right) \cdot \exp{\big\{ -\beta \; D_{\x, \xhat} \big\}} 
	\\ \overset{\eqref{eq:partial-beta-deriv-calcs-KL-divergence}}{=}
	- \left( D_{\x, \xhat} - \beta \sum_{\ypp} \frac{p(\ypp|\x) }{p_{i}(\ypp|\xhat) } \dbeta{} p_{i}(\ypp|\xhat)  \right) \cdot \exp{\big\{ -\beta \; D_{\x, \xhat} \big\}} 
\end{multline}
Where we have written $D_{\x, \xhat} := D_{KL}\big[ p(\y|\x) || p_{i}(\y|\xhat) \big]$ for short.
Thus, for the partition function's Equation \algref{algo:BA-IB}{eq:IB-BA-partition-func} we have,
\begin{multline}			\label{eq:partial-beta-deriv-calcs-partition-func-eq}
	\dbeta{} Z_{i}(\x, \beta) =
	\dbeta{} \sum_{\xhatpp} p_{i}(\xhatpp) \exp{\big\{ -\beta \; D_{\x, \xhatpp} \big\}}
	\\ \overset{\eqref{eq:partial-beta-deriv-calcs-exp-beta-divergence}}{=}
	\sum_{\xhatpp} \left( \dbeta{p_{i}(\xhatpp) } - p_{i}(\xhatpp) D_{\x, \xhatpp} + \beta \; p_{i}(\xhatpp) \sum_{\ypp} \frac{p(\ypp|\x) }{p_{i}(\ypp|\xhatpp) } \dbeta{} p_{i}(\ypp|\xhatpp) \right) 
	\cdot \exp{\big\{ -\beta \; D_{\x, \xhatpp} \big\}} 
\end{multline}
Finally, for the encoder Equation \algref{algo:BA-IB}{eq:IB-BA-new-direct-enc} we have
\begin{multline}			\label{eq:partial-beta-deriv-calcs-encoder-eq}
	\dbeta{} p_{i+1}(\xhat|\x) =
	\dbeta{} \left( \frac{p_{i}(\xhat) e^{-\beta\;D_{\x, \xhat} }}{Z_{i}(\x,\beta)} \right)
	\\ \overset{\eqref{eq:partial-beta-deriv-calcs-exp-beta-divergence}}{=}
	\frac{p_{i}(\xhat) e^{-\beta\;D_{\x, \xhat} }}{Z_{i}(\x,\beta)} \left[ \frac{1}{p_{i}(\xhat) } \dbeta{p_{i}(\xhat)}
	- \left( D_{\x, \xhat} - \beta \sum_{\ypp} \frac{p(\ypp|\x) }{p_{i}(\ypp|\xhat) } \dbeta{} p_{i}(\ypp|\xhat)  \right) - \frac{1}{Z_{i}(\x,\beta)} \dbeta{Z_{i}(\x,\beta)} \right]
	\\ \overset{\algref{algo:BA-IB}{eq:IB-BA-new-direct-enc} }{=}
	p_{i+1}(\xhat|\x) \cdot \left[ \frac{1}{p_{i}(\xhat) } \dbeta{p_{i}(\xhat)}
	- \left( D_{\x, \xhat} - \beta \sum_{\ypp} \frac{p(\ypp|\x) }{p_{i}(\ypp|\xhat) } \dbeta{} p_{i}(\ypp|\xhat)  \right) - \frac{1}{Z_{i}(\x,\beta)} \dbeta{Z_{i}(\x,\beta)} \right]
\end{multline}

\medskip
Next, picking $\beta$ and the inputs $\log p_i(\y|\xhat)$ and $\log p_i(\xhat)$ of $BA_\beta$ as our independent variables, we compose the differentiations above to obtain $D_\beta BA_\beta$ at an output coordinate.
That is, we seek $\dbeta{} \log p_{i+1}(\y|\xhat)$ and $\dbeta{} \log p_{i+1}(\xhat)$.
By the chain rule, we trace the dependencies graph \eqref{eq:dependencies-graph-for-BA-IB-variables} (Section \ref{sub:BA-IB-jacobian-appendix:calculation-setup-and-goals}) backwards, from the output nodes $p_{i+1}(\y|\xhat)$ and $p_{i+1}(\xhat)$ back to the input nodes. 
The derivatives of the latter with respect to $\beta$ vanish, as these are our independent variables.

Starting with a decoder output coordinate,
\begin{multline}			\label{eq:beta-deriv-of-log-output-decoder-composition-step}
	\dbeta{} \log p_{i+1}(\y|\xhat) =
	\frac{1}{p_{i+1}(\y|\xhat) } \dbeta{} p_{i+1}(\y|\xhat) \overset{\eqref{eq:partial-beta-deriv-calcs-decoder-eq}}{=}
	\frac{1}{p_{i+1}(\y|\xhat) } \sum_{\x} p(\y|x) \dbeta{} p_{i+1}(x|\xhat) 
	\\ \overset{\eqref{eq:partial-beta-deriv-calcs-inv-enc-eq}}{=}
	\frac{1}{p_{i+1}(\y|\xhat) } \sum_{\x} p(\y|\x) \left[ \frac{p(\x)}{p_{i+1}(\xhat)} \dbeta{p_{i+1}(\xhat|\x)} -
	\frac{p_{i+1}(\xhat|\x)p(\x)}{p_{i+1}(\xhat)^2} \dbeta{p_{i+1}(\xhat)} \right] 
	\\ \overset{\eqref{eq:partial-beta-deriv-calcs-marginal-eq}}{=}
	\sum_{\x} \frac{p(\y|\x) p(\x)}{p_{i+1}(\y|\xhat) p_{i+1}(\xhat)} \left[ \dbeta{p_{i+1}(\xhat|\x)} -
	\frac{p_{i+1}(\xhat|\x) }{p_{i+1}(\xhat)} \sum_{\xp} p(\xp) \dbeta{p_{i+1}(\xhat|\xp) } \right] 
	\\ \overset{\eqref{eq:partial-beta-deriv-calcs-encoder-eq}}{=}
	\sum_{\x} \frac{p(\y|\x) p(\x)}{p_{i+1}(\y|\xhat) p_{i+1}(\xhat)} \left\{ 
		p_{i+1}(\xhat|\x) \cdot \left[ \frac{1}{p_{i}(\xhat) } \dbeta{p_{i}(\xhat)}
		- \left( D_{\x, \xhat} - \beta \sum_{\ypp} \frac{p(\ypp|\x) }{p_{i}(\ypp|\xhat) } \dbeta{p_{i}(\ypp|\xhat) } \right) - \frac{1}{Z_{i}(\x,\beta)} \dbeta{Z_{i}(\x,\beta)} \right] \right. \\ \left.
		- \frac{p_{i+1}(\xhat|\x) }{p_{i+1}(\xhat)} \sum_{\xp} p(\xp) \left(
		p_{i+1}(\xhat|\xp) \cdot \left[ \frac{1}{p_{i}(\xhat) } \dbeta{p_{i}(\xhat)}
		- \left( D_{\xp, \xhat} - \beta \sum_{\ypp} \frac{p(\ypp|\xp) }{p_{i}(\ypp|\xhat) } \dbeta{p_{i}(\ypp|\xhat) } \right) - \frac{1}{Z_{i}(\xp,\beta)} \dbeta{Z_{i}(\xp,\beta)} \right]\right)
	\right\}
\end{multline}
Since $p_{i}(\y|\xhat)$ and $p_{i}(\xhat)$ are independent input variables, their derivatives with respect to the independent variable $\beta$ vanish, yielding 
\begin{multline}			\label{eq:beta-deriv-of-log-output-decoder-composition-step-B}
	- \sum_{\x} \frac{p(\y|\x) p(\x)}{p_{i+1}(\y|\xhat) p_{i+1}(\xhat)} \left\{ 
		p_{i+1}(\xhat|\x) \cdot \left[ D_{\x, \xhat} + \frac{1}{Z_{i}(\x,\beta)} \dbeta{Z_{i}(\x,\beta)} \right] 
		\right. \\ \left.
		- \frac{p_{i+1}(\xhat|\x) }{p_{i+1}(\xhat)} \sum_{\xp} p(\xp) 
		p_{i+1}(\xhat|\xp) \cdot \left[ D_{\xp, \xhat} + \frac{1}{Z_{i}(\xp,\beta)} \dbeta{Z_{i}(\xp,\beta)} \right]
	\right\}
\end{multline}
To complete the calculation at \eqref{eq:beta-deriv-of-log-output-decoder-composition-step}, note that the same argument can be used for two of the three summands in \eqref{eq:partial-beta-deriv-calcs-partition-func-eq}, reducing it to
\begin{equation}			\label{eq:partial-beta-deriv-calcs-partition-func-reduced}
	\dbeta{Z_{i}(\x, \beta)} =
	- \sum_{\xhatpp} p_{i}(\xhatpp) D_{\x, \xhatpp} e^{ -\beta \; D_{\x, \xhatpp} } 
\end{equation}
since $p_{i}(\y|\xhat)$ and $p_{i}(\xhat)$ are considered as independent variables. Therefore,
\begin{multline}
	\dbeta{} \log p_{i+1}(\y|\xhat) \underset{\eqref{eq:partial-beta-deriv-calcs-partition-func-reduced}}{\overset{\eqref{eq:beta-deriv-of-log-output-decoder-composition-step-B}}{=}}
	- \sum_{\x} \frac{p(\y|\x) p(\x)}{p_{i+1}(\y|\xhat) p_{i+1}(\xhat)} \left\{ 
		p_{i+1}(\xhat|\x) \cdot \left[ D_{\x, \xhat} - \sum_{\xhatpp} \left(\frac{p_{i}(\xhatpp) }{Z_{i}(\x,\beta)} e^{ -\beta \; D_{\x, \xhatpp} } \right) D_{\x, \xhatpp} \right] 
	\right. \\ \left.
		- \frac{p_{i+1}(\xhat|\x) }{p_{i+1}(\xhat)} \sum_{\xp} p(\xp) 
		p_{i+1}(\xhat|\xp) \cdot \left[ D_{\xp, \xhat} - \sum_{\xhatpp} \left(\frac{p_{i}(\xhatpp) }{Z_{i}(\xp,\beta)} e^{ -\beta \; D_{\xp, \xhatpp} } \right) D_{\xp, \xhatpp} \right]
	\right\} \\ \overset{\algref{algo:BA-IB}{eq:IB-BA-new-direct-enc}}{=}
	- \sum_{\x} \frac{p(\y|\x) p(\x)}{p_{i+1}(\y|\xhat) p_{i+1}(\xhat)} \left\{ 
		p_{i+1}(\xhat|\x) \cdot \left[ D_{\x, \xhat} - \sum_{\xhatpp} p_{i+1}(\xhatpp|\x) D_{\x, \xhatpp} \right] 
	\right. \\ \left.
		- \frac{p_{i+1}(\xhat|\x) }{p_{i+1}(\xhat)} \sum_{\xp} p(\xp) 
		p_{i+1}(\xhat|\xp) \cdot \left[ D_{\xp, \xhat} - \sum_{\xhatpp} p_{i+1}(\xhatpp|\xp) D_{\xp, \xhatpp} \right]
	\right\} 
	\\ \underset{\algref{algo:BA-IB}{eq:IB-BA-decoder-eq}}{\overset{\algref{algo:BA-IB}{eq:IB-BA-bayes-for-computing-inverse-enc}}{=}}
	\sum_{\x} p_{i+1}(\x|\xhat) D_{\x, \xhat} 
	- \sum_{\x} \frac{p(\y|\x)}{p_{i+1}(\y|\xhat) } p_{i+1}(\x|\xhat) D_{\x, \xhat} \\
	+ \sum_{\x, \xhatpp} \frac{p(\y|\x)}{p_{i+1}(\y|\xhat) } p_{i+1}(\xhatpp|\x) p_{i+1}(\x|\xhat) D_{\x, \xhatpp} 
	- \sum_{\x, \xhatpp} p_{i+1}(\xhatpp|\x) p_{i+1}(\x|\xhat) D_{\x, \xhatpp} \\ =
	\sum_{\x} \left[ 1 - \frac{p(\y|\x)}{p_{i+1}(\y|\xhat) } \right] p_{i+1}(\x|\xhat) D_{\x, \xhat} 
	- \sum_{\x, \xhatpp} \left[ 1 - \frac{p(\y|\x)}{p_{i+1}(\y|\xhat) } \right] p_{i+1}(\xhatpp|\x) p_{i+1}(\x|\xhat) D_{\x, \xhatpp} 
\end{multline}
At the second equality to the bottom we started with the third summand, then with the first, and only then with the third and fourth summands.
And so,
\begin{equation}			\label{eq:partial-beta-deriv-calcs-beta-deriv-of-output-decoder}
	\dbeta{} \log p_{i+1}(\y|\xhat) =
	\sum_{\x, \xhatpp} \left[ 1 - \frac{p(\y|\x)}{p_{i+1}(\y|\xhat) } \right] \cdot \Big[ \delta_{\xhat, \xhatpp} - p_{i+1}(\xhatpp|\x) \Big] \cdot p_{i+1}(\x|\xhat) D_{\x, \xhatpp} 
\end{equation}

Next, consider a cluster marginal output coordinate,
\begin{multline}			\label{eq:partial-beta-deriv-calcs-beta-deriv-of-output-margianl}
	\dbeta{} \log p_{i+1}(\xhat) = 
	\frac{1}{p_{i+1}(\xhat)} \dbeta{} p_{i+1}(\xhat) \overset{\eqref{eq:partial-beta-deriv-calcs-marginal-eq}}{=}
	\frac{1}{p_{i+1}(\xhat)} \sum_{\x} p(\x) \dbeta{} p_{i+1}(\xhat|\x)
	\\ \overset{\eqref{eq:partial-beta-deriv-calcs-encoder-eq}}{=}
	\frac{1}{p_{i+1}(\xhat)} \sum_{\x} p(\x) p_{i+1}(\xhat|\x) \cdot \left[ \frac{1}{p_{i}(\xhat) } \dbeta{p_{i}(\xhat)}
	- \left( D_{\x, \xhat} - \beta \sum_{\ypp} \frac{p(\ypp|\x) }{p_{i}(\ypp|\xhat) } \dbeta{} p_{i}(\ypp|\xhat)  \right) - \frac{1}{Z_{i}(\x,\beta)} \dbeta{Z_{i}(\x,\beta)} \right]
\end{multline}
Since $p_{i}(\y|\xhat)$ and $p_{i}(\xhat)$ are independent variables, their derivative with respect to $\beta$ vanish, yielding 
\begin{multline}
	- \frac{1}{p_{i+1}(\xhat)} \sum_{\x} p_{i+1}(\xhat|\x) p(\x) \left[ 
	D_{\x, \xhat} + \frac{1}{Z_{i}(\x,\beta)} \dbeta{Z_{i}(\x,\beta)} \right]
	\\ \overset{\eqref{eq:partial-beta-deriv-calcs-partition-func-reduced}}{=}
	- \frac{1}{p_{i+1}(\xhat)} \sum_{\x} p_{i+1}(\xhat|\x) p(\x) \left[ D_{\x, \xhat} 
	- \sum_{\xhatpp} \left( \frac{p_{i}(\xhatpp) }{Z_{i}(\x,\beta)} e^{ -\beta \; D_{\x, \xhatpp} } \right) D_{\x, \xhatpp} \right]
	\\ \underset{\algref{algo:BA-IB}{eq:IB-BA-bayes-for-computing-inverse-enc}}{\overset{\algref{algo:BA-IB}{eq:IB-BA-new-direct-enc}}{=}}
	- \sum_{\x, \xhatpp} \Big[ \delta_{\xhat, \xhatpp} - p_{i+1}(\xhatpp|\x) \Big] \cdot p_{i+1}(\x|\xhat) D_{\x, \xhatpp} 
\end{multline}
Thus, for the marginals' coordinates, we have obtained 
\begin{equation}			\label{eq:partial-beta-deriv-calcs-beta-deriv-of-output-marginal}
	\dbeta{} \log p_{i+1}(\xhat) = 
	- \sum_{\x, \xhatpp} \Big[ \delta_{\xhat, \xhatpp} - p_{i+1}(\xhatpp|\x) \Big] \cdot p_{i+1}(\x|\xhat) D_{\x, \xhatpp} 
\end{equation}

When evaluated at an IB root, Equations \eqref{eq:partial-beta-deriv-calcs-beta-deriv-of-output-decoder} and \eqref{eq:partial-beta-deriv-calcs-beta-deriv-of-output-marginal} form respectively the decoder and marginal coordinates of $D_\beta BA_\beta$, which appears at the right-hand side of the IB ODE \eqref{eq:IB-beta-ODE-in-decoder-coords} (note the extra minus sign in the implicit ODE \eqref{eq:implicit-beta-ODE}).

\medskip
\subsection{The coordinates exchange Jacobians between log-decoder and log-encoder coordinates}

\label{sub:coordinates-exchange-jacobians-appendix}

Following the discussion in Section \ref{sec:coords-exchange-for-the-IB} on the pros and cons of each coordinate system, we leverage the observations of Appendix \ref{sub:BA-IB-jacobian-appendix:calculation-setup-and-goals} in order to derive the coordinate exchange Jacobians, between the log-decoder and log-encoder coordinate systems. 
Exchanging between the other coordinate system pairs adds little to the below and thus is omitted. 

\medskip 
Given the encoder's logarithmic derivative $\dbeta{} \log p_\beta(\xhatp|\xp)$, we would like to compute from it the logarithmic derivative $\big(\dbeta{} \log p_\beta(\y|\xhat), \dbeta{} \log p_\beta(\xhat)\big)$ in decoder coordinates, and vice versa. 
To that end, recall that an (arbitrary) encoder $p(\xhatp|\xhat)$ determines a decoder-marginal pair $\big( p(\y|\xhat), p(\xhat) \big)$ and vice versa (e.g., Equation \eqref{eq:coordinate-sets-parameterizing-an-IB-root} in Section \ref{sec:coords-exchange-for-the-IB}). 
So, one can follow the dependencies graph \eqref{eq:dependencies-graph-for-BA-IB-variables} (in Appendix \ref{sub:BA-IB-jacobian-appendix:calculation-setup-and-goals}) backward between these coordinate systems to exchange the coordinates of an implicit derivative. 
For example, consider $p_i(\y|\xhat)$ and $p_i(\xhat)$ as functions of the encoder $p_i(\xhatp|\xp)$ preceding it in the graph \eqref{eq:dependencies-graph-for-BA-IB-variables}. 
When at an IB root, multiplying by the coordinates exchange Jacobian yields 
\begin{align}
	\dbeta{\log p_\beta(\y|\xhat)} &=			\label{eq:coordinates-exchange-enc-to-dec-implicit-dec:appendix}
	\frac{d\log p_\beta(\y|\xhat)}{d\log p_\beta(\xhatp|\xp)} \dbeta{\log p_\beta(\xhatp|\xp)} 
	\quad \text{and} \\
	\dbeta{\log p_\beta(\xhat)} &=				\label{eq:coordinates-exchange-enc-to-dec-implicit-marginal:appendix}
	\frac{d\log p_\beta(\xhat)}{d\log p_\beta(\xhatp|\xp)} \dbeta{\log p_\beta(\xhatp|\xp)} \;.
\end{align}
Similarly, considering an encoder $p_\beta(\xhat|\x)$ as a function of $p_\beta(\yp|\xhatp)$ and $p_\beta(\xhatp)$,
\begin{equation}								\label{eq:coordinates-exchange-dec-to-enc-implicit-enc:appendix}
	\dbeta{\log p_\beta(\xhat|\x)} = 
	\frac{d\log p_\beta(\xhat|\x)}{d\log p_\beta(\yp|\xhatp)} \dbeta{\log p_\beta(\yp|\xhatp)} +
	\frac{d\log p_\beta(\xhat|\x)}{d\log p_\beta(\xhatp)} \dbeta{\log p_\beta(\xhatp)} +
	\pbeta{\log p_\beta(\xhat|\x)} \;.
\end{equation}
The last term $\pbeta{\log p_\beta(\xhat|\x)}$ in \eqref{eq:coordinates-exchange-dec-to-enc-implicit-enc:appendix} stems from the fact that the encoder Equation \algref{algo:BA-IB}{eq:IB-BA-new-direct-enc} depends explicitly on $\beta$, unlike the decoder and marginal Equations 
\algref{algo:BA-IB}{eq:IB-BA-decoder-eq} and \algref{algo:BA-IB}{eq:IB-BA-cluster_marginal}. 
cf., the comments around \eqref{eq:partial-deriv-def} in Appendix \ref{sub:BA-IB-jacobian-appendix:calculation-setup-and-goals}. 

The matrices $\frac{d\log p_\beta(\y|\xhat)}{d\log p_\beta(\xhatp|\xp)} $ and $\frac{d\log p_\beta(\xhat)}{d\log p_\beta(\xhatp|\xp)} $ for exchanging from encoder to decoder coordinates follow from the chain rule, and are calculated in Appendix \ref{subsub:encoder-to-decoder-coords-jacobian-appendix} below, at Equations \eqref{eq:coordinates-exchange-enc-to-dec-marginal-in-terms-of-enc:appendix} and \eqref{eq:coordinates-exchange-enc-to-dec-decoder-in-terms-of-enc:appendix}.
Similarly, the matrices $\frac{d\log p_\beta(\xhat|\x)}{d\log p_\beta(\yp|\xhatp)} $ and $\frac{d\log p_\beta(\xhat|\x)}{d\log p_\beta(\xhatp)} $ and the partial derivative $\pbeta{\log p_\beta(\xhat|\x)}$ for exchanging from decoder to encoder coordinates are Equations \eqref{eq:coordinates-exchange-dec-to-enc-enc-in-terms-of-dec:appendix}, \eqref{eq:coordinates-exchange-dec-to-enc-enc-in-terms-of-marginal:appendix} and \eqref{eq:coordinates-exchange-dec-to-enc-enc-in-terms-of-beta:appendix}, in Appendix \ref{subsub:decoder-to-encoder-coords-jacobian-appendix}.

\medskip
\subsubsection{Exchanging from encoder to decoder coordinates}

\label{subsub:encoder-to-decoder-coords-jacobian-appendix}

An input encoder $p_i(\xhatp|\xp)$ determines a decoder $p_i(\y|\xhat)$ and a marginal $p_i(\xhat)$.
As in previous subsections, we follow the dependencies graph \eqref{eq:dependencies-graph-for-BA-IB-variables} along all the paths between these. 

\medskip 
Using diagram \eqref{eq:log-variable-dep-w-derivatives-no-normalization} from Section \ref{subsub:BA-IB-jacobian-appendix:differentiating-along-depend-graph}, for the marginal one has 
\begin{equation}			\label{eq:coordinates-exchange-enc-to-dec-marginal-in-terms-of-enc:appendix}
	\frac{d\log p_i(\xhat)}{d\log p_i(\xhatp|\xp)} =
	p_i(\xp|\xhatp) \; \delta_{\xhat, \xhatp} \;.
\end{equation}
While for the decoder,
\begin{multline}
	\frac{d\log p_i(\y|\xhat)}{d\log p_i(\xhatp|\xp)} =
	\frac{\partial \log p_i(\y|\xhat)}{\partial \log p_i(x_1|\hat{x}_2)} \left[ 
		\frac{\partial \log p_i(x_1|\hat{x}_2)}{\partial \log p_i(\xhatp|\xp)} +
		\frac{\partial \log p_i(x_1|\hat{x}_2)}{\partial \log p_i(\hat{x}_3)} 
		\frac{\partial \log p_i(\hat{x}_3)}{\partial \log p_i(\xhatp|\xp)}
	\right] \\ =
	\frac{p(\y|x_1) p_i(x_1|\hat{x}_2)}{p_i(\y|\hat{x}_2)} \; \delta_{\hat{x}_2, \xhat} \Big[ 
		\delta_{x_1, \xp} \; \delta_{\hat{x}_2, \xhatp} 
		-\delta_{\hat{x}_2, \hat{x}_3} \; p_i(\xp|\hat{x}_3) \; \delta_{\hat{x}_3, \xhatp}
	\Big] 
\end{multline}
Summing over the three dummy variables as before, the latter simplifies to
\begin{equation}			\label{eq:coordinates-exchange-enc-to-dec-decoder-in-terms-of-enc:appendix}
	\frac{d\log p_i(\y|\xhat)}{d\log p_i(\xhatp|\xp)} =
	\left[ 
		\frac{p(\y|\xp) }{p_i(\y|\xhat)} - 1 
	\right] p_i(\xp|\xhat) \; \delta_{\xhat, \xhatp} \;.
\end{equation}

\medskip
\subsubsection{Exchanging from decoder to encoder coordinates}

\label{subsub:decoder-to-encoder-coords-jacobian-appendix}

In the other way around, a decoder $p_i(\yp|\xhatp)$ and a marginal $p_i(\xhatp)$ determine the subsequent encoder $p_{i+1}(\xhat|\x)$.
Using diagram \eqref{eq:log-variable-dep-w-derivatives-no-normalization}, one has
\begin{multline}
	\frac{d\log p_{i+1}(\xhat|\x)}{d\log p_i(\yp|\xhatp)} =
	\frac{\partial \log p_{i+1}(\xhat|\x)}{\partial \log p_i(\yp|\xhatp)} +
	\frac{\partial \log p_{i+1}(\xhat|\x)}{\partial \log Z_i(x_1)} \frac{\partial \log Z_i(x_1)}{\partial \log p_i(\yp|\xhatp)} \\ =
	\beta \; \delta_{\xhat, \xhatp} \; p(\yp|\x) -
	\delta_{\x, x_1} \; \beta \; p_{i+1}(\xhatp|x_1) p(\yp|x_1)
\end{multline}
Summing over the dummy variable $x_1$, this is the coordinates exchange Jacobian $J_{\text{dec}}^{\text{enc}}$ mentioned in Section \ref{sec:coords-exchange-for-the-IB},
\begin{equation}				\label{eq:coordinates-exchange-dec-to-enc-enc-in-terms-of-dec:appendix}
	\frac{d\log p_{i+1}(\xhat|\x)}{d\log p_i(\yp|\xhatp)} =
	\beta \; p(\yp|\x) \Big[ \delta_{\xhat, \xhatp} - p_{i+1}(\xhatp|\x) \Big]
\end{equation}

Next, for the derivative with respect to the marginal,
\begin{multline}
	\frac{d\log p_{i+1}(\xhat|\x)}{d\log p_i(\xhatp)} =
	\frac{\partial \log p_{i+1}(\xhat|\x)}{\partial \log p_i(\xhatp)} +
	\frac{\partial \log p_{i+1}(\xhat|\x)}{\partial \log Z_i(x_1)} \frac{\partial \log Z_i(x_1)}{\partial \log p_i(\xhatp)} \\ +
	\left[\frac{\partial \log p_{i+1}(\xhat|\x)}{\partial \log Z_i(x_1)} \frac{\partial \log Z_i(x_1)}{\partial \log p_i(y_2|\hat{x}_3)} + \frac{\partial \log p_{i+1}(\xhat|\x)}{\partial \log p_i(y_2|\hat{x}_3)} \right] 
	\frac{\partial \log p_i(y_2|\hat{x}_3)}{\partial \log p_i(x_4|\hat{x}_5)} \frac{\partial \log p_i(x_4|\hat{x}_5)}{\partial \log p_i(\xhatp)}
	\\ =
	\delta_{\xhat, \xhatp} -
	\delta_{\x, x_1} \; p_{i+1}(\xhatp|x_1) \\ +
	\Big[ -\delta_{\x, x_1} \; \beta \; p_{i+1}(\hat{x}_3|x_1) p(y_2|x_1) + \beta \; \delta_{\hat{x}_3, \xhat} \; p(y_2|\x) \Big]
	\frac{p(y_2|x_4) p_i(x_4|\hat{x}_3)}{p_i(y_2|\hat{x}_3)} \; \delta_{\hat{x}_3, \hat{x}_5} \cdot (-\delta_{\hat{x}_5, \xhatp})
\end{multline}
Summing over the five dummy variables, this is the coordinates exchange Jacobian $J_{\text{mrg}}^{\text{enc}}$ from Section \ref{sec:coords-exchange-for-the-IB},
\begin{equation}				\label{eq:coordinates-exchange-dec-to-enc-enc-in-terms-of-marginal:appendix}
	\frac{d\log p_{i+1}(\xhat|\x)}{d\log p_i(\xhatp)} =
	\left( 1 - \beta \right) \Big[ \delta_{\xhat, \xhatp} - p_{i+1}(\xhatp|\x) \Big]
\end{equation}

Finally, note that the encoder Equation \algref{algo:BA-IB}{eq:IB-BA-new-direct-enc} depends on $\beta$ explicitly, rather than indirectly only via its other variables. 
So, to calculate the partial derivative term $\pbeta{\log p_{i+1}(\xhat|\x)}$ in \eqref{eq:coordinates-exchange-dec-to-enc-implicit-enc:appendix}, write as follows for $\log Z$, 
\begin{multline}			\label{eq:partial-beta-deriv-of-Z-in-appendix}
	\frac{\partial}{\partial \beta} Z_i(x, \beta) \overset{\algref{algo:BA-IB}{eq:IB-BA-partition-func}}{=}
	\sum_{\hat{x}} p_{i}(\hat{x}) \frac{\partial}{\partial \beta} \exp{\left\{ -\beta \; D_{KL}\big[p(\y|\x) || p_{i}(y|\hat{x})\big] \right\}} \\ =
	-\sum_{\hat{x}} p_{i}(\hat{x}) D_{KL}\big[p(\y|\x) || p_{i}(y|\hat{x})\big] \exp{\left\{ -\beta \; D_{KL}\big[p(\y|\x) || p_{i}(y|\hat{x})\big] \right\}} 
\end{multline}
Thus,
\begin{multline}			\label{eq:partial-beta-deriv-of-log-Z-in-appendix}
	\frac{\partial}{\partial \beta} \log Z_i(x, \beta) =
	\frac{1}{Z_i(x, \beta)} \frac{\partial}{\partial \beta} Z_i(x, \beta) \\ \overset{\eqref{eq:partial-beta-deriv-of-Z-in-appendix}}{=}
	-\sum_{\hat{x}} \frac{p_{i}(\hat{x}) \exp{\left\{ -\beta \; D_{KL}\big[p(\y|\x) || p_{i}(y|\hat{x})\big] \right\}}}{Z_i(x, \beta)} D_{KL}\big[p(\y|\x) || p_{i}(y|\hat{x})\big] \\ \overset{\algref{algo:BA-IB}{eq:IB-BA-new-direct-enc}}{=}
	-\sum_{\hat{x}} p_{i+1}(\hat{x}|x) D_{KL}\big[p(\y|\x) || p_{i}(y|\hat{x})\big] \;.
\end{multline}

And so, from the encoder Equation \algref{algo:BA-IB}{eq:IB-BA-new-direct-enc} we have
\begin{multline}				\label{eq:coordinates-exchange-dec-to-enc-enc-in-terms-of-beta:appendix}
	\pbeta{\log p_{i+1}(\xhat|\x)} =
	\pbeta{\log p_i(\xhat)} 
	- \pbeta{\log Z_i(x, \beta)}
	- \pbeta{} \Big(\beta D_{KL}\big[p(\y|\x) || p_i(\y|\xhat)\big] \Big)
	\\ \overset{\eqref{eq:partial-beta-deriv-of-log-Z-in-appendix}}{=}
	\sum_{\xhatpp} p_{i+1}(\xhatpp|\x) D_{KL}\big[p(\y|\x) || p_{i}(\y|\xhatpp)\big] 
	- D_{KL}\big[p(\y|\x) || p_i(\y|\xhat)\big] 
\end{multline}
where the term $\pbeta{\log p_i(\xhat)}$ vanishes since it is considered as an independent variable here.

\medskip
\section{Proof of Lemma \ref{lem:smaller-matrix-for-ker-of-J-in-decoder-coords}, on the kernel of the Jacobian of the IB operator in log-decoder coordinates}

\label{sec:kernel-of-BA-IB-J-decoder-coords:appendix}

We prove Lemma \ref{lem:smaller-matrix-for-ker-of-J-in-decoder-coords} from Section \ref{sec:IB-ODE}, using the results of Appendix \ref{sec:first-order-deriv-tensors-of-BA-IB-appendix}. 

\medskip
In the first direction, suppose that $\big( \left(v_{\y, \xhat}\right)_{\y, \xhat}, \left(u_{\xhat}\right)_{\xhat} \big)$ is a vector in the \textit{left} kernel of the Jacobian of the IB operator \eqref{eq:IB-operator-def} in log-decoder coordinates, $I - D_{\log p(\y|\xhat), \log p(\xhat)} BA_\beta$, as in \eqref{eq:IB-beta-ODE-in-decoder-coords} in Section \ref{sec:IB-ODE}.
Using the Jacobian's implicit form \eqref{eq:BA-Jacob-wrt-decoder-coords-as-block-matrix-implicit} (Appendix \ref{sub:BA-IB-jacobian-appendix:decoder-deriv-matrix}), this is to say that
\begin{align}
	v_{\yp, \xhatp}	&=			\label{eq:lkernel-vec-decoder-dist-coords}
	\sum_{\y, \xhat}	v_{\y, \xhat}	\frac{d\log p_{i+1}(\y|\xhat)}{d \log p_i(\yp|\xhatp)} +
	\sum_{\xhat} 		u_{\xhat}		\frac{d\log p_{i+1}(\xhat)}{d \log p_i(\yp|\xhatp)} \quad \text{and} \\
	u_{\xhatp}		&=			\label{eq:lkernel-vec-decoder-marginal-coords}
	\sum_{\y, \xhat}	v_{\y, \xhat}	\frac{d\log p_{i+1}(\y|\xhat)}{d \log p_i(\xhatp)} +
	\sum_{\xhat} 		u_{\xhat}		\frac{d\log p_{i+1}(\xhat)}{d \log p_i(\xhatp)}
\end{align}
hold, for every $\yp$ and $\xhatp$.
We spell out and manipulate these equations to obtain the desired result.

By the Jacobian's explicit form \eqref{eq:BA-Jacob-wrt-decoder-coords-as-block-matrix-explicit} from Appendix \ref{sub:BA-IB-jacobian-appendix:decoder-deriv-matrix}, Equation \eqref{eq:lkernel-vec-decoder-dist-coords} spells out as
\begin{multline}			\label{eq:lkernel-vec-decoder-dist-coords-spelled-out}
	v_{\yp, \xhatp} =
	\beta \cdot \sum_{\y, \xhat}	v_{\y, \xhat} \sum_{\xhatpp, \ypp} 
	\left( \delta_{\xhatpp, \xhatp} - \delta_{\xhat, \xhatp} \right)
	\cdot \left(1 - \tfrac{\delta_{\ypp, \y} }{p_{i+1}(\y|\xhat)} \right)
	C(\xhat, \xhatpp; i+1)_{\yp, \ypp} \\
	+ \beta \cdot \sum_{\xhat} u_{\xhat}
	\Big[ \delta_{\xhat, \xhatp} \; p_{i+1}(\yp|\xhat) - B(\xhat, \xhatp; i+1)_{\yp} \Big] \;,
\end{multline}
while the second Equation \eqref{eq:lkernel-vec-decoder-marginal-coords} spells out as
\begin{multline}			\label{eq:lkernel-vec-decoder-marginal-coords-spelled-out}
	u_{\xhatp} =
	\left( 1 - \beta \right) \cdot \sum_{\y, \xhat} v_{\y, \xhat} \sum_{\ypp} \Big[ 
	1 - \tfrac{\delta_{\ypp, \y}}{p_{i+1}(\y|\xhat)} 
	\Big] B(\xhat, \xhatp; i+1)_{\ypp} 
	\\ +
	\left(1 - \beta \right) \cdot \sum_{\xhat} u_{\xhat} \Big( \delta_{\xhat, \xhatp} - A(\xhat, \xhatp; i+1) \Big) \;.
\end{multline}
Next, we expand and simplify each of the terms in \eqref{eq:lkernel-vec-decoder-dist-coords-spelled-out} and \eqref{eq:lkernel-vec-decoder-marginal-coords-spelled-out}, using the definition \eqref{eq:B_C_defs-for-BA-Jacob-wrt-decoder-coords-appendix} of $A, B$ and $C$ from Appendix \ref{sub:BA-IB-jacobian-appendix:decoder-deriv-matrix}. 

For the first summand to the right of \eqref{eq:lkernel-vec-decoder-dist-coords-spelled-out},
\begin{multline}			\label{eq:vector-is-in-lker-decoder-jacobian-subblock-explicit}
	\beta \cdot \sum_{\y, \xhat} v_{\y, \xhat} 
	\sum_{\xhatpp, \ypp} 
	\left( \delta_{\xhatpp, \xhatp} - \delta_{\xhat, \xhatp} \right)
	\left(1 - \tfrac{\delta_{\ypp, \y} }{p_{i+1}(\y|\xhat)} \right)
	C(\xhat, \xhatpp; i+1)_{\yp, \ypp} 
	\\ \overset{\eqref{eq:B_C_defs-for-BA-Jacob-wrt-decoder-coords-appendix}}{=}
	\beta \cdot \sum_{\y, \xhat} v_{\y, \xhat} \sum_{\xhatpp, \ypp} 
	\left( \delta_{\xhatpp, \xhatp} - \delta_{\xhat, \xhatp} \right)
	\left(1 - \tfrac{\delta_{\ypp, \y} }{p_{i+1}(\y|\xhat)} \right)
	\sum_{\x} p(\yp|\x) p(\ypp|\x) p_{i+1}(\xhatpp|\x) p_{i+1}(\x|\xhat)
\end{multline}
We simplify each of the four addends to the right of \eqref{eq:vector-is-in-lker-decoder-jacobian-subblock-explicit} while temporarily ignoring the $\beta$ coefficient. 
For the $\delta_{\xhatpp, \xhatp} \cdot 1$ term,
\begin{multline}			\label{eq:ldecoder-jacobian-subblock-term1}
	\sum_{\y, \xhat} v_{\y, \xhat} \sum_{\xhatpp, \ypp} 
	\delta_{\xhatpp, \xhatp} \sum_{\x} p(\yp|\x) p(\ypp|\x) p_{i+1}(\xhatpp|\x) p_{i+1}(\x|\xhat) \\ =
	\sum_{\y, \xhat} v_{\y, \xhat} \sum_{\x} p(\yp|\x) p_{i+1}(\xhatp|\x) p_{i+1}(\x|\xhat) 
\end{multline}
For the $- \delta_{\xhatpp, \xhatp} \cdot \frac{\delta_{\ypp, \y} }{p_{i+1}(\y|\xhat)}$ term,
\begin{multline}			\label{eq:ldecoder-jacobian-subblock-term2}
	- \sum_{\y, \xhat} v_{\y, \xhat} 
	\sum_{\xhatpp, \ypp} \delta_{\xhatpp, \xhatp} \delta_{\ypp, \y} 
	\sum_{\x} \tfrac{1}{p_{i+1}(\y|\xhat)} \; p(\yp|\x) p(\ypp|\x) p_{i+1}(\xhatpp|\x) p_{i+1}(\x|\xhat) \\ =
	- \sum_{\y, \xhat} v_{\y, \xhat} \sum_{\x} \tfrac{1}{p_{i+1}(\y|\xhat)} \; p(\yp|\x) p(\y|\x) p_{i+1}(\xhatp|\x) p_{i+1}(\x|\xhat) 
\end{multline}
For the $-\delta_{\xhat, \xhatp}\cdot 1$ term, 
\begin{multline}			\label{eq:ldecoder-jacobian-subblock-term3}
	-\sum_{\y, \xhat} v_{\y, \xhat} \sum_{\xhatpp, \ypp} 
	\delta_{\xhat, \xhatp} \sum_{\x} p(\yp|\x) p(\ypp|\x) p_{i+1}(\xhatpp|\x) p_{i+1}(\x|\xhat) \\ =
	-\sum_{\y} v_{\y, \xhatp} \sum_{\x} p(\yp|\x) p_{i+1}(\x|\xhatp) =
	-\sum_{\y} v_{\y, \xhatp} \; p_{i+1}(\yp|\xhatp) 
\end{multline}
And for the last $-\delta_{\xhat, \xhatp} \cdot \tfrac{-\delta_{\ypp, \y} }{p_{i+1}(\y|\xhat)}$ term,
\begin{multline}			\label{eq:ldecoder-jacobian-subblock-term4}
	\sum_{\y, \xhat} v_{\y, \xhat} \sum_{\xhatpp, \ypp} 
	\delta_{\xhat, \xhatp} \cdot \tfrac{\delta_{\ypp, \y} }{p_{i+1}(\y|\xhat)}
	\sum_{\x} p(\yp|\x) p(\ypp|\x) p_{i+1}(\xhatpp|\x) p_{i+1}(\x|\xhat) \\ =
	\sum_{\y} \frac{v_{\y, \xhatp} }{p_{i+1}(\y|\xhatp)} \sum_{\x} p(\yp|\x) p(\y|\x) p_{i+1}(\x|\xhatp) 
\end{multline}
Collecting \eqref{eq:ldecoder-jacobian-subblock-term1}, \eqref{eq:ldecoder-jacobian-subblock-term2}, \eqref{eq:ldecoder-jacobian-subblock-term3} and \eqref{eq:ldecoder-jacobian-subblock-term4} back into \eqref{eq:vector-is-in-lker-decoder-jacobian-subblock-explicit}, we obtain
\begin{multline}			\label{eq:bottom-line-of-lkernel-vec-for-upper-left-block}
	\beta \cdot \sum_{\y, \xhat} v_{\y, \xhat} \sum_{\x} p(\yp|\x) p_{i+1}(\xhatp|\x) p_{i+1}(\x|\xhat) \Big[ 1 - \tfrac{p(\y|\x) }{p_{i+1}(\y|\xhat)} \Big]
	\\ +
		\beta \cdot \sum_{\y} v_{\y, \xhatp} \tfrac{1 }{p_{i+1}(\y|\xhatp)} \sum_{\x} p(\y|\x) p(\yp|\x) p_{i+1}(\x|\xhatp) 
		- \beta \cdot p_{i+1}(\yp|\xhatp) \sum_{\y} v_{\y, \xhatp} 
\end{multline}
for the first summand to the right of \eqref{eq:lkernel-vec-decoder-dist-coords-spelled-out}.

The second summand to the right of \eqref{eq:lkernel-vec-decoder-dist-coords-spelled-out} equals,
\begin{multline}			\label{eq:bottom-line-of-lkernel-vec-for-upper-right-block}
	\beta \cdot \sum_{\xhat} u_{\xhat}
	\Big[ \delta_{\xhat, \xhatp} \; p_{i+1}(\yp|\xhat) - B(\xhat, \xhatp; i+1)_{\yp} \Big] 
	\\ \overset{\eqref{eq:B_C_defs-for-BA-Jacob-wrt-decoder-coords-appendix}}{=}
	\beta \cdot u_{\xhatp} \; p_{i+1}(\yp|\xhatp) - 
	\beta \cdot \sum_{\x} p(\yp|\x) p_{i+1}(\xhatp|\x) \sum_{\xhat} u_{\xhat} \; p_{i+1}(\x|\xhat) 
\end{multline}

Combining \eqref{eq:bottom-line-of-lkernel-vec-for-upper-left-block} and \eqref{eq:bottom-line-of-lkernel-vec-for-upper-right-block}, Equation \eqref{eq:lkernel-vec-decoder-dist-coords-spelled-out} is equivalent to
\begin{multline}			\label{eq:bottom-line-of-lkernel-for-decoder-decoder-coords}
	\tfrac{1}{\beta} \cdot v_{\yp, \xhatp} 
	+ p_{i+1}(\yp|\xhatp) \sum_{\y} v_{\y, \xhatp} 
	- u_{\xhatp} \; p_{i+1}(\yp|\xhatp) 
	\\ = \sum_{\y, \xhat} v_{\y, \xhat} \sum_{\x} p(\yp|\x) p_{i+1}(\xhatp|\x) p_{i+1}(\x|\xhat) \Big[ 1 - \tfrac{p(\y|\x) }{p_{i+1}(\y|\xhat)} \Big]
	\\ 
	+ \sum_{\y} v_{\y, \xhatp} \sum_{\x} p(\yp|\x) p_{i+1}(\x|\xhatp) \tfrac{p(\y|\x) }{p_{i+1}(\y|\xhatp)} 
	- \sum_{\x} p(\yp|\x) p_{i+1}(\xhatp|\x) \sum_{\xhat} u_{\xhat} \; p_{i+1}(\x|\xhat) 
\end{multline}
for any $\yp$ and $\xhatp$.
Summing \eqref{eq:bottom-line-of-lkernel-for-decoder-decoder-coords} over $\yp$ and simplifying, we obtain
\begin{multline}			\label{eq:yp-marginalization-over-bottom-line-of-lkernel-for-decoder-decoder-coords}
	\tfrac{1 }{\beta} \cdot \sum_{\y} v_{\y, \xhatp} - u_{\xhatp} \\ =
	\sum_{\y, \xhat} v_{\y, \xhat} \sum_{\x} p_{i+1}(\xhatp|\x) p_{i+1}(\x|\xhat) \Big[ 1 - \tfrac{p(\y|\x) }{p_{i+1}(\y|\xhat)} \Big]
	- \sum_{\x} p_{i+1}(\xhatp|\x) \sum_{\xhat} u_{\xhat} \; p_{i+1}(\x|\xhat) 
\end{multline}
for any $\xhatp$.

Next, we expand and simplify Equation \eqref{eq:lkernel-vec-decoder-marginal-coords-spelled-out}. Using the definition \eqref{eq:B_C_defs-for-BA-Jacob-wrt-decoder-coords-appendix} of $B$, the first summand to its right can be written as
\begin{equation}			\label{eq:lkernel-decoder-marginal-coords-first-term}
	\left( 1 - \beta \right) \cdot \sum_{\y, \xhat} v_{\y, \xhat} \sum_{\x} p_{i+1}(\xhatp|\x) p_{i+1}(\x|\xhat) \Big[ 
	1 - \tfrac{p(\y|\x) }{p_{i+1}(\y|\xhat)} \Big] \;.
\end{equation}
Similarly, the second summand to the right of \eqref{eq:lkernel-vec-decoder-marginal-coords-spelled-out} can be written as
\begin{equation}			\label{eq:lkernel-decoder-marginal-coords-second-term}
	\left(1 - \beta \right) \cdot \Big[ u_{\xhatp} - \sum_{\x} p_{i+1}(\xhatp|\x) \sum_{\xhat} u_{\xhat} \; p_{i+1}(\x|\xhat) \Big] \;.
\end{equation}
Combining \eqref{eq:lkernel-decoder-marginal-coords-first-term} and \eqref{eq:lkernel-decoder-marginal-coords-second-term}, Equation \eqref{eq:lkernel-vec-decoder-marginal-coords-spelled-out} can now be written explicitly, 
\begin{equation}			\label{eq:bottom-line-of-lkernel-for-decoder-marginal-coords}
	\tfrac{\beta}{1 - \beta} \cdot u_{\xhatp} =
	\sum_{\y, \xhat} v_{\y, \xhat} \sum_{\x} p_{i+1}(\xhatp|\x) p_{i+1}(\x|\xhat) \Big[ 1 - \tfrac{p(\y|\x) }{p_{i+1}(\y|\xhat)} \Big] 
	- \sum_{\x} p_{i+1}(\xhatp|\x) \sum_{\xhat} u_{\xhat} \; p_{i+1}(\x|\xhat) 
\end{equation}
for every $\xhatp$.
Next, subtracting \eqref{eq:bottom-line-of-lkernel-for-decoder-marginal-coords} from \eqref{eq:yp-marginalization-over-bottom-line-of-lkernel-for-decoder-decoder-coords}, we obtain
\begin{equation}		\label{eq:relation-between-u-to-sum_y-of-v}
	u_{\xhat} = \tfrac{1 - \beta}{\beta} \cdot \sum_{\y} v_{\y, \xhat} 
\end{equation}
for any $\xhat$.

Substituting \eqref{eq:relation-between-u-to-sum_y-of-v} into \eqref{eq:bottom-line-of-lkernel-for-decoder-decoder-coords} and using the decoder Equation \algref{algo:BA-IB}{eq:IB-BA-decoder-eq} to expand $p_{i+1}(\yp|\xhatp)$ there,
\begin{multline}			\label{eq:towards-v-coords-are-beta-inv-eigenvector-in-proof}
	\tfrac{1}{\beta} \cdot v_{\yp, \xhatp} 
	= \sum_{\y, \xhat} v_{\y, \xhat} \sum_{\x} p_{i+1}(\xhatp|\x) p(\yp|\x) p_{i+1}(\x|\xhat) \Big[ \tfrac{2\beta - 1}{\beta} - \tfrac{p(\y|\x) }{p_{i+1}(\y|\xhat)} \Big]
	\\ 
	- \sum_{\y} v_{\y, \xhatp} \sum_{\x} p(\yp|\x) p_{i+1}(\x|\xhatp) \cdot \tfrac{2\beta - 1}{\beta} 
	+ \sum_{\y} v_{\y, \xhatp} \sum_{\x} p(\yp|\x) p_{i+1}(\x|\xhatp) \tfrac{p(\y|\x) }{p_{i+1}(\y|\xhatp)} 
\end{multline}
Next, inserting $\sum_{\xhat} \delta_{\xhat, \xhatp}$ into the sums on the last line,
\begin{multline}
	\tfrac{1}{\beta} \cdot v_{\yp, \xhatp} 
	= \sum_{\y, \xhat} v_{\y, \xhat} \sum_{\x} p_{i+1}(\xhatp|\x) p(\yp|\x) p_{i+1}(\x|\xhat) \Big[ \tfrac{2\beta - 1}{\beta} - \tfrac{p(\y|\x) }{p_{i+1}(\y|\xhat)} \Big]
	\\ 
	- \sum_{\y, \xhat} v_{\y, \xhat} \sum_{\x} \delta_{\xhat, \xhatp} \; p(\yp|\x) p_{i+1}(\x|\xhat) \Big[ \tfrac{2\beta - 1}{\beta} - \tfrac{p(\y|\x) }{p_{i+1}(\y|\xhat)} \Big]
\end{multline}
Finally, this simplifies to
\begin{equation}			\label{eq:v-coords-are-beta-inv-eigenvector-in-proof}
	v_{\yp, \xhatp} 
	= \sum_{\y, \xhat} v_{\y, \xhat} \sum_{\x} p(\yp|\x) \Big[ \delta_{\xhat, \xhatp} - p_{i+1}(\xhatp|\x) \Big] p_{i+1}(\x|\xhat) \Big[ \beta \cdot \tfrac{p(\y|\x) }{p_{i+1}(\y|\xhat)} + \left(1 - 2\beta\right) \Big] 
\end{equation}
The latter is to say that $\left(v_{\y, \xhat}\right)_{\y, \xhat}$ is a left-eigenvector of the eigenvalue 1 of the matrix to the right. At an IB root, this is precisely the matrix $S$ \eqref{eq:smaller-matrix-for-ker-of-IB-operator-in-dec-coords-main-text} from the Lemma's statement, as desired.

As a side note, we comment that Equations \eqref{eq:lkernel-vec-decoder-dist-coords} and \eqref{eq:lkernel-vec-decoder-marginal-coords} also imply
\begin{equation}		\label{eq:normalization-v-sum-over-clusters}
	\forall \y \; \sum_{\xhat} v_{\y, \xhat} = 0
	\quad \text{and} \quad
	\sum_{\xhat} u_{\xhat} = 0 \;,
\end{equation}
which can be seen by summing \eqref{eq:bottom-line-of-lkernel-for-decoder-decoder-coords} and \eqref{eq:bottom-line-of-lkernel-for-decoder-marginal-coords} respectively over $\xhatp$, and simplifying.

\medskip
At the other direction, let $\bm{v} := \left(v_{\y, \xhat}\right)_{\y, \xhat}$ be a left-eigenvector of the eigenvalue $1$ of $S$ \eqref{eq:smaller-matrix-for-ker-of-IB-operator-in-dec-coords-main-text}. That is, assume that Equation \eqref{eq:v-coords-are-beta-inv-eigenvector-in-proof} holds. Define a vector $\bm{u} := (u_{\xhat})_{\xhat}$ by Equation \eqref{eq:relation-between-u-to-sum_y-of-v}.
Reversing the algebra, \eqref{eq:v-coords-are-beta-inv-eigenvector-in-proof} is equivalent to \eqref{eq:towards-v-coords-are-beta-inv-eigenvector-in-proof}.
Substituting \eqref{eq:relation-between-u-to-sum_y-of-v} into the latter yields back \eqref{eq:bottom-line-of-lkernel-for-decoder-decoder-coords}, which is equivalent to the explicit form \eqref{eq:lkernel-vec-decoder-dist-coords-spelled-out} of Equation \eqref{eq:lkernel-vec-decoder-dist-coords}.
Next, summing \eqref{eq:bottom-line-of-lkernel-for-decoder-decoder-coords} over $\yp$ and simplifying yields \eqref{eq:yp-marginalization-over-bottom-line-of-lkernel-for-decoder-decoder-coords}. Adding the latter to \eqref{eq:relation-between-u-to-sum_y-of-v} yields back \eqref{eq:bottom-line-of-lkernel-for-decoder-marginal-coords}, which is equivalent to Equation \eqref{eq:lkernel-vec-decoder-marginal-coords-spelled-out}, the explicit form of \eqref{eq:lkernel-vec-decoder-marginal-coords}.
To conclude, both of the Equations \eqref{eq:lkernel-vec-decoder-dist-coords} and \eqref{eq:lkernel-vec-decoder-marginal-coords} hold, as claimed.

\medskip 
\section{Approximate error analysis for deterministic annealing and for Euler's method with BA}

\label{sec:approx-error-analysis-for-BA-and-Euler-for-IB-appendix}
Complementing the results of Section \ref{sec:euler-method}, we provide an approximate error analysis for two computation methods for the IB: deterministic annealing and Euler's method combined with a fixed number of BA iterations.

\medskip 
First, we recap the linearization argument around \cite[Equation (10)]{agmon2021critical}. 
Denote repeated BA iterations initialized at $\bm{p}_0$ by
\begin{equation}
	\bm{p}_{k+1} := BA_\beta[\bm{p}_k] \;.
\end{equation}
Linearizing around a fixed-point $\bm{p}_\beta$ of BA,
\begin{equation}
	BA [\bm{p}_k] \simeq \bm{p}_\beta + D \; BA_\beta \rvert_{\bm{p}_\beta} \cdot \left( \bm{p}_k - \bm{p}_\beta \right) \;,
\end{equation}
where $D \; BA_\beta\rvert_{\bm{p}_\beta}$ denotes the Jacobian matrix of $BA_\beta$ evaluated at $\bm{p}_\beta$. 
Rewriting in terms of the error $\delta \bm{p}_k := \bm{p}_k - \bm{p}_\beta$ of the $k$-th iterate,
\begin{equation}
	\delta\bm{p}_{k+1} \simeq D \; BA_\beta \rvert_{\bm{p}_\beta} \cdot \delta \bm{p}_k \;.
\end{equation}
Thus, to first order, repeated applications of $BA_\beta$ reduce the initial error according to
\begin{equation}		\label{eq:linearized-error-of-repeated-BA-iterations}
	\norm{\delta\bm{p}_{k+1}} \simeq 
	\norm{ \left(D \; BA_\beta \rvert_{\bm{p}_\beta}\right)^k \cdot \delta \bm{p}_0 } \;.
\end{equation}

\medskip 
Next, consider $k > 0$ applications of $BA_{\beta+\Delta \beta}$ to a root $\bm{p}_\beta$ at $\beta$. This is similar to deterministic annealing, but with a capped number of BA iterations. 
Plugging the initial error $\delta \bm{p}_0 := \bm{p}_\beta - \bm{p}_{\beta + \Delta\beta} \simeq -\Delta \beta \; \tfrac{d\bm{p}}{d\beta}\big\rvert_\beta$ into Equation \eqref{eq:linearized-error-of-repeated-BA-iterations} shows that this method is of the first order,
\begin{equation}
	\norm{\delta\bm{p}_{k+1} } \simeq 
	|\Delta \beta| \cdot \norm{\left(D \; BA_{\beta + \Delta \beta} \rvert_{\bm{p}_{\beta + \Delta \beta}}\right)^k \; \tfrac{d\bm{p}}{d\beta}\big\rvert_\beta } \;.
\end{equation}

\medskip 
Finally, we combine BA with Euler's method for the IB, Equation \eqref{eq:first-order-approx}. 
Consider $k > 0$ applications of $BA_{\beta+\Delta \beta}$ to the approximation $\bm{p}_\beta + \Delta \beta \; \tfrac{d\bm{p}}{d\beta}\big\rvert_\beta$ produced by an Euler method step. 
Its initial error is
\begin{equation}
	\delta \bm{p}_0 := 
	\bm{p}_\beta + \Delta \beta \; \tfrac{d\bm{p}}{d\beta}\big\rvert_\beta
	- \bm{p}_{\beta + \Delta\beta} = 
	-\tfrac{1}{2} (\Delta \beta)^2 \; \tfrac{d^2\bm{p}}{d\beta^2}\big\rvert_{\beta'} \;,
\end{equation}
where the last equality follows from the second-order expansion $\bm{p}_{\beta + \Delta\beta} = \bm{p}_{\beta} + \Delta \beta \; \tfrac{d\bm{p}}{d\beta}\big\rvert_\beta + \tfrac{1}{2} (\Delta \beta)^2 \; \tfrac{d^2\bm{p}}{d\beta^2}\big\rvert_{\beta'}$, with $\beta' \in [\beta, \beta + \Delta \beta]$. 
Similar to before, plugging this into Equation \eqref{eq:linearized-error-of-repeated-BA-iterations} shows that this method is of the second order,
\begin{equation}
	\norm{\delta\bm{p}_{k+1} } \simeq 
	\tfrac{1}{2} |\Delta \beta|^2 \cdot \norm{\left(D \; BA_{\beta + \Delta \beta} \rvert_{\bm{p}_{\beta + \Delta \beta}}\right)^k \; \tfrac{d^2\bm{p}}{d\beta^2}\big\rvert_{\beta'} } \;.
\end{equation}

\medskip
\section{An exact solution for a binary symmetric channel}

\label{sec:analytical-IB-sol-for-BSC-appendix}

Define an IB problem by $Y\sim \text{Bernoulli}(\tfrac{1}{2})$ and $X := Y \oplus Z$ for $Z\sim \text{Bernoulli}(\alpha)$ independent of $Y$, $0 < \alpha <\tfrac{1}{2}$, where $\oplus$ denotes addition modulo 2. Explicitly, it is given by $p_{Y|X} = \mat{1-\alpha & \alpha \\ \alpha & 1-\alpha}$ and $p_X = (\tfrac{1}{2}, \tfrac{1}{2})$. We synthesize exact solutions for this problem using Mrs. Gerber's Lemma \citep{wyner1973MrsGerbersLemma} and by following \citep{witsenhausen1975conditional}.

\medskip
Let $h(p) := -p \log p - (1-p) \log (1-p)$ be the binary entropy, with $h(0) := h(1) := 0$. It is injective on $[0, \tfrac{1}{2}]$, with a maximal value of $\log 2$ at $p = \tfrac{1}{2}$. So, its inverse function $h^{-1}$ is well defined on $[0, \log 2]$. 
Given a constraint $I_X\in [0, \log 2]$ on $I(\hat{X}; X)$, $I(\hat{X}; X) \leq I_X$, define a random variable $V \sim \text{Bernoulli}(\delta)$ and set $\hat{X} := X \oplus V$, where $\delta$ is defined by $h(\delta) = \log 2 - I_X$ or equivalently in terms of $h^{-1}$ by $\delta := h^{-1}(\log 2 - I_X)$. 
Explicitly, $p(\xhat|x) = \mat{1-\delta & \delta \\ \delta & 1-\delta}$, with its rows indexed by $\xhat$ and columns by $\x$. 
$\hat{X}$ is also a $\text{Bernoulli}(\tfrac{1}{2})$ variable since $X$ is, and so
\begin{equation}			\label{eq:BSC-I_X-constraint-holds}
	I(\hat{X}; X) = 
	H(\hat{X}) - H(\hat{X}|X) =
	\log 2 - h(\delta) = I_X,
\end{equation}
showing that the constraint on $I(\hat{X}; X)$ holds.
The chain $\hat{X} \to X \to Y$ of random variables is readily seen to be Markov. 
By \cite[Corollary 4]{wyner1973MrsGerbersLemma}, it follows that $I(\hat{X}; Y) \leq \log 2 - h(\alpha * \delta)$, where $a*b := a(1-b) + b(1-a)$. Finally, equality follows by Theorem 1 there.
Thus, the above $p(\xhat|x)$ is IB-optimal.

The above defines an IB solution $p(\xhat|x)$ as a function of $I_X$. However, our numerical computations are phrased in terms of the IB's Lagrange multiplier $\beta$. To that end, \cite[IV.A]{witsenhausen1975conditional} show that
\begin{equation}		\label{eq:beta-as-a-func-of-delta-for-BSC-witsenhausen75}
	\beta \cdot (1 - 2\alpha) \log \frac{1 - \alpha * \delta}{\alpha * \delta} = 
	\log \frac{1 - \delta}{\delta} \;,
\end{equation}
and that the bifurcation of this problem occurs at
\begin{equation}
	\beta_c = \frac{1}{(1 - 2\alpha)^2} \;.
\end{equation}
To conclude, we have $\beta = \beta(\delta)$ as a function of $\delta$, $\delta = \delta(I_X)$ as a function of $I_X$, and the encoder $p(\xhat|x)$ as a function of $\delta$.
These functional dependencies are summarized as follows,
\begin{equation}
	\xymatrix@C=4em{
		p(\xhat|x)\ar[r]		&		\delta\ar[dl]	\\
		I_X						&		\beta\ar[u]
	}
\end{equation}
where the variable at the tail of each arrow is a function of that at its head.

Writing $\bm{p} = \big(p(\xhat|x)\big)_{\xhat, x}$, its derivative with respect to $\beta$ can be calculated by the chain rule,
\begin{equation}			\label{eq:expression-for-encoder-deriv-wrt-beta-BSC-example}
	\frac{d\bm{p}}{d\beta} =
	\frac{d}{d\beta} \Big( \bm{p}\big(\beta^{-1}(\delta)\big) \Big) =
	\frac{d\bm{p}}{d\delta} \left(\frac{d\beta}{d\delta}\right)^{-1} \;,
\end{equation}
where we have applied the derivative of an inverse function $(f^{-1})' = \nicefrac{1}{f'}$ to $\beta(\delta)$ in \eqref{eq:beta-as-a-func-of-delta-for-BSC-witsenhausen75}, to differentiate $\delta(\beta)$. 
From the argument around \eqref{eq:BSC-I_X-constraint-holds}, $\frac{d\bm{p}}{d\delta} = \mat{-1 & 1 \\ 1 & -1}$.
While this yields an analytical expression for the derivative $\tfrac{d\bm{p}}{d\beta}$, both of the terms to the right of \eqref{eq:expression-for-encoder-deriv-wrt-beta-BSC-example} are evaluated at $\delta(\beta)$, for a given $\beta$ value. 
Although it is straightforward to compute $\delta(\beta)$ numerically from \eqref{eq:beta-as-a-func-of-delta-for-BSC-witsenhausen75}, this entails numerical error, especially as $\delta$ approaches $\nicefrac{1}{2}$ near the bifurcation.
For the solution with respect to decoder coordinates, an immediate application of the Bayes rule shows that
\begin{equation}
	p(\xhat) = \frac{1}{2}		\quad \quad \text{and} \quad \quad 
	p(\y|\xhat) = 
	\mat{
		\alpha * (1 - \delta)		&		\alpha * \delta		\\
		\alpha * \delta				&		\alpha * (1 - \delta)
	} \;,
\end{equation}
where the rows of $p(\y|\xhat)$ are indexed by $\y$, and columns by $\xhat$.
Along with $\tfrac{dp(\y|\xhat)}{d\delta} = (2\alpha-1)\cdot \mat{1 & -1 \\ -1 & 1}$, its derivatives with respect to $\beta$ follow as in \eqref{eq:expression-for-encoder-deriv-wrt-beta-BSC-example}.

\medskip
\section{Equivalent conditions for cluster-merging bifurcations}

\label{sec:equivalent-conds-for-cluster-splitting-appendix}

We briefly discuss the equivalent conditions for cluster-merging bifurcations in the IB (Subsection \ref{sub:continuous-IB-bifs}) found in the literature.

\medskip 
\cite[Section 4]{rose1990statistical} derive a condition for cluster-splitting phase transitions (Equation (17) there) in the context of fuzzy clustering. 
Following this, \cite[3.2 in Part III]{zaslavsky2019thesis} derives an analogous condition for cluster splitting in the IB,
\begin{equation}		\label{eq:Talis-C_X-condition-for-bif}
	\left( I - \beta \; C_X(\xhat; \beta) \right) \bm{u} = \bm{0} \;,
\end{equation}
which is Equation (12) there. 
Namely, for a cluster $\xhat$ to split it is necessary that $\nicefrac{1}{\beta}$ would be an eigenvalue of an $|\mathcal{X}|$-by-$|\mathcal{X}|$ matrix $C_X(\xhat; \beta)$, whose entries at an IB root are given by 
\begin{equation}		\label{eq:Cxx'_def}
	C_X(\xhat; \beta)_{\x, \xp} := \sum_{\y} \frac{p(\y|\x) p(\y|\xp) p_\beta(\xp|\xhat)}{p_\beta(\y|\xhat) } \;,
\end{equation}
and $I$ is the identity. 
While the coefficients matrix \eqref{eq:Cxx'_def} for the IB differs from the one for fuzzy clustering, inter-cluster interactions are explicitly neglected in both derivations (see therein). 
Indeed, the definition \eqref{eq:Cxx'_def} of $C_X$ involves the coordinates of cluster $\xhat$ alone, as one might expect when considering a root in either decoder or in inverse-encoder coordinates (Section \ref{sec:coords-exchange-for-the-IB}). 
Reversing the dynamics in $\beta$, condition \eqref{eq:Talis-C_X-condition-for-bif} characterizes cluster-merging bifurcations in the IB (Subsection \ref{sub:continuous-IB-bifs}). 

\cite{zaslavsky2019thesis} notes that \eqref{eq:Talis-C_X-condition-for-bif} is closely related to the bifurcation analysis of \cite{gedeon2012mathematical}. 
The latter provides a condition to identify the critical $\beta$ values of IB bifurcations, given in their Theorem 5.3. 
Indeed, their condition is equivalent to \eqref{eq:Talis-C_X-condition-for-bif}, and therefore it also characterizes cluster-merging bifurcations. 
To see this, the necessary condition they give for a phase transition at $\beta$ is that $\nicefrac{1}{\beta}$ must be an eigenvalue of a matrix $V$ (Equation (21) there). When written in our notation, this matrix is given by
\begin{equation}			\label{eq:V-matrix-from-GPD12}
	V(\xhat; \beta)_{\x, \xp} := 
	\sum_{\y} \frac{p(\xp, \y) p(\x, \y) p_\beta(\xhat | \x) }{ p_\beta(\y, \xhat) p(\xp) } \;.
\end{equation}
However, $V$ \eqref{eq:V-matrix-from-GPD12} is readily seen to be the transpose of $C_X$ \eqref{eq:Cxx'_def}, and so they have the same eigenvalues. 

\medskip 
\section{Lyapunov-stability of an optimal IB root}

\label{sec:optimal-IB-root-is-Lyapunov-stable-in-decreasing-beta}

We provide the essential parts of a proof that an optimal IB root is Lyapunov uniformly asymptotically stable on closed intervals which do not contain a bifurcation when following the flow dictated by the IB's ODE \eqref{eq:IB-beta-ODE-in-decoder-coords} in \textit{decreasing} $\beta$. 
Definitions for the below are as in \citep{slotine1991applied} (see especially Section 4.2 there). 
See Subsection \ref{sub:IBRT1-discussion} for a discussion of the results below.

\medskip 
Let $\bm{p}^*(\beta)$ be an optimal IB root. 
We start by rewriting it as an equilibrium of a non-autonomous ODE, as in \cite[Equation (4.1)]{slotine1991applied}. 
Consider the implicit ODE \eqref{eq:implicit-beta-ODE} $\tfrac{d\bm{p}}{d\beta} = - (D_{\bm{p}} F)^{-1} D_\beta F$, specialized to the IB by setting $F := Id - BA_\beta$ \eqref{eq:IB-operator-def}. 
Denote $\delta \bm{p} := \bm{p} - \bm{p}^*$, for an arbitrary $\bm{p}$. 
Subtracting the ODE at $\bm{p}$ from that at $\bm{p}^*$ yields a non-autonomous ODE in the error $\delta\bm{p}$ from the optimal root,
\begin{equation}		\label{eq:IB-ODE-for-error-from-optimal}
	\tfrac{d\delta\bm{p}}{d\beta} = 
	(D_{\bm{p}} F)^{-1} D_\beta F\rvert_{\bm{p}^*} - (D_{\bm{p}} F)^{-1} D_\beta F\rvert_{\bm{p}^* + \delta\bm{p}}
\end{equation}
This rewrites the given root $\bm{p}^*$ as an equilibrium $\delta\bm{p} = \bm{0}$ of this ODE \eqref{eq:IB-ODE-for-error-from-optimal}, simplifying the below. 

Next, we define a Lyapunov function for the flow of the equilibrium $\delta\bm{p} = \bm{0}$ along the ODE \eqref{eq:IB-ODE-for-error-from-optimal}, when its dynamics in $\beta$ is reversed. 
Consider the IB's Lagrangian $\mathcal{L}_\beta := I(X; \hat{X}) - \beta \cdot I(Y; \hat{X})$ as a functional in $\bm{p}$, and let $\mathcal{L}_\beta^* := \mathcal{L}_\beta[\bm{p}^*]$ be its optimal value at $\beta$. Then,
\begin{equation}		\label{eq:IB-free-energy-diff}
	\left(\mathcal{L}_\beta - \mathcal{L}_\beta^*\right)(\delta\bm{p})
\end{equation}
is the desired Lyapunov function. 
Specifically, (i) $\mathcal{L}_\beta - \mathcal{L}_\beta^*$ is positive definite and (ii) $\dbeta{}\left(\mathcal{L}_\beta - \mathcal{L}_\beta^*\right)$ is negative definite, when the dynamics in $\beta$ are reversed. 
Theorem 4.1 in \citep{slotine1991applied} then implies that $\delta\bm{p} = \bm{0}$ is uniformly asymptotically stable, \cite[Definition 4.6]{slotine1991applied}. 

For (i), $\mathcal{L}_\beta - \mathcal{L}_\beta^*$ \eqref{eq:IB-free-energy-diff} is immediately seen to be positive semi-definite from the definition of $\mathcal{L}_\beta^*$, up to technicalities ignored here\footnote{ cf., \cite[Definition 4.7]{slotine1991applied}.}. 
The results of Subsection \ref{sub:discontinuous-IB-bifs} (after Proposition \ref{prop:bif-is-detectable-only-on-enough-clusters}) imply that representing $\bm{p}$ in reduced log-decoder coordinates renders \eqref{eq:IB-free-energy-diff} strictly positive definite. 
Indeed, $D(Id - BA_\beta)$ is non-singular in a reduced representation in these coordinates, as mentioned there, and so an optimal root $\bm{p}^*$ is locally unique. 
As for condition (ii), from the definition of $\mathcal{L}_\beta$ we have
\begin{equation}		\label{eq:beta-deriv-of-IB-functional}
	\tfrac{d}{d\beta}\mathcal{L}_\beta = 
	\tfrac{d}{d\beta} I(X; \hat{X}) - \beta \tfrac{d}{d\beta} I(Y; \hat{X}) - I(Y; \hat{X}) =
	- I(Y; \hat{X}) \;,
\end{equation}
where $\tfrac{d}{d\beta} I(X; \hat{X}) = \beta \tfrac{d}{d\beta} I(Y; \hat{X})$ in the last equality follows by direct calculations similar to those in the Appendix of \cite[Part III]{zaslavsky2019thesis}. 
Thus, for the $\beta$-derivative of \eqref{eq:IB-free-energy-diff} we have
\begin{equation}		\label{eq:beta-deriv-of-IB-free-energy}
	\tfrac{d}{d\beta} \left(\mathcal{L}_\beta - \mathcal{L}_\beta^*\right)(\delta\bm{p}) =
	I(Y; \hat{X})\rvert_{\bm{p}^*} - I(Y; \hat{X})\rvert_{\bm{p}} \;. 
\end{equation}
The latter is always positive semi-definite around $\bm{p}^*$, since by definition \eqref{eq:IB-curve-def} $\bm{p}^*$ yields the maximal $Y$-information subject to a constraint on the $X$-information. 
The same argument as above shows that it is strictly positive definite. 
Finally, reversing the dynamics in $\beta$ leaves the ODE \eqref{eq:IB-ODE-for-error-from-optimal} unaffected but flips the sign of \eqref{eq:beta-deriv-of-IB-free-energy}, rendering it negative definite as required.

\medskip 
\section{Introducing degeneracies cannot increase the nullity of the IB operator in decoder coordinates}
\label{sec:degeneracies-of-IB-operator}

We show that evaluating the kernel of the IB operator on a degenerate representation cannot increase its nullity rank.

\medskip 
Let $\bm{p} \in \Delta\left[\Delta[\mathcal{Y}]\right]$ be an IB root of effective cardinality $T_1$. 
A $T$-clustered representation of a root (e.g., in decoder coordinates) is a function $\pi: \Delta\left[\Delta[\mathcal{Y}]\right] \to \bb{R}^{(|\mathcal{Y}| + 1)\cdot T}$, defined on some neighborhood of the root. 
In the other way around, one can consider the inclusion $i: \bb{R}^{(|\mathcal{Y}| + 1) \cdot T} \rightarrow \Delta\left[\Delta[\mathcal{Y}]\right]$, defined on normalized decoder coordinates in the obvious way. 
Let $\pi$ be a representation of $\bm{p}$ in its effective cardinality $T_1$, and $\tilde{\pi}$ a degenerate one on $T_2 > T_1$ clusters.  
These satisfy
\begin{equation}			\label{eq:reduce-a-roots-representations}
	\pi = reduc \circ \tilde{\pi}
\end{equation}
where $reduc$ is the reduction map\footnote{ Defined similar to the root-reduction Algorithm \ref{algo:root-reduction}, by setting its thresholds to zero, $\delta_1 = \delta_2 = 0$, and replacing its strict inequalities with non-strict ones. Note that Algorithm \ref{algo:root-reduction} has a well-defined output for every input.}. 
In the other way around, one can pick a particular degenerating map $degen$ (e.g., ``split the third cluster to two copies of probability ratio 1:2''). Applying a particular degeneracy and then reducing is the identity,
\begin{equation}			\label{eq:reduc-is-left-inv-of-degen}
	reduc \circ degen = Id \;,
\end{equation}
though not the other way around. 
Let $i$ and $\tilde{i}$ be the inclusions corresponding to $\pi$ and $\tilde{\pi}$ respectively. 
Similar to \eqref{eq:reduce-a-roots-representations}, introducing degeneracy to a root has no effect before including it in $\Delta\left[\Delta[\mathcal{Y}]\right]$,
\begin{equation}			\label{eq:degen-inclusion}
	i = \tilde{i} \circ degen
\end{equation}
Recall from Subsection \ref{sub:IB-as-an-RD-problem-and-non-singularity-conj} (before Conjecture \ref{conj:BA-IB-Jacob-in-decoder-coords-is-nonsingular-at-reduced-root}) that $BA_\beta$ in decoder coordinates may be considered as an operator on $\Delta\left[\Delta[\mathcal{Y}]\right]$. 
To summarize, we have the following diagram,
\begin{equation}			\label{eq:reductions-at-bif-in-dec-coords}
	\xymatrix@C=4em@R=.5em{
		\bb{R} \ar[rr]^(.43){\bm{p}} &&
		\Delta\left[\Delta[\mathcal{Y}]\right] \ar@(ul, ur)^{BA_\beta} \ar@(r,l)[drr]^{\tilde{\pi}} \ar@(dr,l)[ddddrr]^{\pi} && \\
		&& && \bb{R}^{(|\mathcal{Y}| + 1) \cdot T_2} \ar@<5pt>[ddd]^(.45){reduc} \ar@(u,ur)[llu]_{\tilde{i}} \\ \\ \\
		&& && \bb{R}^{(|\mathcal{Y}| + 1) \cdot T_1} \ar@(dl,dl)[uuuull]^i \ar@<5pt>[uuu]^(.55){degen}
	}
\end{equation}

Next, consider the representations of the IB operator $Id - BA_\beta$ \eqref{eq:IB-operator-def} on $T_1$ and $T_2$ clusters. 
These amount to pre-composing with the inclusions and post-composing with the representation maps. 
Denote by $Id_i$ the identity operator on $\bb{R}^{(|\mathcal{Y}| + 1) \cdot T_i}$. 
By identities \eqref{eq:reduce-a-roots-representations}, \eqref{eq:reduc-is-left-inv-of-degen} and \eqref{eq:degen-inclusion}, we have
\begin{multline}
	Id_1 - \pi\circ BA_\beta \circ i =
	reduc \circ degen - reduc \circ \tilde{\pi} \circ BA_\beta \circ \tilde{i} \circ degen \\ = 
	reduc \circ \left[ Id_2 - \tilde{\pi} \circ BA_\beta \circ \tilde{i} \right] \circ degen
\end{multline}
Differentiating, by the chain rule we have
\begin{equation}
	D\left( Id_1 - \pi\circ BA_\beta \circ i \right) =
	D\left(reduc\right) D\left(Id_2 - \tilde{\pi} \circ BA_\beta \circ \tilde{i} \right) D\left(degen\right) \;.
\end{equation}
Multiplying matrices can only enlarge the kernel, $\dim \ker AB \geq \dim \ker A$, and so
\begin{equation}
	\dim \ker D\left( Id_1 - \pi\circ BA_\beta \circ i \right) \geq
	\dim \ker D\left(Id_2 - \tilde{\pi} \circ BA_\beta \circ \tilde{i} \right) 
\end{equation}
Thus, introducing degeneracies to the IB operator in decoder coordinates cannot increase its nullity rank.

\newpage
\bibliographystyle{plain}
\bibliography{my_bib}

\end{document}